\documentclass[12pt]{article}

\usepackage{newtxtext,newtxmath}

\usepackage{graphicx}

\usepackage[letterpaper,margin=1in]{geometry}

\renewenvironment{abstract}
	{\quotation}
	{\endquotation}

\date{}

\makeatletter
\renewcommand{\fnum@figure}{\textbf{Figure \thefigure}}
\renewcommand{\fnum@table}{\textbf{Table \thetable}}
\makeatother

\usepackage{scicite}

\usepackage{url}

\newtheorem{prop}{Proposition}[section]

\newtheorem{lemma}[prop]{Lemma}
\newtheorem{proposition}[prop]{Proposition}

\newenvironment{proof}{\par\noindent\textit{Proof.}\ }{\hfill\(\square\)\par}

\makeatletter
\newenvironment{fixedtable}{\par\medskip\noindent\begin{minipage}{\linewidth}\def\@captype{table}}{\end{minipage}\par\medskip}
\makeatother

\newcommand{\ind}{\mathbf{1}}
\newcommand{\indep}{\mathbin{\perp\!\!\!\!\!\:\perp}}

\newcommand{\simiid}{\stackrel{iid}{\sim}}
\newcommand{\simind}{\stackrel{ind}{\sim}} 
\newcommand{\E}{\mathbb{E}}
\newcommand{\Var}{\mathbb{V}}
\newcommand{\Cov}{\text{Cov}}
\renewcommand{\Pr}{\mathop{\mathbb{P}}\nolimits}

\usepackage{caption}
\usepackage{float}

\newcommand{\mgvec}{\textbf{m}_g}
\newcommand{\msampsize}{m}
\newcommand{\mgs}{m_{gs}}
\newcommand{\mg}{m_{g \cdot }}
\newcommand{\ms}{m_{\cdot s}}

\newcommand{\alphavec}{\boldsymbol{\alpha}}
\newcommand{\thetavec}{\boldsymbol{\theta}}
\newcommand{\thetagvec}{\boldsymbol{\theta}_g}
\newcommand{\Data}{\boldsymbol{\mathcal{D}}}
\newcommand{\gammavec}{\boldsymbol{\gamma}}

\newcommand{\popoverall}{N}
\newcommand{\popoverallr}{N_{\cdot r}}
\newcommand{\popoverallg}{N_{g\cdot}}
\newcommand{\propoverallgprime}{N_{g'\cdot}}
\newcommand{\popoverallgs}{N_{gs\cdot}}
\newcommand{\popoverallgr}{N_{g\cdot r}}
\newcommand{\popoverallgsr}{N_{gsr}}

\newcommand{\popoverallstar}{N^*}
\newcommand{\popoverallgstar}{N^*_{g\cdot}}
\newcommand{\popoverallgsstar}{N^*_{gs\cdot}}

\newcommand{\samplednum}{n}
\newcommand{\sampledg}{n_{g\cdot}}
\newcommand{\sampledgr}{n_{g\cdot r}}
\newcommand{\sampledgsr}{n_{gsr}}

\newcommand{\sampledrareg}{n_{g\cdot 1}}
\newcommand{\sampledrare}{n_{\cdot 1}}
\newcommand{\sampledraresg}{n_{gs1}}

\newcommand{\cS}{\mathit{S}}
\newcommand{\cG}{\mathit{G}}

\newcommand{\targetoverall}{\tilde{n}}
\newcommand{\targetg}{\tilde{n}_g}
\newcommand{\targetgi}{\tilde{n}_{G_i}}

\newcommand{\probnormg}{\pi(g)}
\newcommand{\probnormgstar}{\pi^*(g)}
\newcommand{\probnormgi}{\pi(G_i)}
\newcommand{\probnormgistar}{\pi^*(G_i)}

\def\scititle{
	Improving Minority Population Sampling with BISG Probabilities: Evidence from a Survey of Jewish Americans
}
\title{\bfseries \boldmath \scititle}

\author{
	Kyla~Chasalow$^{1}$,
	Eitan~Hersh$^{2}$,
	Kosuke~Imai$^{1,3\ast}$,
    Laura~Royden$^{3}$\and 
	\small$^{1}$Department of Statistics, Harvard University, Cambridge, USA.\and\
	\small$^{2}$Department of Political Science, Tufts University, Medford, USA.\and
	\small$^{3}$Department of Government, Harvard University, Cambridge, USA.\and
	\small$^\ast$Corresponding author. Email: imai@harvard.edu\and
}

\begin{document}

\maketitle



\begin{abstract} 
Sampling geographically dispersed minority populations poses substantial challenges when individual group membership cannot be directly observed. Although stratified sampling can offer efficiency gains, these gains are typically modest unless the minority population is highly concentrated within a small number of strata. In this paper, we propose using Bayesian Improved Surname Geocoding (BISG) to enhance the efficiency of minority population sampling. BISG generates individual-level probabilities of minority group membership based on names and residential addresses. We incorporate these probabilities into a stratified Poisson probability sampling design. Applying the proposed approach to a national survey of Jewish Americans, we find that our estimates closely align with those from a large-scale Pew Research Center survey of the same population, which relied on a substantially more expensive sampling strategy involving geographic stratification and screening. At a fraction of the cost, our survey reproduces nearly identical patterns observed by Pew, including estimates of religious denomination prevalence and participation in specific religious activities.
\end{abstract}

\section{Introduction}

Sampling minority populations---defined by characteristics such as race, ethnicity, or religion---is essential for understanding their attitudes, experiences, and behaviors, which might otherwise be overlooked or misunderstood. While some minority populations are difficult to study because of reluctance to participate, others pose challenges primarily because they are small and lack a readily available sampling frame that directly identifies group membership \cite{LEPKOWSKI1991416, 2014h2s}. As a result, obtaining samples of adequate size and quality through standard probability sampling methods is akin to finding a few needles in a haystack, often requiring large initial samples and extensive screening efforts that incur substantial cost \cite{kalton2009methods}. Without efficient sampling strategies, researchers and organizations seeking to understand or serve these populations face barriers to producing reliable information.

In this paper, we develop an efficient sampling method for settings in which the minority population of interest is present in a sampling frame but cannot be individually identified. Our approach leverages Bayesian Improved Surname Geocoding (BISG), a widely used technique for estimating the probability that an individual belongs to a minority group based on their surname and geographic location \cite{elliott2009using, imai:khan:16, Imai_etal_2022, rosenman2023race, mccartan2025estimating, greengard2026calibrated, DasanaikeImai2026}.
We show that these BISG probabilities can be incorporated into a stratified Poisson probability sampling design, yielding substantial efficiency gains for sampling minority populations.  Importantly, the proposed methodology remains valid even if BISG probabilities are biased. We also discuss how and under what assumptions BISG probabilities can be used to estimate minority-group-specific distributions of variables in the sampling frame \cite{Chen_etal2019, mccartan2025estimating}, and how these estimates can inform non-response adjustments.   

We apply our methodology to efficiently sample from the U.S. registered voter adult Jewish population, which has been of interest to both researchers \cite{waldmartinez01} and Jewish organizations \cite{aronson2016all, cohen2016deficient}.  Several challenges make sampling from this population particularly difficult.  First, the Jewish population constitutes only about 2.4\% of American adults \cite{pew_jewish_population_2021}. Second, no government agency identifies households by Jewish ethnicity or religion. Unlike for some racial and ethnic groups (e.g., Black, Hispanic, and Asian populations), the Census Bureau does not publish data on the geographic distribution or surname frequencies of Jewish Americans. This makes it impossible to compute BISG probabilities from Census data alone.  Third, although Jewish Americans are more geographically concentrated in certain states and cities than in others \cite{saxe21}, the population remains relatively small and spatially dispersed even within those areas. As a result, stratified sampling based on geography alone offers only limited gains in efficiently identifying Jewish respondents.

Our approach combines existing estimates of the geographic distribution of the minority group with name data from the U.S. voter file and a novel source of Jewish name data from Jewish obituaries. Using these data, which are described in Section~\ref{sec-sampling-Jewish}, we estimate the probability that each individual in the U.S. voter file is Jewish. In Section~\ref{sec-sampling}, we propose a general sampling methodology that incorporates BISG probabilities into Poisson sampling. Standard approaches to BISG probability estimation could be used here when applicable, but in Section~\ref{sec-est-bisg}, we present a method for estimating the BISG probabilities when geo-located surnames are available as in our obituary data. We emphasize that although cleaner data and more realistic assumptions will yield better BISG estimates and more efficient sampling, using a probability sampling procedure guarantees the unbiasedness of downstream estimation even when BISG probabilities are biased -- or at least, it does so assuming no non-response. In Section~\ref{sec-estimation}, we briefly review how to analyze the final survey and discuss methods for lessening the impact of non-response bias in this setting. Finally, in Section~\ref{sec-estquality}, we discuss how the precision and accuracy of BISG probabilities affect the method's performance and how to assess their quality before and after running the survey.

We validate our methodology by comparing our survey estimates with the corresponding estimates based on the 2020 national survey of Jewish Americans conducted by the Pew Research Center \cite{pewjewish2020}, which, while not perfect itself, is widely considered to be the best available survey \cite{weisberg19}. We find that our approach substantially improves the sampling efficiency of the U.S. Jewish population while producing similar estimates to Pew, including those of religious denomination prevalence and participation in religious activities. For a fraction of Pew's costs, we obtain a survey in which 57\% of respondents self-identified as Jewish. 

\subsection*{Related literature}\label{sec-related-lit}

Our approach is closely related to recent applied work on predicting sample eligibility to improve sampling efficiency  \cite{harter2016quality,  ridenhour2017propensity, harterpredictive,seeskin2022address,  McPhee2022PredictiveModeling, Jennetti_2023}. However, many of these approaches form binary predictions by thresholding probabilities so that they can be easily used as part of standard stratified sampling.  In contrast, we develop stratified Poisson sampling that directly uses BISG probabilities without arbitrarily thresholding them. 

As noted above, an extensive literature addresses the estimation of BISG probabilities \cite{elliott2009using, imai:khan:16, Imai_etal_2022, rosenman2023race, mccartan2025estimating, greengard2026calibrated, DasanaikeImai2026}. Here, we show how to estimate these probabilities when census data are inapplicable and labeled training data that cover both minority and non-minority populations are unavailable. We also show how to apply the downstream estimation methods of \cite{Chen_etal2019,mccartan2025estimating} to estimate minority-group-specific distributions of variables in the original sampling frame. These estimates can then inform adjustments to reduce non-response bias.

The work most similar to ours is \cite{elliott2013using}, which samples Hispanic couples in Los Angeles using probabilities proportionate to census-based BISG probabilities. However, because their sampling is limited to a few zip codes, they do not combine this with additional geographic allocation as we do when sampling the entire country. As recognized in \cite{elliott2013using}, the efficiency gain of their method is limited by their sampling frame already having a high concentration of Hispanics. Our work builds on their results to show that BISG-based sampling can also yield major efficiency gains when sampling a much rarer population from a larger sampling frame. We also add guidance on expected yield, the quality of BISG estimation, the role of surname filtering when targeting a minority population, and adjustments for nonresponse bias when target distributions are not directly available. 

A separate methodological literature focuses on sampling minority populations with screening designs and disproportionate stratification \cite{kalton1986sampling, waksberg1997, cervantes2007, kalton2009methods, clark_etal_2009, Kalton_2014, chen2015geographic}. In screening approaches, a large pool of potential respondents is first asked a few questions to determine minority group membership; only those who self-identify as belonging to the target population are then given the full survey. The primary drawback of this strategy is its high cost, as extensive screening is often required to identify a relatively small number of eligible respondents. 

Disproportionate stratification oversamples strata (often geographic) known to contain higher concentrations of minority group members, possibly while also accounting for differences in screening costs across strata. However, this approach typically yields modest efficiency gains unless the minority population is both highly concentrated in a small number of strata {\it and} constitutes a substantial share of those strata. Our approach can incorporate both screening and geographic targeting, but uses more individualized probabilities to achieve greater efficiency gains.

The potential of using surnames for sampling minority populations has long been recognized \cite{lauderdale2000asian, kalton2009methods} and applied in a variety of contexts \cite{HueyHuey1990, cervantes2007, kim2013surname, schnell2013new}, including studies of Jewish Americans \cite{himmelfarb83, kosmin1989, elia24}. However, these applications have typically relied on relatively small lists of highly distinctive surnames to restrict the sampling frame, a strategy that can introduce selection bias. In contrast, we use data derived from obituaries that reference Jewish funeral homes or burial sites, capturing a substantially broader set of both first and last names and thereby mitigating some of the limitations associated with reliance on small lists of highly distinctive surnames.

In its study of Jewish Americans, the Pew Research Center has conducted the most comprehensive surveys to date, notably in 2013 and 2020. Pew's methodology relies on disproportionate stratification. Specifically, Pew identifies geographic areas in which Jewish Americans are likely to be concentrated using Brandeis University’s American Jewish Population Project (AJPP) \cite{saxe21} along with additional data sources. Pew then conducts extensive screening within these areas to identify individuals who self-identify as ethnically, culturally, or religiously Jewish.

As noted above, the primary limitation of this approach is its high cost. In 2020, Pew screened more than 300,000 individuals, initially paying respondents \$2 to complete screening questions and subsequently offering Jewish-identifying respondents incentives ranging from \$10 to \$50 to participate in the main survey. This process led to the final sample of 4{,}718 respondents who identified as Jewish \cite{pewjewish2020}. Although Pew's studies serve as benchmarks that are highly valued by researchers and organizations seeking to understand the Jewish American public, the substantial resources required make these methods prohibitively expensive for other researchers and organizations. Our goal is to develop a more efficient sampling approach.

Overall, this paper develops a general framework for using BISG probabilities -- regardless of how they are estimated -- to sample minority populations and analyze the resulting data. It also introduces an approach to obtaining BISG probabilities that may be relevant for sampling both the Jewish and other minority populations with distinctive burial practices or institutions. Although Poisson sampling, BISG estimation, stratification, and rare population sampling methods have well-established literatures, we are not aware of prior work that integrates them into a unified framework for sampling rare and geographically dispersed populations.

\section{Materials and Methods}
\subsection{Data Sources}
\label{sec-sampling-Jewish}

Our proposed methodology estimates BISG probabilities by combining multiple data sources in settings where a single labeled dataset distinguishing minority and non-minority populations is unavailable. Our approach requires three types of data: 1. the name distribution of the minority population, preferably disaggregated by geography; 2. the geographic distribution of the minority population; and 3. the name distribution of the larger population to sample from. Below, we describe how we assembled these data sources for our empirical application.

\subsubsection{Name distribution through obituaries}

For some minority populations, the distribution of names is readily available.  For example, in the United States, the Census Bureau makes the distribution of common surnames available for several racial and ethnic groups, including Whites, Blacks, Hispanics, and Asians.  Other common data sources include voter files \cite{rosenman2023race} and the Loan Application Registers \cite{tzioumis2018demographic}.  However, for Jewish Americans, such data cannot be easily obtained.

We overcome this difficulty by collecting obituary data from Jewish funeral homes. 
Death rituals vary by religious and cultural communities \cite{walter05}, and Jewish Americans have long used distinctive funeral homes and burial grounds to accommodate specific religious practices \cite{amanik17}. While not all Jewish American families use Jewish-specific funeral homes, funeral homes still segment the population into a group likely to be Jewish and a group that contains a mix of Jewish and non-Jewish people. Obituaries are publicly accessible documents that usually include information about names and birth dates. Importantly, for our purposes, the names of funeral homes are typically also listed in publicly accessible obituaries. This allowed us to identify obituaries from Jewish funeral homes by assembling a list of Jewish funeral homes across the United States and finding all obituaries linked to them.  This provides information about which names are rare or common in at least a broad subset of Jewish Americans. 

After manually assembling a list of 107 Jewish funeral homes across the United States (see Supplement~\ref{supplement-data-detail} for details), we contracted with PBI Research Services, a commercial provider that compiles obituaries and death notices across the country from ``over 40,000 sources,'' including government sources and online publication \cite{pbi_industries}. We queried the PBI database for obituaries of individuals who died between January~1, 2000 and December~31, 2023 and in which one of the Jewish funeral homes was listed. This yielded 232,585 obituaries. For each record, we extracted the name of the deceased, state of residence (at time of death), birth date (when available), death date, and the full biographical text of the obituary. The geographic label for each surname is useful for distinguishing subgroups within the Jewish American population, such as Persian Jews in California, Russian Jews in New York, and Syrian Jews in New York and New Jersey.

Our obituary data do not capture all deaths among Jewish American adults. According to estimates from Pew, in 2021 there were approximately 7.5 million Jewish Americans, of whom about 5.8 million were adults, representing 2.4\% of the U.S. population \cite{pew_jewish_population_2021}. Assuming a similar population size over the 2000--2023 period and that Jewish people die at roughly the same rate as the overall population (so that 2.4\% of the roughly 3 million annual deaths in the United States \cite{CDC_DeathsMortality_FastStats} are Jewish), a back-of-the-envelope calculation suggests that roughly 1.6 million Jewish Americans died during this time frame. Under these assumptions, our obituary data cover approximately 15\% of all Jewish deaths over this period.

There are several potential sources of incomplete coverage in our obituary data. First, not all obituaries for individuals who used Jewish funeral homes explicitly list the funeral home by name. Second, obituaries are not published online for all deaths or for all demographics at equal rates. Research on online obituaries has found that while they accurately capture some vital statistics, they can have skewed coverage by race, age, gender, education, and death type -- though this selection bias may have decreased over time as internet use became more widespread \cite{boak2008internet, BUONANNO2021535, violo_2025}. 
Such selection matters here only to the extent it correlates with names. Third, PBI’s data collection process may not capture the full universe of published obituaries. 

Finally, not all individuals who identify as Jewish use Jewish funeral homes, and this choice may be correlated with religiosity or the strength of Jewish identity. At the same time, many Jews who are not religious still choose Jewish funeral homes. For example, our obituary collection includes Jewish individuals whom their families describe as secular or engaging in practices such as eating lobster, a non-kosher food.  Moreover, because surnames are relatively stable across generations, individuals who do not themselves use Jewish funeral homes may still be represented in our data through relatives who do. In addition, although rising rates of intermarriage with non-Jewish spouses (including slightly higher rates among Jewish women) have increased the prevalence of surnames that are not historically associated with Jews, those who continue to use Jewish funeral homes remain represented in our data \cite{cooperman2013intermarry}. Indeed, our obituary data include over 50,000 surnames, many of which are not distinctively Jewish. 

The extent to which these and other selection factors bias our estimated BISG probabilities is unknown. In Supplement~\ref{sec-validation}, we present a series of validation checks to assess the plausibility of our data. We show that the state-by-state distribution of obituaries closely mirrors the known geographic distribution of Jewish Americans and that many of the surnames listed on Wikipedia as being of Jewish origin appear in the data and at relatively high frequencies. Our survey results, reported in Section~\ref{sec-empiricalresults}, also indicate that we reached Jewish people of varying denominations and degrees of religious practice. 

Lastly, our obituary data also provides a list of first names of the deceased and their surviving relatives.  Specifically, we used OpenAI's ChatGPT to extract the first names of family members listed in obituaries (see Supplement~\ref{supplement-data-detail} for details).  This yields 2,216,036 first names representing 35,965 unique names. By incorporating the first names of family members, who are typically much younger than the deceased (the mean age at death is 82), we can partially account for generational shifts in first names while also mitigating the impact of children of Jewish intermarriages bearing non-Jewish surnames \cite{Twenge_etal_2010}.

\subsubsection{Geographic distribution from Brandeis's AJPP Project}

Obtaining reliable estimates of the geographic distribution of the minority population is also challenging for some groups, including Jewish Americans. For instance, the Census Bureau provides population counts at fine geographic levels only for the major racial and ethnic groups noted above, and does not report comparable statistics for religious populations. Following prior work by Pew, we use estimates from Brandeis University’s American Jewish Population Project (AJPP) to characterize the state-level population size of Jewish Americans \cite{tighe21}.

The AJPP synthesizes data from a large number of existing surveys, in which some respondents self-identify as Jewish. Using respondents' geographic and demographic information, Brandeis produces estimates of the number of Jewish individuals and Jewish adults at both the state and county levels. It is possible to obtain some county-level location information for funeral homes and obituaries, but as discussed in Supplement~\ref{app-statelevelchoice}, a county-level surname modeling approach would create additional estimation challenges. We therefore use state-level estimates throughout our analysis. 

These data show that the Jewish population in the United States is highly geographically concentrated: approximately half of all Jewish Americans reside in just four states—New York, California, Florida, and New Jersey—while many states have Jewish population shares below 1\%. Still, even in the state with the highest prevalence, New York, Jewish Americans comprise only about 7\% of the population (see Supplement figure~\ref{fig-AJPPjewpropestimates}). Thus, despite geographic clustering, Jewish Americans remain a relatively rare population in every state.

\subsubsection{The Voter File as the target distribution}\label{sec-VF}

BISG probabilities also require information about the baseline prevalence of names, ideally by geography. For example, a first name such as Robert or a surname such as Smith being highly prevalent in our obituary data need not mean Robert is particularly indicative of being Jewish. What matters is the relative prevalence.  Even a fairly rare name such as Avi can become indicative of being Jewish if it is less rare in our data than in the general population. 

We use the most recent U.S. voter file from the vendor L2, Inc. as both our sampling frame and the source of this baseline information. L2 is a leading national non-partisan firm and the oldest organization in the United States that supplies voter data and related technology to candidates, political parties, pollsters, and consultants for use in campaigns. L2 combines public state-level lists of registered voters with other commercial data and modeling to give insights on voters. Since the population of voters is not equivalent to the population of all U.S. adults, our survey is best viewed as targeting the population of Jewish American adult registered voters. Because we lack specific estimates of the geographic distribution for this subset of the Jewish adult population, we assume that the AJPP estimates are still quite accurate for describing Jewish registered voters. Nevertheless, our methodology is not specific to voter files and would apply equally to a sampling frame augmented to include non-registered adults.

\subsection{Sampling with BISG probabilities}\label{sec-sampling}

In this section, we introduce the setup and discuss sampling methods. We first explain why standard disproportionate stratified sampling is too inefficient for sampling minority populations such as Jewish Americans. We then propose an alternative stratified Poisson sampling method based on BISG probabilities. We show how to use BISG probabilities to optimally allocate stratum-level target sample sizes and derive the expectation and variance of the number of minority population members sampled, which we call ``yield.'' Throughout, we assume that BISG probabilities are given (we discuss how to estimate them in Section~\ref{sec-est-bisg}). Finally, we present the expected success rates for sampling Jewish Americans based on different sampling methods. While we are unaware of prior work presenting BISG-based allocation formulas or minority-population yield calculations, they are direct consequences of adopting Poisson sampling for studying minority populations.

\subsubsection{Setup}

Our goal is to sample from a {\it sampling frame} that contains $N$ individuals, where surname $S_i$, geographic stratum $G_i$, and contact information are available for each individual $i$. In our application, this is the voter file. Let $R_i$ be an unobserved indicator variable where $R_i=1$ means that the individual is in the minority population and $R_i=0$ otherwise. For each stratum $g=1,\ldots,L$, there are $\popoverallgsr$ individuals with minority membership status $R_i = r$ and surname $S_i=s$ in the sampling frame such that $\popoverall = \sum_{g=1}^{L}\sum_{s=1}^{|\cS|} \sum_{r=0}^1 \popoverallgsr$ (see Supplement \ref{sec-notation} for a summary of all notation used in the paper).  The challenge is that the sampling frame does not contain $R_i$ labels.
  
We will assume that information about the geographic distribution of the minority population is available from a separate source such as the AJPP data for Jewish sampling. Throughout, we will treat both $\Pr(G = g\mid R=1)$ and $\Pr(R=1\mid G = g)$ as known population probabilities. Though in practice, these will be estimated, if $G$ is fairly low dimensional as in our application to U.S. states, they are less likely to be the primary source of estimation uncertainty compared to surname distribution estimation. 

For the surname data set, suppose that we have a separate i.i.d. sample of $\msampsize$ individuals from the minority group that includes their surnames and geographic locations $\{S_j,G_j\}_{j=1}^{\msampsize}$. We let $\ms, \mg$, and $\mgs$ be the sample counts by surname $s$, geography $g$, and both, respectively, i.e., $\ms:=\sum_{j=1}^{\msampsize} \ind\{S_j = s\}$, $\mg:=\sum_{j=1}^{\msampsize} \ind\{G_j = g\}$, and $\mgs:=\sum_{j=1}^{\msampsize}\ind\{G_j = g, S_j = s\}$.  In general, this sample is not directly linkable to the sampling frame. For example, in our application, the obituary data contains deceased people who should no longer be in the voter file. 

The survey will contain a subset of the sampling frame of size $\samplednum=\sum_{g=1}^L \sampledg$ where $\sampledg$ respondents are sampled within each stratum $g$. Let $I_i$ denote the sampling indicator. The number of individuals sampled from group $R=1$ is $\sampledrare:=\sum_{i=1}^{\popoverall}I_i\ind\{R_i=1\}$, though the number who actually respond may be lower. Similarly, $\sampledrareg:=\sum_{i=1}^{\popoverall}I_i\ind\{G_i=g, R_i=1\}$ is the number of $R=1$ sampled from stratum $g$ and $\sampledraresg$ is the number with $R=1$ sampled from stratum $g$ with surname $s$.

The survey contains a set of questions that measures $R_i$ as well as $p$ variables of interest $Y_i = (Y_{i1},...,Y_{ip})$ for each individual. Our goal is to estimate descriptive characteristics such as means or correlations of the distribution $\Pr(Y\mid R=1)$ in the minority population. We may also import $S$, $G$, and other variables from the sampling frame and include them in $Y$. For example, in our survey, we asked some of the same questions as the 2020 Pew survey and could import variables such as age and turnout from the voter file.

\subsubsection{Sampling Methods}\label{sec-samplemethod}

We next discuss sampling methods. We first review the standard disproportionate stratified sampling and then propose our Poisson sampling method based on BISG probabilities.

The standard disproportionate stratified sampling method for rare population sampling chooses target stratum sample sizes $\sampledg$ that sum to a fixed target size $n$ and minimize the approximate variance of the stratified estimator for some outcome $Y$.
Typically, the variance is assumed to be homogeneous within each stratum \cite{kalton1986sampling, UN1993} -- this is only assumed when deriving the allocation formula and not elsewhere.  The stratified estimator of the outcome mean for the minority population, i.e., $\E[Y \mid R = 1]$, is given by the following weighted average,
\begin{equation}\label{eq-post-strat}
\widehat{\mu}_{\text{strat}} = \sum_{g=1}^L \overline{Y}_{g1}\Pr(G = g\mid R=1), \ \ \ \text{where } \ \overline{Y}_{g1} := \frac{1}{\sampledrareg}\sum_{i=1}^{\popoverall}\ind\{G_i=g, I_i=1,R_i=1\}Y_i
\end{equation} 

Under the homogeneity of outcome variance and screening costs across strata, equation~\ref{eq-optimal_ng} gives the classic formula for sample size allocation, which shows that it is desirable to sample more individuals from strata with a larger number of minority members, 
\begin{equation}\label{eq-optimal_ng}
       \sampledg \ = \ \frac{n\sqrt{\Pr(R=1\mid G = g)}\popoverallg}{\sum_{g'=1}^L \sqrt{\Pr(R=1\mid G = g')}\propoverallgprime} 
\end{equation}
For completeness, we include the derivation of this well-known result in Supplement~\ref{sec-opt}. See \cite{kalton1986sampling} and \cite{UN1993} for sample size calculations that relax some of these assumptions.

In practice the efficiency gains from disproportionate stratified sampling are typically modest unless the minority population of interest is geographically highly concentrated in relatively small areas. Moreover, incorporating additional information from surnames (e.g., our obituary data) into this sampling framework is not straightforward. Simply replacing $\Pr(R=1 \mid G = g)$ in equation~\ref{eq-optimal_ng} with BISG probabilities $\Pr(R=1 \mid S = s, G = g)$ is infeasible, because the high dimensionality of the surname variable leads to strata that are too small to support effective sampling.  

An alternative that overcomes these difficulties is stratified Poisson sampling \cite{hajek1964asymptotic, TilleYves2006Sa}. This independently samples each individual with some probability $\pi_i$ based on the information about both the unit and its stratum. 
The goal of Poisson sampling is to improve upon disproportionate stratified sampling by increasing the expected number of sampled minority population members, $\sampledrare$. In particular, we seek a design under which the probability that a sampled unit from stratum $g$ belongs to the minority population exceeds the baseline stratum proportion, i.e., $\Pr(R=1\mid I=1,G=g) \geq \Pr(R=1 \mid G = g)$.

Suppose that we wish to target an overall sample size of $\targetoverall$ with stratum-specific target sizes $\targetg$ such that $\sum_{g=1}^{L}\targetg=\targetoverall$. Given the stratum-level target sample size, we set the sampling probability of each unit $\pi_i$ proportional to the BISG probability, i.e., to $\Pr(R_i=1\mid S_i,G_i)$, which may be estimated using the method discussed in Section~\ref{sec-est-bisg}.  We normalize the individual sampling probability to ensure that the expected sample size $\E[\sampledg]$ from each geography $g$ is $\targetg$, i.e.,
\begin{equation}\label{eq-sampprob}
    \pi_i = \frac{\targetgi\Pr(R_i=1\mid S_i, G_i)}{\probnormgi} \quad \quad \text{ where } \quad  \probnormg = \sum_{i:G_i=g} \Pr(R_i=1\mid S_i,G_i)
\end{equation}

Under this Poisson sampling scheme, we can show that when sampling with the true BISG probabilities, 
\begin{equation}\label{eq-poissonsuccessrate}
    \Pr(R=1\mid I=1, G = g)= \frac{\sum_{s=1}^{|\cS|} \popoverallgs\Pr(R=1\mid S = s, G = g)^2}{\sum_{s=1}^{|\cS|}\popoverallgs\Pr(R=1\mid S = s, G = g)}
\end{equation}
which is larger when $\Pr(R=1\mid S = s, G = g)$ is larger for many surnames within a given stratum (see Proposition \ref{prop-pois-counts} in Supplement~\ref{sec-opt-pois} for a formal statement). Notably, this probability, which is the expected per-stratum success rate, does not depend on the choice of $\targetg$. However, the total number of minority population members ultimately sampled does depend on $\targetg$. In particular, the sampling design should allocate more units to strata with larger minority populations or higher success probabilities. 

Next, we discuss how to allocate the target sample size across strata. Like the case of disproportionate allocation above, we minimize the variance of the stratified estimator $\widehat{\mu}_{\mathrm{strat}}$. Under some simplifying assumptions also made in the classic stratified sampling calculation, the optimal allocation is given by, 
\begin{equation}\label{eq-Tg_opt_poststrat}
    \frac{\targetg}{\targetoverall} \propto  \frac{\Pr(G=g\mid R=1)}{\sqrt{\Pr(R=1\mid I=1, G = g)}}
\end{equation}
Supplement~\ref{app-sampling} provides a formal statement (Proposition~\ref{prop-Tg-opt}) and derivation. Finally, to prevent $\pi_i$ from exceeding $1$, we require the following restriction for each $g$,
\begin{equation}\label{eq-upper}
    \targetg \leq \frac{\probnormg}{\max_{\{i: G_i = g\}}\Pr(R_i=1\mid G_i=g,S_i)}.
\end{equation}
We can easily satisfy this constraint so long as $\targetg$ is chosen in a reasonable way that reflects the  $\Pr(R_i=1\mid S_i,G_i)$ values and stratum sizes (e.g., $\targetg\leq \probnormg$). 

The expression for $\Pr(R=1\mid I=1,G=g)$ in equation~\ref{eq-poissonsuccessrate} indicates that in practice, the allocation formula in equation~\ref{eq-Tg_opt_poststrat} will rely on the estimated BISG probabilities $\widehat{\Pr}(R=1\mid S = s, G = g)$, but because it involves summing many of them, we expect the estimation of $\targetg$ to be stable. Although it may seem counterintuitive that $\Pr(R=1 \mid I=1, G=g)$ appears in the denominator of equation~\eqref{eq-Tg_opt_poststrat}, this formulation prevents strata with large (or small) BISG probabilities from being assigned excessively high (or low) sampling probabilities. See Supplement~\ref{app-sampling} for a detailed discussion.

In practice, this allocation formula under Poisson sampling gives results that are strongly correlated with those obtained under stratified sampling.
Although the former incorporates surnames, they give similar results unless the degree to which surnames are informative about within-strata minority population differs across strata. More formally, these formula give the identical allocation results if $\Pr(R=1\mid I=1, G=g)=c\Pr(R=1\mid G=g)$ for some constant scaling factor $c\in [0,1]$ over all $g$.
In our application, the two are strongly correlated.

The variances of the final per-stratum sample size and overall sample size are given by,
\begin{equation}\label{eq-sampesizevar}
\Var(\sampledg) = \targetg - \sum_{i:G_i=g}\pi_i^2,
\quad\text{and}\quad  \Var(n) =\targetoverall - \sum_{i=1}^{\popoverall}\pi_i^2
\end{equation}
In particular, a high variance of total sample size $n$ can be undesirable, as it creates uncertainty about the statistical precision and costs of the survey.  The above formulas imply that the stratum-specific standard deviation of $\targetg$ can be on the order of $\sqrt{\targetg}$ and the overall standard deviation on the order of $\sqrt{\targetoverall}$. In Supplement \ref{sample-size-variance}, we also derive the slightly more complicated variance formulas of $\sampledrareg$ and $\sampledrare$ under the Poisson sampling. 

We recommend that researchers calculate these variances beforehand and evaluate whether they are acceptable for their application. The acceptability of exceeding the target sample size depends primarily on budget or other implementation constraints. If necessary, the target can be set conservatively to account for the possibility that $\samplednum > \targetoverall$. The consequences of falling short can be assessed relative to the expected number of respondents. For example, if $\targetoverall=50,000$ and the likely response rate is 5\%, a decrease of 100 persons in the realized sample size translates into only five fewer respondents on average, which is a small loss compared with that attributable to non-response.  If sample size variance is not acceptable, see e.g., \cite{brewer1984poisson, TilleYves2006Sa} for adjustments to Poisson sampling that ensure a fixed sampling size, albeit at the cost of more complicated sampling probability calculations.

A final consideration in designing the sampling scheme is whether to first filter out from the sampling frame some surnames that are highly unlikely to belong to the minority population. As mentioned in Section~\ref{sec-related-lit}, existing surname-based sampling approaches often restrict to small lists of distinguishing surnames. Under our Poisson sampling scheme, we may use a much more conservative filtering approach to further increase sampling efficiency. We caution that filtering may introduce bias in downstream analysis if it removes minority members who are systematically different from those not removed. In Supplement~\ref{app-filtering-pois}, we discuss how to account for such filtering in the Poisson sampling probabilities and allocation calculations. To allow comparison to our proposed method, we also briefly show in Supplement~\ref{app-filtering-alone} how to adjust the stratified sampling allocation under a best-case scenario where filtering removes no minority population members.

\subsubsection{Sample size calculations for the application}\label{sec-sampsizeapp}

In our application, we filtered out surnames that are not present the obituary data, leaving about 120 million people or 56\% of the voter file in the sampling frame. We reason that because we have a large pool of surnames, we expect a much lower bias due to filtering than common approaches that use small curated lists of very distinctively Jewish surnames. For surnames in the obituary data, we estimated the BISG probabilities $\Pr(R=1\mid S=s,G=g)$ using the methods discussed in Section~\ref{sec-est-bisg} and then calculated the post-filtering versions of the allocation formulas for a target overall sample size of $\targetoverall=50,000$. 

\medskip 
\begin{table}
   \centering
   \renewcommand{\arraystretch}{1.1}
  \begin{tabular}{lrl}
\hline
\hline 
\textbf{Sampling method} & Estimate & (std. dev.)
 \\ 
\hline
    \hline 
    Proposed BISG Poisson sampling  with filtering adjustment & $\ \ \ 58.6\%$& $(0.43\%)\ \ $ \\ 
     Simple random sampling     & $1.9\ \ $ & $(0.06)$ \\
     Stratified sampling\\
    \quad by size of Jewish pop. in each state &  $2.7\ \ $ &  $(0.07)\ \ $ \\
    \quad by size of Jewish pop. in each state with filtering adjustment &  $3.6\ \ $ & $(0.08)\ \ $\\
    \hline
    \hline 
   \end{tabular}
      \caption{\textbf{Estimated proportion of Jewish respondents under different sampling methods.} Point estimates and estimated standard deviations are for the fraction of Jewish people sampled given a target sample size of $50,000$. Filtering adjustment refers to removing surnames that do not appear in the obituary data. See Supplement \ref{sec-tabledetails} for details.} 
     \label{tab-samplingsuccerates}
\end{table}
\medskip 

Table~\ref{tab-samplingsuccerates} provides the estimated proportion of Jewish respondents under different sampling methods.  Under the proposed BISG Poisson sampling with filtering, the estimated success rate exceeds 50\% with a small standard deviation.  In contrast, the stratified sampling by geography using the AJPP data according to equation~\ref{eq-optimal_ng} is much less effective even with filtering adjustment, leading to success rates of only a few percentage points. 

 We also conduct a sensitivity analysis described in Supplement \ref{sec-tabledetails} to examine whether the expected success rate of Jewish people sampled would have differed if we had not filtered out the about 94 million people with surnames not observed in the obituary data. We find that if, instead of filtering, we had set these people to have a tiny positive probability of being sampled, this does not make a difference to expected yield. However, the yield can decline if we sample people with unobserved surnames with a too-high probability. 

\begin{figure}
    \centering    
\includegraphics[width=\linewidth]{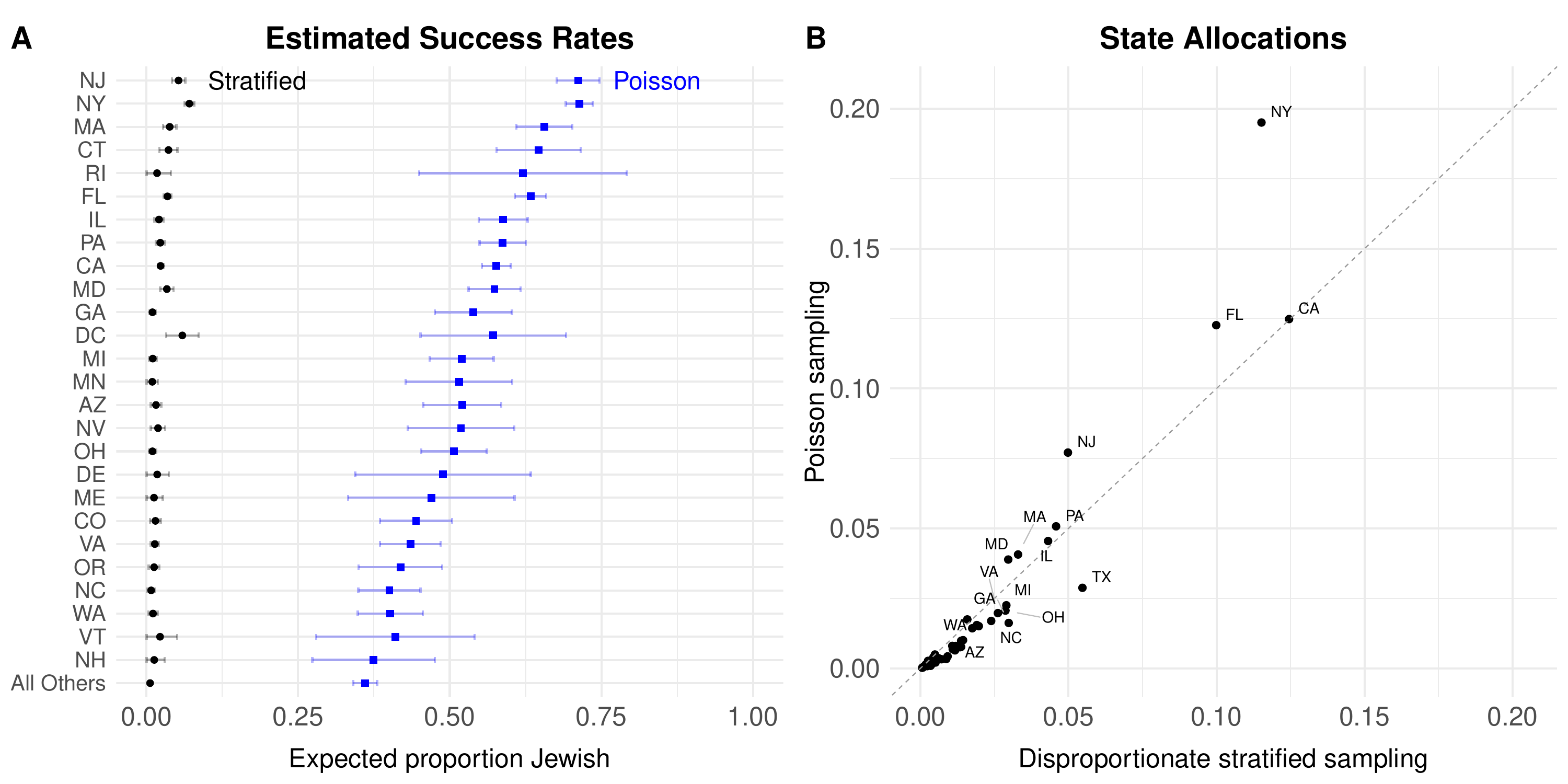} 
\caption{\textbf{Sampling allocations and expected success rates by state.} \textbf{(A)} estimates of $\sampledrareg/\sampledg$ by state with $\pm 2$SD bars. For Poisson sampling, these use our BISG estimates. See Supplement \ref{sec-tabledetails} for details. \textbf{(B)} a comparison of allocations under equation~\ref{eq-optimal_ng} and allocations under equation~\ref{eq-Tg_opt_poststrat}, both adjusting for filtering.}
    \label{fig-yield_and_allocation}
\end{figure}

Figure~\ref{fig-yield_and_allocation} shows the allocation and estimated expected gain patterns by state. In the left plot, we find that for every state, the proposed method yields a much larger gain in the expected proportion of Jewish people sampled. Although the standard errors of the success proportion are greater for the proposed Poisson sampling than for stratified sampling, our method is highly likely to provide a meaningful improvement even for states where sampling the minority population is difficult. The right plot shows that the state-level allocation formulas under Poisson and stratified sampling are similar, with the Poisson sampling favoring New York more. This suggests that gains observed in the left plot are mainly coming from better targeting within each state rather than from improved allocation across states.

Under our Poisson sampling, the standard deviation of the overall sample size $n$ was approximately $223$, while the estimated standard deviation of the total number of minority members sampled was $170$. The state-specific standard deviations for the sample size ranged from 3 (North Dakota) to 98 (New York), while the estimated standard deviation of minority members sampled per state ranged from 1 (North Dakota) to 83 (New York). We decided that this level of variability in sample size was acceptable because it was orders of magnitude smaller than the target size and would likely make little difference relative to the decrease in sample size that was likely to come from survey non-response.  Table~\ref{tab-state-response} includes the allocations across states, which as expected, favor states like New York, California, Florida, and New Jersey. Supplement~\ref{app-samplingproc-summary} gives a summary of our final sampling probabilities.

\subsection{Estimating BISG Probabilities}\label{sec-est-bisg}

In this section, we present a Bayesian hierarchical model for estimating BISG probabilities when geographically-coded surnames are available as in the Jewish obituary data. We also show how to incorporate the first name data extracted from the obituary texts. We emphasize that, depending on the available information in one's application, different methods may be used to estimate BISG probabilities. These include standard BISG estimation methods that leverage conditional independence assumptions about geography and surname when geo-located surname data are unavailable \cite{elliott2009using, imai:khan:16, Imai_etal_2022, rosenman2023race}. The power of probability sampling is that even if estimated BISG probabilities have error, valid downstream analysis is possible based on known sampling probabilities. This provides a protection against estimation error though it does not prevent non-response bias in the final survey (see Section~\ref{sec-estimation}). Inaccuracy may hurt efficiency by lowering the fraction of minority population members sampled, but it does not undermine the validity of downstream analysis so long as sampling probabilities are positive for all members of the minority population.  

\subsubsection{Learning BISG probabilities from geo-located surname frequency data}\label{sec-learning-surname-dist}

Since the joint distribution $\Pr(R,S,G)$ is typically unavailable, BISG probabilities are commonly estimated using Bayes' rule,
\begin{equation}\label{eq-bayes-rule}
   \Pr(R=1\mid S,G) = \frac{\Pr(S\mid G,R=1)\Pr(R=1\mid G)}{\Pr(S\mid G)}
\end{equation}
where $\Pr(R = 1 \mid G)$ is readily available from the external data (typically the Census data but the AJPP data from Brandeis in our case) and the denominator can be directly computed from the sampling frame.  Unfortunately, the conditional probability $\Pr(S\mid G,R=1)$ is difficult to obtain. Researchers often estimate this probability by assuming that surname and geography are independent given group membership ($S\indep G\mid R=1$). In our application, the Jewish obituary data provide an opportunity to learn geography-specific surname distributions. 

We could simply estimate the conditional probability $\Pr(S=s\mid G=g,R=1)$ using the corresponding sample proportion in the obituary data, i.e., $\mgs/\mg$.  However, these proportions are based on small counts and hence are highly variable. To mitigate this high variance problem, we use a hierarchical Bayesian model that shrinks location-specific estimates towards the overall surname distribution in a data-driven way.

Concretely, let $\mgvec := (m_{g1},...,m_{g|\mathcal{S}|})^\top$ be the observed vector of surname counts for geography unit $g$. We fit the following hierarchical Bayesian Multinomial-Dirichlet model.
\begin{align}
\mgvec &\sim \text{Multinomial}(\mg, \thetagvec), &\text{ where  }& \theta_{gs}:=\Pr(S=s\mid G=g,R=1) \text{ for } s=1,...,|\cS|\\
\thetagvec &\overset{iid}{\sim} \text{Dirichlet}(\eta\alphavec) & \text{ where  } & \alphavec \in \mathbb{R}^{|\cS|} \text{ with } \sum_{s=1}^{|\cS|}\alpha_s = 1 \text{ and } \eta > 0\\ 
\alphavec &\sim \text{Dirichlet}(\gammavec)&\text{ where } & \gammavec \in \mathbb{R}_{+}^{|\cS|}\\
\eta &\sim  \pi & \text{ where } & \pi \text{ is a prior distribution}
\end{align}
In this model, $\alphavec$ characterizes the population distribution of surnames and $\gammavec$ is a length $|\mathcal{S}|$ hyperparameter, which, in line with empirical Bayes, we set to $\gamma_s=\ms+1$. For the prior on $\eta$, we used a diffuse gamma distribution with mean $100$ and standard deviation $100$, which gave reasonable results in simulations.

There will usually be surnames in the sampling frame with no training data ($\ms=0$). One can either assign them a small positive prior probability using $\gammavec$ and thereby avoid 0 estimates \cite{Imai_etal_2022} or exclude them 
entirely so that they are assigned $\Pr(R=1\mid S=s,G=g)=0$ and filtered out in the sampling. We filtered them out. As mentioned in Section~\ref{sec-sampling}, we expect our list of surnames is expansive enough to avoid major selection bias in any minority members that this removes from the sampling frame. Fitting the hierarchical model over thousands of additional surnames would also have increased computational burden.

Under the model, the conditional posterior mean of $\theta_{gs}$ given parameters $\gammavec,\alphavec, \eta$ and data $\Data = \{\mgvec : g\in \mathcal{G}\}$ is 
\begin{align}\label{eq-conditpostmean-main}
    \E[\theta_{gs} \mid \alphavec, \eta, \Data, \gammavec] = \frac{\mgs+\eta\alpha_s}{\mg+ \eta} =  (1-\rho_g) \frac{\mgs}{\mg} + \rho_g \alpha_s, 
\end{align}
where $\rho_g = \eta/(\mg + \eta)$ controls the degree of shrinkage towards the overall surname distribution. To estimate the posterior mean of $\theta_{gs}$, we use an MCMC sampler described in  Supplement \ref{sec-sampler} and sample from the following marginal posterior of $\alphavec, \eta$, 
\begin{align}
\pi(\alphavec,\eta \mid \Data, \gammavec)&\propto  \pi(\eta)\Gamma(\eta)^L \left(\prod_{s=1}^{|\mathcal{S}|}\alpha_s^{\gamma_s -1}\right) \left(\prod_{g=1}^{L}\prod_{s=1}^{|\cS|} \prod_{k=1}^{\mgs}(\eta\alpha_s + \mgs - k)\right) \prod_{g=1}^{L} \frac{1}{\Gamma(\mg+ \eta)},
\end{align}
where $\Gamma$ denotes the gamma function. Given $B$ draws of $\alphavec^{(b)},\eta^{(b)}$ from the posterior, we estimate the unconditional posterior mean of $\theta_{gs}$ and the posterior mean of $\rho_g$ by
\begin{align}
   \widehat{\theta}_{gs} = \frac{1}{B}\sum_{b=1}^{B}\frac{\mgs + \eta^{(b)}\alpha_s^{(b)}}{\mg+ \eta^{(b)}}, \hspace{1cm} \widehat{\rho}_g = \frac{1}{B}\sum_{b=1}^{B}\frac{\eta^{(b)}}{\mg + \eta^{(b)}}.
\end{align}

In our application, we fit this model to the 49,198 surnames and the 232,585 obituaries in the Jewish obituary data, which are geolocated across 50 states and Washington, DC. We ran the sampler for 45,000 iterations and discarded the first 15,000 as burn-in before calculating posterior mean estimates for each $\theta_{gs}$ and $\rho_g$. Figure~\ref{fig-rho-postmean} gives the posterior mean of $\rho_g$ for each state. It shows that for states with many observations, the model relies strongly on the per-state distribution, while for states with fewer observations, it relies increasingly on the overall distribution.

\begin{figure}[h]
    \centering
    \includegraphics[width=\linewidth]{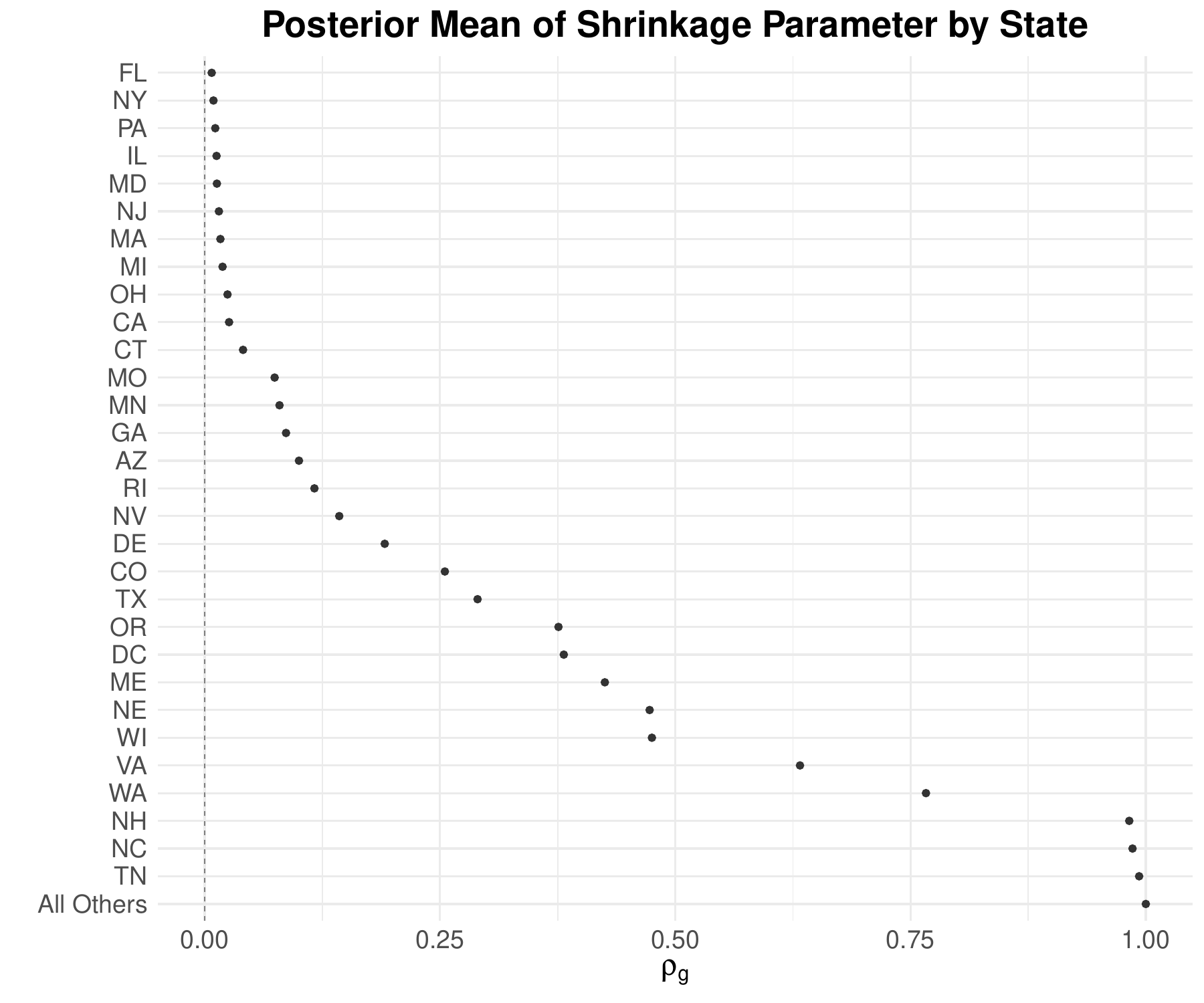}
    \caption{\textbf{Posterior means of $\rho_g$ for each state.} The ``All Others'' category represents states with no observations in the obituary data. For these, the model relies entirely on the overall surname distribution.}
    \label{fig-rho-postmean}
\end{figure}

In supplement Section~\ref{sec-sampler}, we provide further theoretical justification and details on the sampler, including how to deal with the issue that estimation errors in estimating different components of Bayes' rule from different datasets can sometimes lead to probability estimates larger than $1$ and how to the model could be extended to capture the idea that a rare surname (e.g., Berlowitz) being similar to other common surnames in the obituary data (e.g., Berkowitz) is more indicative of Jewish identity. Supplement~\ref{sec-sim} provides a simulation study, evaluating the performance of the proposed method under different values of key parameters. Supplement~\ref{sec-surname_dist_results} gives further details on fitting the model in our application. 

\subsubsection{Incorporating first names}\label{sec-first names}

First names can also carry signals for membership in some populations and can be incorporated into BISG probabilities if first name data are available \cite{voicu2018using, emery2013name}. For example, in our application, a first name like the Hebrew-origin name Eitan provides a strong positive signal. 

For the Bayes' rule based approach to learning BISG probabilities, adding first names gives
\begin{equation}\label{eq-bayes-rule2}
   \Pr(R=1 \mid F,S,G) = \frac{\Pr(F\mid S,G,R=1)\Pr(S\mid G,R=1)\Pr(R=1\mid G)}{\Pr(F\mid S,G)\Pr(S\mid G)} 
\end{equation}

\noindent In general, using this for estimation requires simplifying assumptions. A common practice is to assume conditional independence between first name and surname within each racial category, i.e., $F\indep S,G\mid R$.
This leads to the following multiplicative factor that can be used in combination with the standard BISG estimate, \begin{equation}
    \Pr(R=1 \mid F,S,G) = \frac{\Pr(F=f\mid R=1)}{\Pr(F=f)} \Pr(R=1 \mid S,G)
\end{equation}
\noindent Although there likely is some joint relationship between first and last names, the independence may have practical advantages for extrapolating to new first and last name combinations not observed in the obituary data (e.g., living relatives). In our application, we used this formula based on the combined dataset of first names of the deceased and first names extracted from the obituary text. In our validation in Supplement~\ref{sec-validation}, we find that the first name frequency ratios recover meaningful signals of having a Jewish first name.

When available, middle names could be incorporated similarly. They may provide additional information if they 
honor family members or signal heritage alongside a more Americanized first name  \cite{Edwards_etal_2008, WuEllenDionne1999, gerhards2009hasan,
rosenman2023race}. However, prior work on BISG probabilities of race has found only incremental gains from including middle names \cite{Imai_etal_2022}. We omitted middle names in our application because they are often missing or recorded only as initials. Our data also contain few middle names for younger generations as obituaries often omit them when listing relatives.

\subsection{Estimation and Survey Nonresponse}\label{sec-estimation}

Before presenting our empirical results, we briefly discuss how to estimate quantities of interest from a survey conducted under the proposed sampling methodology.  Importantly, the final survey should include a question that measures the minority group membership $R_i$, along with the variables of interest $Y_i$.
In general, because we have a probability sample, standard design-based estimators are available for estimation using the responses from units who belong to the minority group $R_i=1$ \cite{Lohr1998}. For example, the H\'ajek Inverse Probability Weighting (IPW) estimator of the mean in the minority population $\mu := \E[Y\mid R=1]$ is  
\begin{equation}\label{eq-hajek}
 \widehat{\mu}_{\text{H}} := \left(\sum_{i=1}^{\popoverall}\frac{Y_iI_iR_i}{\pi_i}\right)\left(\sum_{i=1}^{\popoverall}\frac{I_iR_i}{\pi_i}\right)^{-1}
\end{equation}
where $\pi_i$ is the sampling probability and $I_i$ denotes the indicator for unit $i$ being included in the final sample. If there were no non-response, this would be a consistent estimator of $\mu$ whose variance depends on the yield rate of minority population members, the inherent variability of $Y$, and the variability of the sampling weights. With non-response, however, weighting by $\pi_i$ may be insufficient to ensure unbiased results.

We now briefly discuss the issue of survey non-response in the minority population setting but also refer readers to the extensive literature on the topic [See, for example, \cite{Groves1998,Lohr1998, RubinLittle2002}].
First, when sampling a minority population, one concern is that the screening process, which identifies whether a respondent is a member of the minority group, could induce bias.
In particular, knowing that the survey is specifically for the minority group may make some potential respondents less likely to participate. This problem may unavoidable given that the principle of informed consent generally requires disclosing to participants the purpose of a study \cite{belmont1979}. 

To minimize such screening bias, we designed our survey of Jewish Americans such that the initial disclosure was about the contents of survey questions (e.g., civic, political, and religious attitudes).  After the survey was completed, we informed respondents about the fact that the target population of the survey was Jewish Americans and gave respondents an option to opt out. See Section~\ref{sec-empiricalresults} and Supplement~\ref{sec-designimplementation} for further discussion of the survey methodology and response. 

Popular methods for addressing non-response bias related to observable variables are post-stratification and raking, each of which may be combined with IPW weights $1/\pi_i$. Both rely on some variables $X=(X_1,...,X_d)$ measured in the survey or sampling frame, for which external information about their distribution in the minority population is available. At a minimum, we can use the geographic information $G$. In post-stratification, per-group H\'ajek estimators of $\E[Y\mid R=1, X=x]$ are combined using known base rates $\Pr(X=x\mid R=1)$ \cite{holt1979post, Little01091993, Zhang01082000, Smith1991}. However, if $X$ contains multiple variables, cross-tabulating the survey data can result in small or empty cells. Raking, a widely-used alternative, instead iteratively finds weights to simultaneously match the weighted distribution of each $X_j$ in the survey to the target marginal distribution $\Pr(X_j\mid R=1)$ as closely as possible \cite{DemingStephan1940, Deville01091993, kalton2003weighting}. 
\cite{mercer2018weighting} have found that it can be advantageous to initialize the raking algorithm weights to the IPW weights as a way to prioritize solutions that still align somewhat with the IPW weighting.

However, a potential problem in minority population sampling is that target distributions of $X$ will often not be directly available and may differ from that of the general population. When $X$ is available in the sampling frame, estimated BISG probabilities provide two ways to overcome this challenge. First, if there are $X$, which researchers believe are \textit{conditionally} invariant given surname and geography (i.e., $X\indep R\mid S,G$), then an asymptotically unbiased estimator of $\Pr(X\mid G=g, R=1)$ is given by the following weighted average \cite{Chen_etal2019, mccartan2025estimating},
\begin{equation}
   \widehat{\Pr}(X=x\mid G=g, R=1)=   \frac{ \sum_{i:G_i=g} \Pr(R_i=1\mid S_i,G_i)\ind(X_i=x)}{\sum_{i:G_i=g}\Pr(R_i=1\mid S_i,G_i)}.\label{eq-weightedest}
\end{equation}
We can also aggregate these over geography using the Law of Total Probability and geographic base rates, i.e., $\Pr(X=x\mid R=1) = \sum_{g}\Pr(X=x\mid G=g,R=1)\Pr(G=g\mid R=1)$. However, this estimator will be biased if there is a remaining direct relationship between $X$ and $R$. 

Another approach, which was introduced by \cite{mccartan2025estimating} and is implemented in the \texttt{birdie} package in \texttt{R}, relies on the assumption $X\indep S\mid G,R$ and requires the BISG probabilities to vary sufficiently across surnames for each geography. Under these assumptions, one can estimate $\Pr(X=x\mid R=1,G)$ by regressing an indicator variable $\ind(X=x)$ on the BISG probabilities and geography or through an EM algorithm detailed in \cite{mccartan2025estimating}.  Overall, which of the two approaches is more plausible will depend on the nature of $X$ and $Y$. Though the raking targets will have estimation error, if they are closer to the true distribution than the unadjusted survey distribution, raking can be beneficial. 

In our application, age, partisanship, and gender are available in the voter file, but the appropriate conditional independence assumptions differ across these variables.  For age and gender, we assume that Jewish identity is independent of these characteristics conditional on surname and geography, allowing names to be closely associated with age and gender. However, for partisanship, we instead assume that, conditional on Jewish identity and state of residence $G$, surname provides no additional information about partisanship. We therefore use equation~\ref{eq-weightedest} for age and gender and the \texttt{birdie} approach for partisanship. 

\subsection{Understanding the quality of BISG probabilities}\label{sec-estquality}

Next, we consider how the quality of the BISG probabilities, regardless of how they are estimated, affects the expected yield and downstream estimation discussed in Sections~\ref{sec-sampling}~and~\ref{sec-estimation}. We start by distinguishing two aspects of BISG probability quality --- inherent precision and estimation error --- while setting aside for now the problems of non-response. 

\subsubsection{Precision}

The precision of the true BISG probabilities creates an inherent limitation on yield and efficiency that can only be addressed by adding more informative predictors. For example, as shown in a simplified setting in Supplement~\ref{app-precisionanalysis}, if in each geography, $a\%$ of units have high BISG probability $p$ and $(1-a)\%$ of units have low BISG probability $q<p$, the expected yield satisfies 
\begin{equation}\label{eq-expectedyieldcasestudy} \E[\sampledrare] = q\targetoverall \left(\frac{1 + \frac{a}{(1-a)}\frac{p^2}{q^2}}{1 + \frac{a}{(1-a)}\frac{p}{q}}\right)  = p\targetoverall\left(\frac{1 + \frac{(1-a)}{a}\frac{q^2}{p^2}}{1 + \frac{(1-a)}{a}\frac{q}{p}}\right) \in [q\targetoverall, p\targetoverall]
\end{equation}
As $a\rightarrow 1$ (the high-BISG probability group becomes common), the expected yield approaches the upper bound, but as $a\rightarrow 0$ (the high-BISG probability group becomes rare), the expected yield approaches the lower bound. Hence, even if $p$ is much greater than $q$, the precision of the highest-BISG probability units constrains the yield, and the presence of a large number of low BISG probability units can decrease the expected yield. 

In Supplement~\ref{sec-sim}, we simulate a sampling set-up that mimics our application and varies the precision of the BISG probabilities. We find that even when the $R=0$ and $R=1$ group have only a small mean difference in their BISG probabilities, the success rate of the Poisson sampling is better than random sampling $\targetg$ units in each state (see figure~\ref{fig-mainsimresult}). Yet even when the mean difference in BISG probabilities is quite large (e.g., in the range 0.50--0.75), the success rate in sampling minority population members can be much less than $1$.

Precision also affects the variance of downstream estimation. If BISG probabilities vary more across individuals, this also increases the variance of the sampling weights, which can increase variance in IPW style estimators. However, artificially upper or lower bounding the BISG probabilities at the sampling stage has costs for similar reasons to the dynamics of equation \ref{eq-expectedyieldcasestudy}: setting low probability units to even slightly higher probabilities can create big declines in yield if there are many such units. In our application, we instead use post-hoc trimming of the most extreme IPW weights, which trades-off some increased bias for a decrease in variance.

\subsubsection{Estimation error}

The degree of estimation error in the $\widehat{\Pr}(R=r\mid S=s,G=g)$ reflects data quality and modeling techniques. Although estimation error in sampling probabilities does not introduce bias into final survey estimation (as long as it does not set sampling probabilities to $0$), it can hurt yield as well as our ability to estimate raking targets and expected yields prior to the survey. For example, when sampling with estimated BISG probabilities, the per-stratum yield probability in equation~\ref{eq-poissonsuccessrate} above becomes
\begin{equation}\label{eq-poissonsuccessrate2}
    \Pr(R=1\mid I=1, G = g)= \frac{\sum_{s=1}^{|\cS|} \popoverallgs\widehat{\Pr}(R=1\mid S = s, G = g)\Pr(R=1\mid S=s,G=g)}{\sum_{s=1}^{|\cS|}\popoverallgs\widehat{\Pr}(R=1\mid S = s, G = g)}.
\end{equation}
Hence, this can be affected by the difference between estimated and true BISG probabilities. Although the true BISG probability is unknown, such formulas can be used for theoretically exploring the impact of different magnitudes and kinds of errors.

In Supplement~\ref{app-sensitivitytoerror}, we provide an illustrative example of what happens to the minority population expected yield $\E[\sampledrare]$ and yield variance $\Var(\sampledrare)$ under the following additive mean-zero error mechanism,
\begin{equation}
    \widehat{\Pr}(R=r\mid S=s,G=g) =   \Pr(R=r\mid S=s,G=g)+\epsilon_{s,g},
\end{equation}
where this probability is assumed to be non-zero. First, when $\epsilon_{s,g}$ are uncorrelated with the BISG probabilities, the expected yield $\E[\sampledrare]$ is unchanged, but if they are negatively correlated (i.e., large BISG probabilities are under-estimated and small probabilities are over-estimated), then this decreases the yield. In contrast, positively correlated errors, where high BISG probabilities are boosted even higher, may increase yield. However, this can come at a cost downstream if the highest BISG units differ systematically from others. 

Again, assuming no probabilities are set to zero and setting aside non-response problem, estimation error in BISG probabilities affects downstream estimation mainly through variance. While no bias is introduced in design-based estimators, if the errors $\epsilon_{s,g}$ are correlated with the outcome of interest $Y$, this can increase variance if sampling probabilities are decreased (weights increased) for units with a greater variance of $Y$. For example, in sampling the Jewish population, if surnames are correlated with denomination, over-estimating one denomination's BISG probabilities much higher than their true values would have a high variance cost for estimating the denomination distribution. Even when BISG estimation errors are non-systematic, errors that increase variance in the estimated BISG probabilities may increase downstream variance. On the other hand, errors such as those coming from smoothing and shrinkage could also decrease downstream variance if they decrease survey weight variability. 

Returning to the problem of non-response, we emphasize that while BISG estimation error does not introduce bias at the sampling design stage, it also offers no protection against non-response bias. Estimation error affecting who we choose to sample will not in itself have any effect on non-response dynamics. It could have such an effect if errors are correlated with a variable related to non-response propensity. For example, over-estimating probabilities for and thus over-sampling a Jewish denomination which is also particularly unlikely to respond to the survey would hurt yield. A fundamental issue is that BISG estimation error can affect some methods of non-response adjustment. As described in Section \ref{sec-estimation}, BISG probabilities can be used to estimate targets $\Pr(X=x\mid R=1)$ for non-response adjustments. Here, errors in BISG probabilities can bias the estimation of these targets and hence decrease the quality of the downstream analysis.

Finally, although one would like to evaluate the BISG estimation error prior to running the survey, this would require ground-truth labels for minority membership that are generally not available. In our application, we used three indirect pre-survey strategies for validating the data that produced BISG estimates: (1) checking for agreement among data sources in their implied geographic distributions of the Jewish population, (2) checking that the obituary data cover and have higher frequencies for names known to be distinctively Jewish, and (3) manual review of some examples. While each check is imperfect, they together build confidence in the estimated probabilities. We describe these results further in Supplement~\ref{sec-validation}. 

Note that once the survey is conducted, researchers can directly evaluate estimation accuracy by estimating BISG probability calibration among survey respondents \cite{van2016calibration}. We show this for our survey in Figure~\ref{fig-per-state-success}. This post-survey evaluation is especially of use for researchers who are able to run a pilot study. The initial survey then provides a means of correcting calibration errors for the next round of sampling \cite{rosenman2023recalibration}.

\subsection{AI use}\label{sec-AIuse}

AI was used in this project in a few ways. First, we used the ChatGPT API (version GPT-4) to process the text of the obituaries to extract first name information. Details of this extraction are described in Supplement~\ref{app-obit-data}. Second, we used ChatGPT (version GPT-5.6 Sol) to review our derivations in small chunks at a time for typos or algebraic mistakes (see Supplement~\ref{sec-surname_dist_results}). Finally, ChatGPT was used for help with converting \texttt{R} code for the MCMC sampler to \texttt{Rcpp} (version GPT-4). ChatGPT was also used for help with plotting code and debugging (versions GPT-4, GPT-5, and GPT-6).  

\section{Results}\label{sec-empiricalresults}

In June 2025, we applied the proposed methodology to sample registered voters from the L2 voter files. Our target sample size was 50,000, and we ultimately sampled 49,546 voters. We then sent the sampled individuals a postcard solicitation to complete an online survey that asked about Jewish identity and a number of questions related to civic, political, and religious topics. A postcard invitation was used because mailing addresses are more complete in the voter file than phone numbers or emails. In this section, we focus on evaluating the empirical performance of the proposed sampling methodology and report the results of a full analysis of the survey response elsewhere.  

The sampled individuals represented 11,613 distinct surnames and 5,589 distinct first names and resided in all 50 states and the District of Columbia. Table~\ref{tab-surnameranks} summarizes how often the sampled surnames came from the most common names in the obituary. While 8\% of those sampled had one of the most common surnames in the obituary data, 35\% of those sampled had surnames that were not even among the top 1,185 most frequent obituary surnames. This indicates that our sample is substantially more diverse in terms of surnames than a sample formed from filtering to a small list of common surnames would be. The BISG probabilities of those sampled ranged from $3.7\times 10^{-6}$ to $1$ with a median of $0.62$ while the normalized sampling probabilities ranged from $5.7\times 10^{-8}$ to $0.09$ with a median of $0.01$. BISG estimates that are near one mainly occurred because of highly indicative first name ratios as discussed in Section~\ref{sec-first names} and Appendix~\ref{app-samplingproc-summary}. After normalization, these translated to sampling probabilities much lower than one. 

\medskip 
\begin{table}
\centering    
\renewcommand{\arraystretch}{1.2}
    \begin{tabular}{l|ccccc}
       \hline 
       Obituary frequency rank  & $\geq 10$ & $\geq 100$ & $\geq 200$ & $\leq 250$ & $>250$ \\
       \hline 
       Number of obituary surnames & 10 & 110 & 376 &  1,185 &  48,013\\
       Fraction of sampled individuals & 0.08 & 0.32 & 0.48 & 0.65 & 0.35  \\
       \hline 
    \end{tabular}
\caption{\textbf{Obituary frequency ranks of sampled surnames}. Each surname in the final sample from the voter file is assigned its rank in terms of obituary data frequency, where surnames with the same obituary data count are assigned the same rank. Row 1 indicates how many surnames had the given rank range. Row 2 indicates what fraction of the sample had surnames with this rank range. The median surname rank among sampled individuals was 208. The maximum rank was 272, corresponding to surnames occurring once in the obituary data.}
    \label{tab-surnameranks}
\end{table}

\subsection{The proportion of Jewish respondents}\label{sec-response}

Of the 49,546 people who were sent a postcard asking them to complete the online survey, 1,765 individuals responded, resulting in a response rate of approximately 3.6\%. This number represents response  after a single postcard solicitation and with no incentive such as cash payment for completing the survey. Pew's 2020 study of Jewish Americans had a much higher response rate (16.6\%) than ours, but they offered incentives (\$10 - \$50) for respondents, and allowed respondents to answer the survey either online or by mail. The mean response rate across states was 4.2\% with a standard deviation of 1.8\%, and we had respondents from every state except North Dakota (see Table~\ref{tab-sampleresponsecounts} in Supplement~\ref{app-response-rate}).

\begin{figure}
    \centering \includegraphics[width=\linewidth]{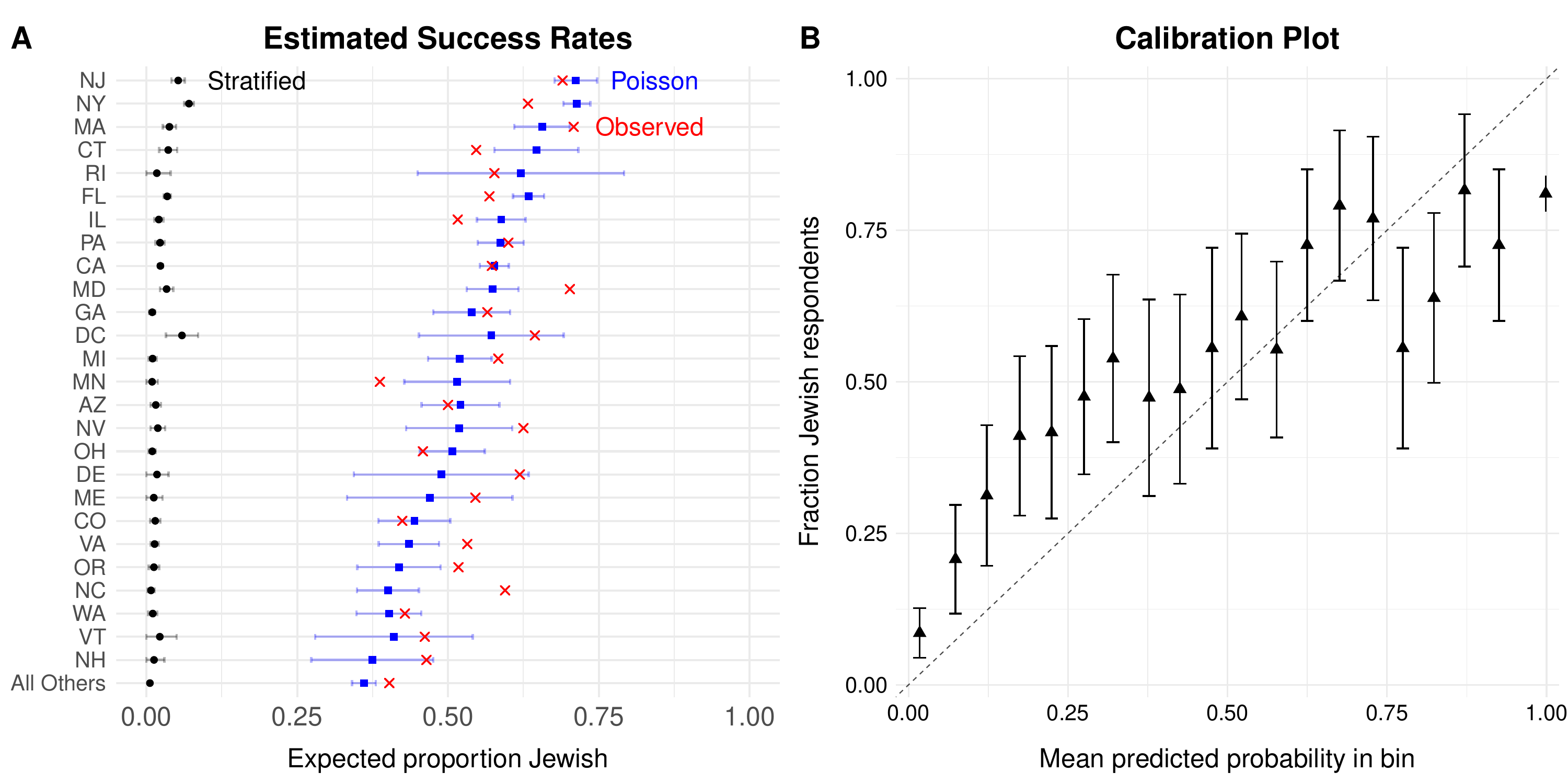}
    \caption{\textbf{Realized success rates by state and overall calibration of BISG probabilities.} \textbf{(A)} The left panel of Figure~\ref{fig-yield_and_allocation}, now with the observed fraction of survey respondents who were Jewish marked by $\times$.  \textbf{(B)} calibration plot for estimated $\Pr(R=1\mid F=f,S=s,G=g)$. Each point is the fraction of Jewish respondents versus the mean estimated probability among respondents with estimated probabilities in that bin (20 equal-width bins). Bars are twice the standard deviation of a binomial draw with the bin average probability.}
    \label{fig-per-state-success}
\end{figure}

Among the respondents, 1,004 or 56.9\% identified themselves as Jewish. This is close to our estimated expected success rate of 58.6\% reported in Table~\ref{tab-samplingsuccerates}, representing a substantial gain in response efficiency over the expected success rate of stratified sampling. Furthermore, the left plot of Figure~\ref{fig-per-state-success} demonstrates that at the state level, the observed fraction of respondents who were Jewish aligns closely with the expected success fractions despite the fact that these observed fractions also reflect unit non-response rather than the sampling of Jewish respondents alone. The right plot of Figure~\ref{fig-per-state-success} indicates that our estimated $\Pr(R=1\mid F=f,S=s,G=g)$ probabilities were reasonably well calibrated for those who responded to the survey.

Among all respondents, the range of BISG probabilities and sampling probabilities was similar to that of the full sample with a slightly higher median BISG probability. The BISG probabilities ranged from $0.0001$ to 1 with a median of $0.73$ and the sampling probabilities from $4.3\times 10^{-6}$ to $0.075$ with a median of $0.011$. This suggests that response was at most only slightly biased in a way that favored people with higher probabilities of being Jewish.  Among the Jewish respondents, these distributions shift upward as one would expect. The unnormalized BISG probabilities of the Jewish respondents ranged from $0.0001$ to $1$ with a 0.25 quantile of $0.59$ and a median of $0.999$. Thus, many of those sampled were those who were assigned high BISG probabilities, though as evidenced in Figure~\ref{fig-per-state-success}, there were still respondents with BISG probabilities across the entire range of probabilities.  While some names are highly indicative of being Jewish, Figure~\ref{fig-per-state-success} suggests that the highest probabilities were somewhat over-estimated -- an example of something a pilot study could have allowed us to detect and adjust for. The final sampling probabilities of the Jewish respondents ranged from $0.000016$ to $0.075$ and had a median of $0.014$.

\medskip 
\begin{table}[!t]
\centering
\small 
\renewcommand{\arraystretch}{1.2}
\begin{tabular}{lccccc}
\hline 
Variable 
& \shortstack{All sampled \\ people}
& \shortstack{Non-Jewish \\ respondents}
& \shortstack{Jewish \\ respondents}
& \shortstack{Jewish \\ respondents \\ (Pew weights)} & \shortstack{Jewish \\ respondents \\ (Voter file weights)} \\
\hline  
Mean Age           & 55.1 & 61.7 & 64.3 & 51.0 & 54.6\\
Percent Female        & 49.5 & 53.7 & 39.0 & 45.5 & 49.3 \\
Percent Democratic    & 45.8 & 52.3 & 68.1 & 68.3 & 50.3\\
\hline 
\end{tabular}
\caption{\textbf{Comparison of sampled, responding, and Jewish group demographics.} The first three columns show unweighted means for all sampled respondents, non-Jewish respondents, and Jewish respondents.  The second-to-last column shows weighted estimates for Jewish respondents only, using inverse-probability weights raked to Pew population benchmarks. The last column shows the result of instead weighting to targets estimated from the voter file as described in Section \ref{sec-estimation} and Supplement \ref{app-bisgestimatedweights}.
}
\label{tab-descriptive-char-respondents}
\end{table}

Using demographic and political information available in the voter file, Table \ref{tab-descriptive-char-respondents} shows the demographics of sampled individuals, those who responded, and Jewish respondents. Our respondents, and particularly our Jewish respondents, skew older, more male, and more Democratic. The Democratic skew is not surprising given that previous research has found Jewish Americans to lean Democratic \cite{pew_jews_politics_2021}, while the age and gender skew may reflect different tendencies to respond to the survey. Higher rates of response for older people and Democrats have been found in other mail surveys unrelated to Jewish sampling \cite{Hersh_doctors_2016}. Although Jewish respondents were only slightly older on average than non-Jewish respondents, the proportion of female respondents was unexpectedly lower among Jewish respondents, and the reason for this pattern is unclear. The fact that we sampled almost equal fractions of males and females suggests this is a response dynamic more than a sampling dynamic and hence is unlikely to arise from Jewish women more often having non-Jewish surnames from their spouses. In our analysis below, we use two raking weight options to adjust age, gender, and partisanship related skew in Jewish respondents. The last two columns of Table~\ref{tab-descriptive-char-respondents} show the resulting weighted survey estimates of age, gender, and partisanship for Jewish respondents.

\subsection{Ethics and disclosure}\label{sec-ethicsdiclosure}

Sampling a minority population can raise ethical concerns if that population has a history of being targeted and attacked, as is the case for the Jewish population. Our study was conducted at a time of heightened antisemitism and political tension in the United States, so surveying the Jewish population required care. As noted in Section~\ref{sec-estimation}, the ethical imperative to be transparent about the survey's purpose may be in tension with the issue of non-response bias. To balance this, we initially disclosed to respondents only that the survey would ask about civic, religious, and political attitudes. At the end of the survey, respondents were informed of the full motivation of the study and given an additional opportunity to opt out. Twenty-two individuals not included in our reported response count withdrew their records after reading the post-survey disclosure. 

Our study was designed in consultation with and was approved by the Tufts University Institutional Review Board (STUDY00005675). The postcard and full text of the pre- and post-survey disclosures are provided in Supplement~\ref{sec-designimplementation}. The survey data, which do not contain personally identifying information, are available in a public dataverse archive. We round the survey weights to reduce the potential ability to infer respondents' exact names from them if anyone was to gain access to the sampling frame and BISG probabilities.  

While the AJPP geographic prevalence estimates are publicly available, we are unable to share our estimated BISG probabilities for two reasons: (1) they were calculated from proprietary data sources (L2 and PBI), and (2) there is risk of misuse and associated harm from making individually identifying information and Jewish probabilities public. Researchers who are interested in applying similar methods for the Jewish population may consider requesting the data from PBI themselves. 

Alternatively, one can apply a new estimation method called $\ell$BISG, which we recently introduced \cite{Chasalow_etal_lbisg} to estimate BISG probabilities when no frequency data are available. While $\ell$BISG does rely on the traditional BISG assumption, which we were able to relax in this paper, we find that the performance of the method can be close to that of the traditional BISG probability in terms of precision, calibration, and downstream inference accuracy. Access to the L2 data is only needed if one intends to use this as the sampling frame -- other sampling frames may be used. Given estimated BISG probabilities for a sampling frame, the rest of our framework is applicable.

\subsection{Cost comparison}

The direct cost of the survey was the printing and postage of 49,546 postcards, which cost \$22,833. If we focus just on the sample of Jewish respondents (excluding people who took the survey but did not identify as Jewish), our sample included 1,004 Jewish respondents, and so the direct costs were about \$23 per respondent in the target population. 

The direct costs for Pew were much higher. There, approximately 339,000 households were sent a screening survey with a \$2 incentive. The main survey, which was completed by 4,718 Jewish respondents, had a \$10-20 incentive, but some respondents were incentivized with up to \$50. Conservatively, if we assume that the 4,718 Pew respondents received, on average, a \$12 incentive, then the main survey incentives cost \$56,616, and total incentive costs were approximately \$734,600. Moreover, if Pew spent at least \$1 on postage and printing for all 339,000 individuals (a likely under-estimate because the Pew study used longer paper mailers), then the total survey cost is at least \$1,073,600, or approximately \$228 per completed Jewish interview. This is about ten times more expensive than our survey. 

The Pew Survey was substantially longer, with dozens more questions than our survey, so Pew obtained more information per survey completion. Nevertheless, the biggest cost savings relative to Pew were in the more efficient screening to identify Jewish-identifying respondents -- had our survey been longer, the cost difference would likely have been similar.

\subsection{Comparison to Pew}\label{sec-pew-comparison}

To probe whether our survey also allows accurate inference after weighting adjustments, we compare survey estimates to those from the 2020 Pew survey. This survey has its own limitations; it was taken 5 years earlier and had a sampling frame based on US Post Office addresses rather than the voter file.  Still, we reason that if our survey were to show major differences even on questions we expect to have fairly stable responses between 2020 and 2025, this could suggest a bias in our survey. Similarity on these questions does not prove accuracy -- the two surveys could have common biases or their target populations could have true differences that the results miss -- but Pew remains the best available source for a baseline comparison.

To enable the comparison, our survey included many of the questions that Pew asked in 2020. Between when Pew surveyed Jewish Americans in 2020 and our survey in 2025, the October 7, 2023 attack in Israel and subsequent war led to increases in antisemitism, a large pro-Palestinian protest movement in the United States, and frequent news stories about Jews and Israel. There is some evidence that there was an increase in the salience of Jewish identity and then a reversion in this period of time, which could lead to changes in Jewish identity and behavior \cite{hershlyss25}. We therefore focus our comparison questions from Pew on denomination, social networks, and rituals and practices that we expect to be more stable in the 2020-2025 period. For example, we expect that  whether one owns a ritual object or has friends who are Jewish is likely to fluctuate less with short-term political factors than the self-reported strength of one's religious identity.

\begin{figure}[!t]
    \centering
    \includegraphics[width=\linewidth]{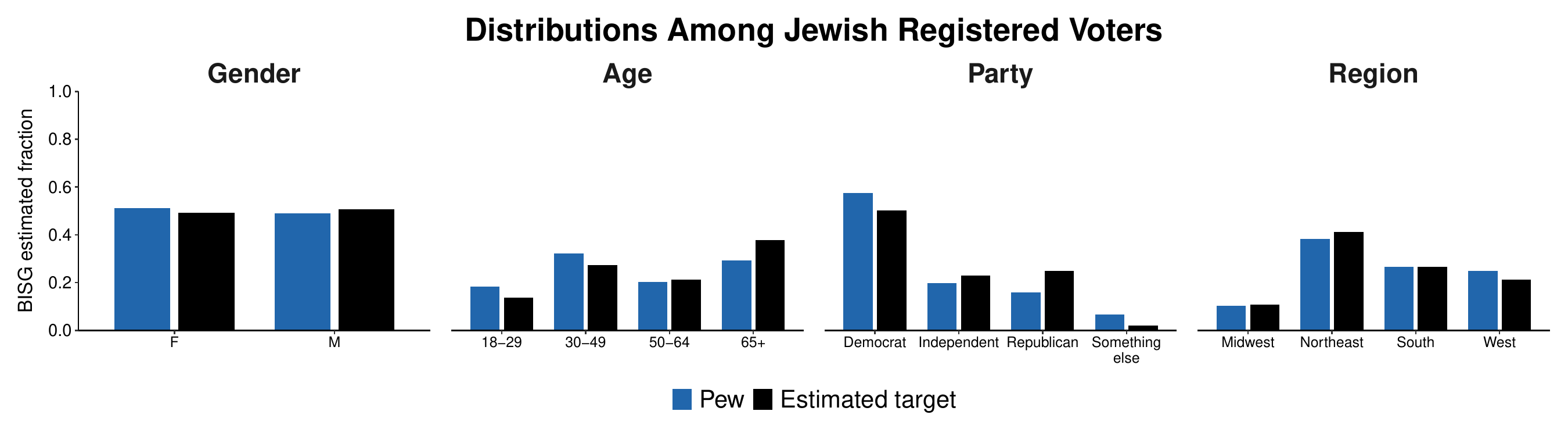}
    \caption{\textbf{Comparison of estimated distributions over gender, age, and partisanship for the Jewish population.} Compares Pew targets to the outputs of the methods described in Section \ref{sec-estimation}. Within each state, we use equation~\ref{eq-weightedest} for gender and age and the BIRDiE estimator from \cite{mccartan2025estimating} for partisanship. Then we use the AJPP estimates of $\Pr(G=g\mid R=r)$ to aggregate over states. For the rightmost region plot, we compare Pew to AJPP estimates for the regions used by Pew.}
    \label{fig-bisgrakingweights}
\end{figure}

To estimate the Jewish mean response on each question on our survey, we compare three types of estimates to the original Pew estimates.  First, we rake to match the marginal demographic distributions reported for the 2020 Pew sample on age, gender, race, party affiliation, and region (Northeast, Midwest, South, West).  These raking targets from Pew are themselves the result of Pew's model-based estimation of the demographic composition of the Jewish adult population, which may not be entirely accurate \cite{pewjewish2020}. However, by aligning our weights to match Pew, we can isolate whether any differences in results arise from differences other than the demographic composition of the sample. Second, we rake to age, gender, and partisanship targets estimated from the voter file sampling frame using our BISG probabilities as discussed in Section~\ref{sec-estimation} and to the region estimates implied by AJPP's state estimates. This represents a more realistic approach available even if no pre-existing raking targets are available.  

Figure~\ref{fig-bisgrakingweights} compares these targets to the Pew targets. The patterns are similar, but there is some discrepancy on partisanship (we estimate fewer Democrats) and age (we estimate more 65+). This could reflect differences between the registered voter population and Pew's Jewish adult target population rather than error. Further details and tables are reported in Supplement~\ref{app-bisgestimatedweights}. For both raking approaches, we initialize the raking algorithm with IPW weights trimmed at the 10th and 90th percentiles to reduce their variance \cite{mercer2018weighting}. 

Our third estimator uses only these IPW weights based on the sampling design. This approach is possible even if when no raking targets are available to researchers. It captures the impact of not adjusting for demographic differences in response.

\begin{figure}
  \centering   
  \includegraphics[width=\linewidth]{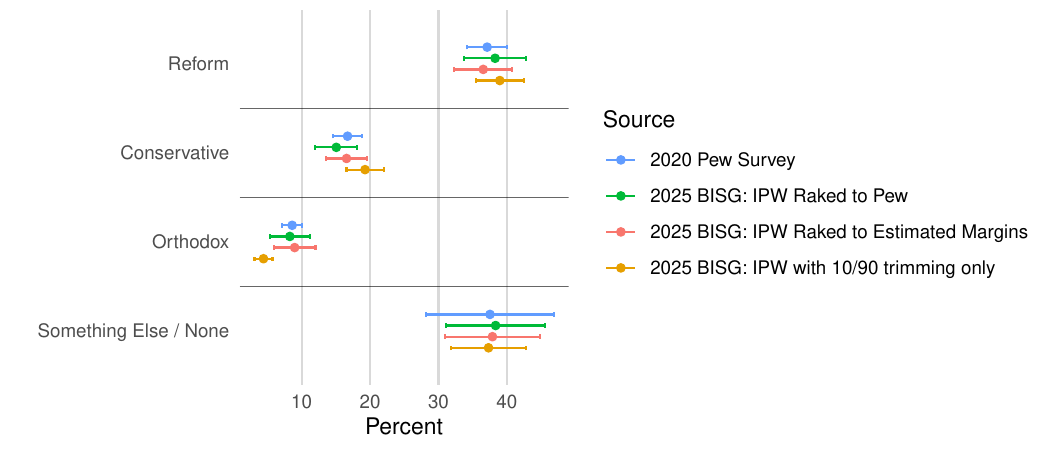}
  \caption{\textbf{Survey results for denomination.} Estimates of the distribution over responses to the question: ``Thinking about Jewish religious denominations, do you consider yourself to be \ldots'' Compares Pew estimates to three possible estimates from our survey.}
    \label{fig:denom}
\end{figure}

Our first analysis examines Jewish denominational affiliation as a measure of Jewish identity. Figure~\ref{fig:denom} compares the distribution of respondents identifying as Reform, Conservative, or Orthodox Jews, with all remaining denominations as well as responses indicating no affiliation combined into an ``Other/None'' category. For the raking estimators, the distribution in our survey closely matches that reported by the Pew survey. Design-adjusted $\chi^2$ tests \cite{rao1984chi} comparing each of the distributions to the Pew distributions yield p-values of 0.8676 (Raked to Pew), 0.9700 (Raked to voter file estimates), and 0.0062 (IPW only). These results indicate no statistically significant difference for the raking options. For the IPW only option, which does not correct for response bias, the pattern is still broadly similar to Pew. However, there is some statistically significant discrepancy arising mainly from a lower estimate of the fraction of Orthodox denomination.

Our survey also closely matches the Pew survey in the proportion of respondents who identify as Jewish culturally or ethnically but not religiously (a separate question from denomination). These individuals describe themselves as agnostic, atheist, or having no religion when asked about their religious affiliation, while also identifying as Jewish. Pew estimates that 26.76\% (95\% CI: 23.82--29.93) of American Jews fall into this category, compared with 25.59\% (95\% CI: 21.54--30.12) in our survey raked to Pew and 23.77 (95\% CI: 20.05--27.94) for our survey raked to voter file estimated targets and 22.91 (95\% CI: 19.99--26.12) for our survey with IPW weighting only.

\begin{figure}
\includegraphics[width=\linewidth]{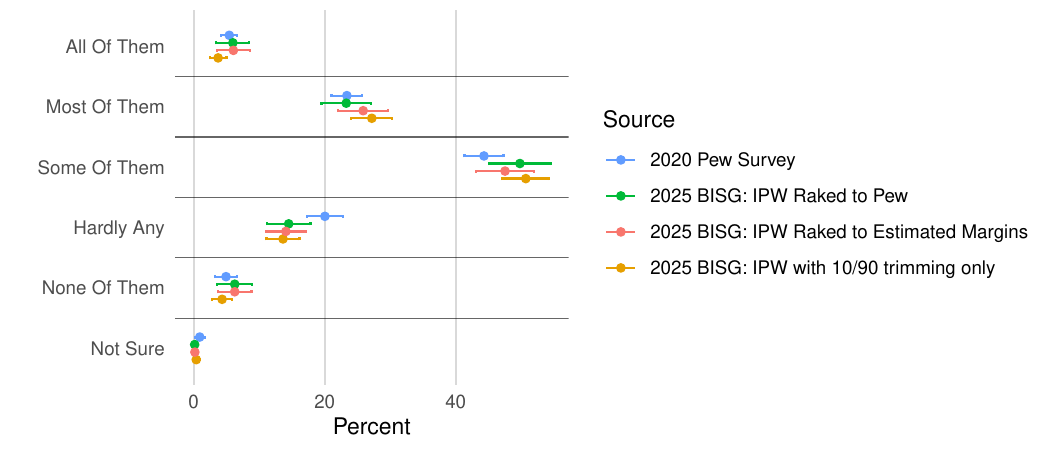}
 \caption{\textbf{Survey results for friendship network.} Estimates of the distribution over responses to the question: ``How many of your close friends are Jewish?'' Compares Pew estimates to three possible estimates from our survey.}
 \label{fig:friend}
\end{figure}

Second, we examine a social network question about how many of one's close friends are Jewish. We see no statistically significant variations in responses. Pew has slightly more respondents who say hardly any of their close friends are Jewish and slightly fewer who say some of their close friends are Jewish, but Figure~\ref{fig:friend} shows that the overall ranking and shape of the distribution is still quite similar. The design-adjusted $\chi^2$ test for distributional difference yields a p-value of $0.0820$ (Raked to Pew), 0.07087 (Raked to voter file estimates), and 0.0018 (IPW only). Again, the IPW estimator is the only statistically significant difference. That is, when raking is used, we have no evidence that, compared to Pew, our survey respondents have different Jewish social networks.

\begin{figure}
\includegraphics[width=\linewidth]{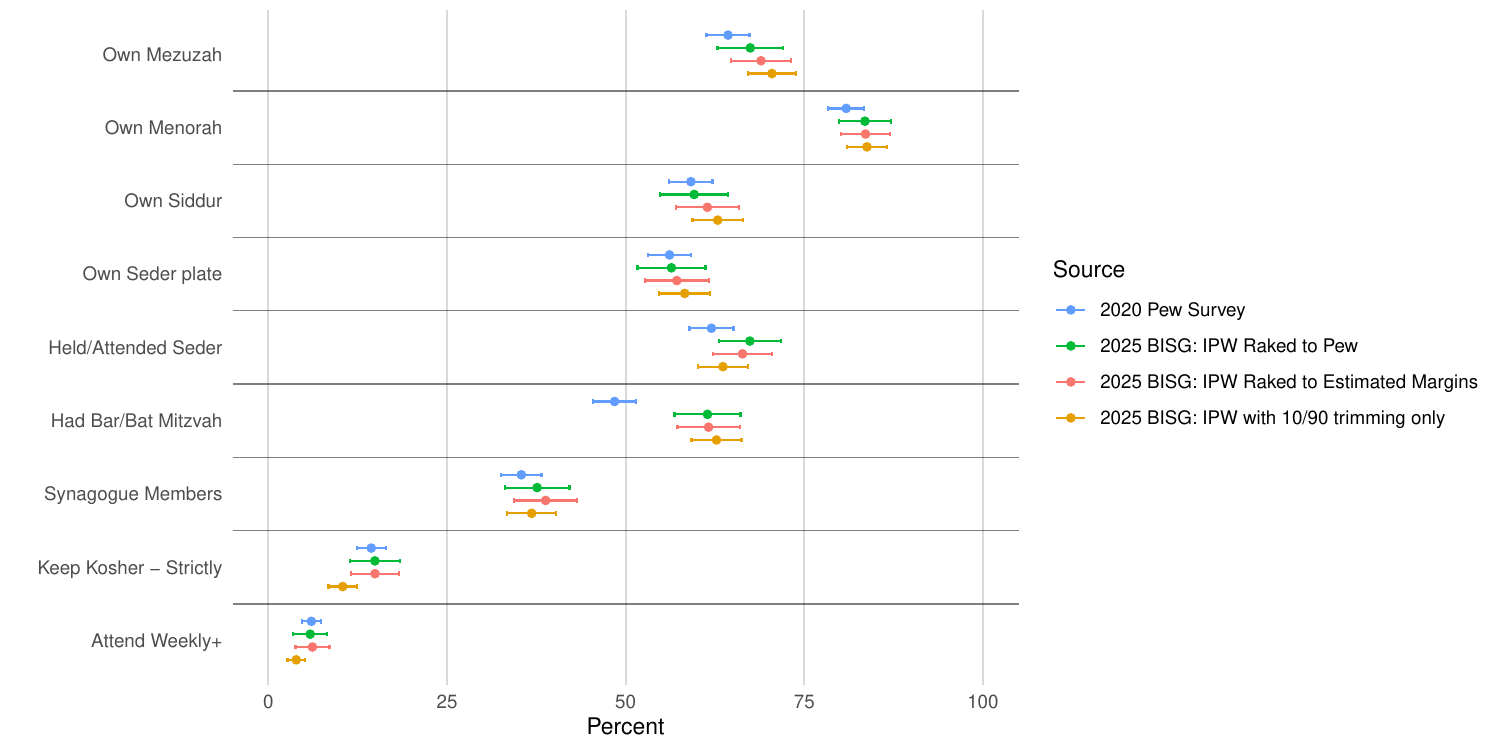}
\caption{\textbf{Survey results for key Jewish practice measures.}  These were each separate survey questions that are here combined on one graph. For both raking estimators, only Bar/Bat Mizvah is statistically significantly different from the original Pew estimate at the $0.01$ level (this is without any multiple testing correction). For the IPW only estimator, ``Own Mezuzah,'' ``Had Bar/Bat Mizvah,'' ``Keep Kosher,'' and ``Attend at least weekly'' all cross this threshold.}
    \label{fig:multi}
\end{figure}

Finally, Figure~\ref{fig:multi} shows questions focused on concrete behaviors and practices. These include the possession of ritual items (a mezuzah, a menorah, a siddur, and a seder plate), and participation in Jewish rituals (attending a seder, having a bar/bat mitzvah as a child, attending weekly synagogue services, keeping kosher).  
Our survey aligns closely with Pew on all but having a bar/bat mitzvah. There is no obvious reason why this particular measure would operate differently than the other measures here and it may be a matter of chance. We also examined the counts of how many of nine religious behaviors respondents said applied to them. In the 2020 Pew survey, Jewish respondents participated in, on average, 4.27 of these (95\% CI: 4.12-4.41). In our 2025 BISG survey, the average is similar, at 4.54 (95\% CI: 4.29--4.78) for raking to Pew weights, 4.58 (95\% CI: 4.35--4.82) for raking to voter file estimated targets, and 4.52 (4.35--4.70) for IPW only. 

Altogether, these results suggest that there is remarkably close alignment between our survey and Pew.  This is true both when weighted in the same way as Pew and when weighted to targets that do not depend on having access to Pew. Our analysis suggests that the proposed method is an effective way to sample the American Jewish population with a relatively high success rate, at a substantially lower cost than the Pew survey method, and that it yields results similar to those of the Pew survey.

In Supplement~\ref{sec-empiricalresults}, we provide additional details on the survey weighting and show that the IPW results are quite stable under other trimming options. Supplement~\ref{app-ces} also compares the Pew survey, our BISG survey, and Jewish-identifying respondents in the most recent Cooperative Election Study (CES). Although the CES was not designed to be representative of the U.S. Jewish population, it yields estimates similar to those from Pew and our BISG survey for the proportion of Jewish respondents who regularly attend religious services. The CES differs more substantially from both Pew and our BISG survey in denominational affiliation, but its estimate of the proportion of Jewish respondents without a college degree is closer to Pew’s than to our BISG survey’s. 

\section{Discussion}

This article aims to solve a practical problem in sampling small populations that are not directly identified in a sampling frame. We proposed a method that leverages BISG probabilities in a stratified Poisson sampling scheme, and showed how to select strata allocations and control the expected sample size. In our case study of sampling American Jewish adults, we leveraged a novel data source -- obituaries --  that contain information about the prevalence of Jewish surnames and first names. We combined these with the voter file and existing estimates of the geographic distribution of Jewish Americans across U.S. states.

Applying our proposed method, we conducted an original survey that asked about the identity, social networks, and religious behaviors of the registered voter Jewish American population. We generated a sample of 1,004 Jewish Americans at a low cost with a single postcard solicitation. The survey estimates are nearly identical to those of Pew, which is widely considered the highest-quality survey of Jewish Americans. This result was robust under multiple weighting schemes, including those that explicitly match to Pew's raking targets and those that construct targets using estimated BISG probabilities and the sampling frame. While it is possible that our survey and Pew’s suffer from some shared biases, we did not observe
any obvious signs of under-representation of important groups such as different denominations. 

There are several opportunities to build on this work. For sampling Jewish Americans, one useful extension would be to apply our approach to a sampling frame which includes people not registered to vote or to a regional survey, such as a single metro area. Regional Jewish communities periodically survey
their populations, and our approach would likely offer an efficient way to do so. It may also be possible to further improve the BISG probabilities by using county-level geographic prevalence and modeling higher-resolution geographic information such as movement between states -- particularly the tendency for retirees to move to Florida (see Supplement~\ref{app-statelevelchoice}). Finally, it would be valuable to see if these results can be replicated in other survey modes, such as phone and texting based surveys, which might have different response dynamics. 

More generally, other minority populations of interest to surveyors may have similar relevant surname frequency data, either from obituaries, census data, or other sources, that allow BISG estimation. For example, other religious groups such as Muslims, Buddhists, Hindus, and Mormons also have specific burial practices or institutions. It may be possible to collect similar name data for these groups. Ensuring that the resulting collection captures representative frequencies by geography remains the key challenge.

Researchers may also develop other new ways of estimating improved BISG probabilities and then use them in the Poisson sampling approach. For example, we can use a new BISG method called $\ell$BISG \cite{Chasalow_etal_lbisg} to estimate BISG probabilities even when only group-specific name lists are available. As mentioned in Section~\ref{sec-ethicsdiclosure}, $\ell$BISG can be combined with the methods in this paper to target minority populations, for which it is difficult to collect any kind of frequency data.

Religious, racial, and other identities are often fluid and subject to individual interpretation. In our survey, we relied on self-identified Jewish identity, whether religious or not. Because our sampling frame was based on obituary data, over half of the voter file had a positive probability of selection, including those with less common or less distinctly Jewish names. However, individuals with surnames that rarely appeared in the obituary data had lower sampling probabilities, limiting our ability to study difficult-to-identify subgroups, such as Jews with non-Jewish fathers. Similar challenges will arise when sampling other populations and may require combining our approach with complementary sampling strategies to ensure adequate coverage of all subgroups of interest.

We close this paper with a word of caution. While there is often scholarly and public interest in generating accurate surveys of minority groups, some may worry that identifying themselves as a member of the group could pose a risk. Even generating predictive probabilities can lead to unethical uses.  Our own study was motivated by scholarly interest and the goals of non-profit organizations supporting the Jewish community that wish to better understand their constituents. We also reflected on feedback from respondents who expressed concerns about our survey and improved our disclosure protocols accordingly. Research in this space should be pursued with particular care and sensitivity.




\clearpage 

%



\section*{Acknowledgments}

We thank reviewers of the Alexander and Diviya Magaro Peer Pre-Review
Program and Marc Elliott for useful comments. Eitan Hersh thanks David Van Riper and Douglas Friedman for research assistance. Kosuke Imai thanks Bruce Willsie, the president and CEO of L2, Inc. for providing access to their voter file database.

\paragraph{Funding:} The Jim Joseph Foundation supported this study.

\paragraph{Author contributions:}$ $\\ 
\textit{Conceptualization:} K.C., E.H., K.I., L.R. \\
\textit{Methodology:} K.C., E.H., K.I., L.R.
\\
\textit{Data curation:} E.H., L.R.\\ 
\textit{Investigation:} K.C., E.H., K.I., L.R.
\\
\textit{Formal Analysis:} K.C., E.H., K.I., L.R.
\\
\textit{Visualization:} K.C., L.R. \\
\textit{Supervision:} E.H., K.I. \\
\textit{Writing-original draft:} K.C., L.R.\\ 
\textit{Writing-review and editing:} K.C., E.H., K.I., L.R. 

\paragraph{Competing interests:} There are no competing interests to declare.

\paragraph{Data and materials availability:} 
The raw data for this project comes from three sources. Geographic-level estimates of Jewish populations are publicly available from \url{https://ajpp.brandeis.edu/}. The voter file data we used is available by subscription from \url{https://l2-data.com/}. The obituary data we used is available by subscription from \url{https://obituary.pbinfo.com/}. We do not release these data sources both because they are proprietary and because there is risk of misuse and associated harm from making individually identifying information and Jewish probabilities public.

The \texttt{R} code used to combine the above three data sources and produce all analysis results, figures, and tables is available via a Harvard Dataverse repository \cite{royden2026replication}. This data repository also includes a de-identified version of the final responses to the survey. Questionnaires, codebooks, and survey code are available there as well. 

\paragraph{IRB Approval:} This study design was approved by the Tufts University Institutional Review Board (STUDY00005675).

\paragraph{AI use:} As described in the main text and supplement, ChatGPT was used for extraction of first name information from obituaries. The only other AI use for this project was (1) in debugging, refinement, and formatting of R code, particularly of plotting code, and (2) for minor fixing of typos and improvements in clarity at the revision stage. The authors wrote the original draft and all major simulation and data analysis code.

\subsection*{Supplementary materials}
Methods \\
Materials\\
Figs. S1 to S29\\
Tables S1 to S11\\
References \textit{(79-\arabic{enumiv})}\\ 


\newpage


\renewcommand{\thefigure}{S\arabic{figure}}
\renewcommand{\thetable}{S\arabic{table}}
\renewcommand{\theequation}{S\arabic{equation}}
\renewcommand{\thepage}{S\arabic{page}}
\setcounter{figure}{0}
\setcounter{table}{0}
\setcounter{equation}{0}
\setcounter{page}{1} 

\setcounter{section}{0}
\renewcommand {\thesection} {S\arabic{section}}


\begin{center}
\section*{Supplementary Materials for\\ \scititle}

Kyla Chasalow, Eitan Hersh, Kosuke Imai$^{\ast}$, Laura Royden\\ 
\small$^\ast$Corresponding author. Email: imai@harvard.edu\\
\end{center}

\subsubsection*{This PDF file includes:}
Methods \\
Materials\\
Figs. S1 to S29\\
Tables S1 to S11\\
References \textit{(79-\arabic{enumiv})}\\ 

\newpage 
\section{Methods}\label{app-methods}

This supplement presents additional details on methods, including proofs of mathematical claims. Section~\ref{sec-notation} summarizes the notation. Section~\ref{app-sampling} provides derivations related to the Poisson sampling procedure, including details for Table~\ref{tab-samplingsuccerates}. Section~\ref{sec-sampler} gives the details of our Bayesian hierarchical sampler. Lastly, Section~\ref{sec-sim} gives the results of a simulated version of our procedure designed to mimic the structure of the real application. 

\subsection{Summary of notation}\label{sec-notation}

\subsubsection*{Data variables}

\begin{fixedtable}
\centering
\renewcommand{\arraystretch}{1.5}  
\begin{tabular}{cp{14cm}}
\textbf{Notation} & \textbf{Description} \\ 
\hline
$R$ & Binary variable where $R=1$ is the minority population of interest  \\
$S$ & Surname \\
$F$ & First name \\ 
$G$ & Stratum, often a geographic unit, which takes $L$ possible values \\ 
$X$ & Additional variables in survey (possibly coming from sampling frame) that might be used for raking if external information about them is available  \\
$Y$ & A variable only available in the final survey \\ 
$I$ & Sampling indicator denoting if a person was selected into the sample. \\ 
$H=H(s)$ & Surname-based filter indicator. Any unit with a surname $s$ such that $H(s)=0$ is effectively filtered out of the sampling frame by sampling with probability 0. \\
\hline
\end{tabular}
\end{fixedtable}\medskip

\subsubsection*{Sampling design}

\begin{fixedtable}
\centering
\renewcommand{\arraystretch}{1.6} 
\begin{tabular}{cp{8.5cm}c}
\textbf{Notation} & \textbf{Description} & \textbf{Formula} \\ 
\hline
$\targetoverall, \targetg$ & Fixed target sample sizes overall and by stratum \\  
$\pi_i$ & Sampling probability for unit $i$ & $\frac{\targetgi}{\probnormgi}\Pr(R_i=1\mid S_i,G_i)$\\
$\probnormg$ & Sum of minority population probabilities for units in stratum $g$ & $\sum_{i:G_i=g}\Pr(R_i=1\mid S_i,G_i)$ \\ 
$\pi_i^*$ & Sampling probability for unit $i$ with filtering & $\frac{\targetgi}{\probnormgistar}H(S_i)\Pr(R_i=1\mid S_i,G_i)$\\
$\probnormgstar$ &Same as $\probnormg$ only excluding filtered out units& $\sum_{i:G_i=g}H(S_i)\Pr(R_i=1\mid S_i,G_i=g)$\\
\hline
\end{tabular}
\end{fixedtable}\medskip

\subsubsection*{Bayesian model parameters}

\begin{fixedtable}
\centering
\renewcommand{\arraystretch}{1.5} 
\begin{tabular}{cp{14cm}}
\textbf{Notation} & \textbf{Description} \\ 
\hline
$\thetagvec$ & Parameter vector representing per-stratum surname distribution $\Pr(S=s\mid G=g,R=1)$ for each $g$  \\
$\alphavec$ & Parameter vector representing overall surname distribution $\Pr(S=s\mid R=1)$  \\
$\gammavec$ & Hyper-parameter vector for overall surname distribution \\
$\eta$ & Positive scalar scaling parameter \\
$\rho_g$ & Scalar parameter calculated from $\eta$ which represents the degree of partial pooling for each $g$  \\
\hline
\end{tabular}
\end{fixedtable}\medskip

\subsubsection*{Counts}

\begin{fixedtable}
\centering
\renewcommand{\arraystretch}{1.4} 
\begin{tabular}{cp{8.5cm}c}
\textbf{Notation} & \textbf{Description} & \textbf{Formula} \\ 
\hline
$N$ 
  & Number of people in sampling frame  & $\sum_{g=1}^{L}\sum_{s=1}^{|\cS|}\sum_{r=0}^1\popoverallgsr$ \\
$\popoverallg$ 
  & Number of people in stratum $g$ in sampling frame  & $\sum_{s=1}^{|\cS|}\sum_{r=0}^1\popoverallgsr$ \\
$\popoverallgs$ 
  & Number of people in stratum $g$ with surname $s$ in sampling frame & $\sum_{r=0}^1\popoverallgsr$ \\
$\popoverallgr$ & Number of people in stratum $g$ from population $r$ in sampling frame &$\sum_{s=1}^{|\cS|}\popoverallgsr$\\ 
$N_{gsr}$ & Number of people in stratum $g$ from population $r$ with surname $s$ in sampling frame & $\sum_{i=1}^{\popoverall}\ind\{G_i=g,S_i=s,R_i=r\}$ \\ 
\hline
$\popoverallstar,\popoverallgstar,\popoverallgsstar$ & Counts as above but after filtering using $H(s)$ \\ 
\hline 
$n$ 
  & Number of people in final sample & $\sum_{g=1}^L \sum_{s=1}^{|\cS|}\sum_{r=0}^1 \sampledgsr$ \\
$\sampledg$ 
  & Number of people in stratum $g$ in final sample & $\sum_{s=1}^{|\cS|}\sum_{r=0}^1 \sampledgsr$ \\
$\sampledgr$ & Number of people in stratum $g$ from group $r$ in final sample & $\sum_{s=1}^{|\cS|}\sampledgsr$\\
$\sampledgsr$ & Number of people in stratum $g$ from group $r$ with surname $s$ in final sample & $\sum_{i=1}^{\popoverall}\ind\{I_i=1,G_i=g,S_i=s,R_i=r\}$\\
\hline
$\msampsize$ 
    & Number of people in training data, which includes only $R=1$ members & $\sum_{g=1}^{L}\sum_{s=1}^{|\cS|}m_{gs}$\\ 
$\mg$ 
    & Number of people in training data from stratum $g$  & $\sum_{s=1}^{|\cS|}m_{gs}$\\
$\ms$ 
    & Number of people in training data with surname $s$ & $\sum_{g=1}^{L}m_{gs}$ \\
$\mgs$ 
    & Number of people in training data from stratum $g$ with surname $s$ \\
\hline
\end{tabular}
\end{fixedtable}\medskip

\clearpage 
\subsection{Sampling Calculations}\label{app-sampling}

In this section, we first prove the well-established equation~\ref{eq-optimal_ng} for completeness; similar calculations appear in the literature \cite{kalton1986sampling, UN1993,waksberg1997, cervantes2007, clark_etal_2009, kalton2009methods, chen2015geographic}. Then, we use a similar proof approach to derive our Poison sampling allocation formulas with and without initial filtering. Next, we derive formulas for the expected sample size and variance calculations overall and for the minority population. These follow directly from the Poisson design after noticing that it implies true group membership $R$ is independent of inclusion in the sample given the BISG predictors. Finally, we present original calculations related to the role of BISG probability precision and estimation error.
      
\subsubsection{Optimal stratified sampling allocation}\label{sec-opt}

For completeness, we prove equation~\ref{eq-optimal_ng}. Suppose that in each stratum $g$, we will take a simple random sample of $\sampledg$ observations. The sample stratum mean within the minority population, denoted $\overline{Y}_{g1}$, is an unbiased estimator of $\E[Y\mid G=g,R=1]$. Consider  estimating $\mu = \E[Y\mid R=1]$ using the post-stratified $\widehat{\mu}_{\mathrm{strat}}  = \sum_{g=1}^{L}\overline{Y}_{g1}\Pr(G=g\mid R=1)$, where we treat $\Pr(G=g\mid R=1)$ as known. Ignoring finite population corrections, the variance of $\overline{Y}_{g1}$ is given by,
\begin{equation}\label{var-approx}
\Var(\overline{Y}_{g1}) = \Var(Y\mid G=g,R=1) \E\left[\frac{1}{\sampledrareg} \right] \approx \frac{\Var(Y\mid G=g,R=1)}{\E[\sampledrareg]}
\end{equation}
where the approximation follows from a first-order Taylor expansion of $\frac{1}{x}$ about $\E[\sampledrareg]$. This implies that, letting $\sigma_{g}^2=\Var(Y\mid G=g,R=1)$ and noting that $\overline{Y}_{g1}\indep \overline{Y}_{g'1}$ and $\E[\sampledrareg] = \sampledg \Pr(R=1\mid G=g)$,	
\begin{equation}\label{var-approx-full}
    \Var(\widehat{\mu}_{\mathrm{strat}}) \approx \sum_{g=1}^{L}\frac{\sigma_{g}^2\Pr^2(g\mid R=1)}{\sampledg\Pr(R=1\mid G=g)} =\sum_{g=1}^{L} \frac{\sigma_g^2 \Pr(R=1\mid G=g)\Pr^2(g)}{\sampledg\Pr^2(R=1)}
\end{equation}
The result follows from the following lemma with $A_g = \sigma^2_g\Pr(R=1\mid G=g)\Pr^2(g)/\Pr^2(R=1)$, assuming $A_g>0$ for all $g$.

\begin{lemma}\label{lem-min}
Suppose $A_g> 0$ for all $g$. Then the minimizer of $\sum_{g=1}^{L}\frac{A_g}{\sampledg}$ subject to $\sum_{g=1}^{L}\sampledg=\samplednum$ is
\begin{equation*}
    \sampledg =\frac{ \samplednum\sqrt{A_g}}{\sum_{g=1}^{L}\sqrt{A_g}}
\end{equation*}
\end{lemma}
\begin{proof}[\textbf{Proof of Lemma \ref{lem-min}}] By Lagrange optimization, we need to minimize $f = \sum_{g=1}^{L}\frac{A_g}{\sampledg }+\lambda(\sum_{g=1}^{L}\sampledg - \samplednum)$. Taking the derivative yields equation $\frac{\partial f}{\partial \sampledg }=-\frac{A_g}{\sampledg ^2}+\lambda=0$, which implies $\sampledg =\frac{\sqrt{A_g}}{\sqrt{\lambda}}$. Applying the constraint yields $\samplednum=\sum_{g=1}^{L}\sampledg =\frac{1}{\sqrt{\lambda}}\sum_{g=1}^{L}\sqrt{A_g}$ which implies $\sqrt{\lambda}=\frac{1}{\samplednum}\sum_{g=1}^{L}\sqrt{A_g}$. Also, the second derivatives are $\frac{\partial^2 f}{\partial^2 \sampledg} = 2\frac{A_g}{\sampledg^3}\geq 0$ and $\frac{\partial^2 f}{\partial \sampledg\partial \sampledg'} = 0$, so the Hessian is a diagonal matrix with all non-negative elements and hence is positive definite. This implies the solution is a minimum.
\end{proof}

Since the $\sigma_g^2$ are unknown, we will assume they are homogeneous, in which case the variance-minimizing value of $\sampledg $ is
\begin{equation}\label{eq-optimal_ng_pgversion}
    \sampledg^{\mathrm{strat}}  = n\frac{\sqrt{\Pr(R=1\mid G=g)}\Pr(G=g)}{\sum_{g'} \sqrt{\Pr(R=1\mid G= g')}\Pr(G=g')}
\end{equation}
We assume estimates of $\Pr(G=g),\Pr(R=1\mid G=g)$ are available prior to sampling. In finite sample notation, letting $\Pr(G=g)=\frac{\popoverallg}{\popoverall}$, this is equivalent to equation~\ref{eq-optimal_ng}.

\subsubsection{Choosing target sample sizes in Poisson sampling}\label{sec-opt-pois}

Again consider the stratified estimator $\widehat{\mu}_{\mathrm{strat}}$ but now under our proposed Poisson sampling scheme. Consider $\Pr(R=1\mid S=s,G=g)$ as fixed and focus on how to choose $\targetg$ that minimizes the same approximate variance from Section~\ref{sec-opt}. Under Poisson sampling, $\widehat{\mu}_{\mathrm{strat}}$ is unbiased only if $Y$ satisfies $Y\indep S\mid G,R=1$. However, once we assume the variances $\Var(\overline{Y}_{g1})$ are homogeneous, this has no impact on the variance-minimizing calculation. In Proposition~\ref{prop-pois-counts}, we first derive component quantities that change in the Poisson sampling context relative to stratified sampling.

\begin{proposition}[\textbf{Expected Minority Population Counts in Sample}]\label{prop-pois-counts}
Suppose in a population of $N$ units, each unit is sampled independently with probability $\pi_i=\frac{\targetg}{\probnormgi}\Pr(R_i=1\mid S_i,G_i)$.  Then 
\begin{equation}\label{eq-pR1_gvn_I1_g}
    \Pr(R=1\mid I=1,G=g)= \frac{\sum_{s=1}^{|\cS|} \popoverallgs\Pr(R=1\mid S=s,G=g)^2}{\probnormg}
\end{equation}
\noindent and the expectation of sample minority population count $\sampledrareg$ from statum $g$ is
\begin{align}
\E[\sampledrareg] &= \frac{\targetg}{\probnormg}\sum_{s=1}^{|\cS|}\popoverallgs\Pr(R=1\mid S=s,G=g)^2 = \targetg\Pr(R=1\mid I=1,G=g)
\end{align}
\end{proposition} 

\begin{proof}[\textbf{Proof of Proposition \ref{prop-pois-counts}}]
By design, $R\indep I \mid S,G$ since the sampling only depends on $S,G$. Using this and the Law of Total Probability, 
\small 
\begin{align}
\Pr(R=1,I=1,S=s,G=g) &= \Pr(S=s,G=g)\Pr(I=1\mid S=s,G=g)\Pr(R=1\mid I=1,S=s,G=g)\\
        &=\Pr(S=s,G=g)\Pr(I=1\mid S=s,G=g)\Pr(R=1\mid S=s,G=g)\\
         &= \frac{\popoverallgs}{\popoverall}\frac{\targetg\Pr(R=1\mid S=s,G=g)^2}{\probnormg}\\ 
    \Pr(R=1, I=1,G=g) &=\sum_{s=1}^{|\cS|}\Pr(R=1,I=1,S=s,G=g) \\
    &= \frac{\targetg}{N\probnormg}\sum_{s=1}^{|\cS|} \popoverallgs\Pr(R=1\mid S=s,G=g)^2 \\
\Pr(I=1,G=g)&=\sum_{s=1}^{|\cS|}\Pr(S=s, G=g)\Pr(I=1\mid S=s,G=g) \\
&= \frac{\targetg}{N\probnormg}\sum_{s=1}^{|\cS|}\popoverallgs\Pr(R=1\mid S=s,G=g) \\
&= \frac{\targetg}{\popoverall} \\
\Pr(S=s\mid I=1,G=g) &=\frac{\popoverall}{\targetg}\Pr(S=s,G=g)\frac{\targetg\Pr(R=1\mid S=s,G=g)}{\probnormg} \\
 & = \frac{\popoverall}{\targetg}\frac{\popoverallgs}{\popoverall}\frac{\targetg\Pr(R=1\mid S=s,G=g)}{\probnormg}= \frac{\popoverallgs\Pr(R=1\mid S=s,G=g)}{\probnormg}
\end{align}
\normalsize

Taking the ratio of the 2nd and 3rd equations above results in equation~\ref{eq-pR1_gvn_I1_g}. We then have the following expectations,
\begin{align}
\E[\sampledraresg] &= \E\left[\sum_{i=1}^{\popoverall}1_{(R_i=1,S_i=s,G_i=g)}I_i\right] \\
&=N\Pr(S=s,G=g)\Pr(I=1\mid S=s,G=g)\Pr(R=1\mid I=1, s,G=g)\\
&=\popoverall\frac{\popoverallgs}{\popoverall}\frac{\targetg}{\probnormg}\Pr(R=1\mid S=s,G=g)^2 \\
&=\targetg\frac{\popoverallgs}{\probnormg}\Pr(R=1\mid S=s,G=g)^2 \\
\E[\sampledrareg]&=\sum_{s=1}^{|\cS|}\E[\sampledraresg] = \frac{\targetg}{\probnormg}\sum_{s=1}^{|\cS|}\popoverallgs\Pr(R=1\mid S=s,G=g)^2
\end{align}
\end{proof}

\noindent Next, we derive the main result. 

\begin{proposition}[\textbf{Poisson sampling optimal allocation}]\label{prop-Tg-opt}
Suppose that in a population of $N$ units, each unit is sampled independently with probability $\pi_i=\frac{\targetg}{\probnormgi}\Pr(R_i=1\mid S_i,G_i)$. Suppose the variance of $\widehat{\mu}_{\mathrm{strat}}$ calculated on the resulting sample is approximated using the first-order Taylor approximation $\E\left[\frac{1}{\sampledrareg}\right]\approx \frac{1}{\E[\sampledrareg]}$, assuming $\E[\sampledrareg]>0$. If the variances $\Var(\overline{Y}_{g1})$ are homogeneous and non-zero across $g$, then this approximate variance of $\widehat{\mu}_{\mathrm{strat}}$ is minimized by
\begin{equation}\label{eq-Tg_opt_poststrat-app}
    \frac{\targetg^{\mathrm{pois}}}{\targetoverall} \propto  \frac{\Pr(G=g\mid R=1)}{\sqrt{\Pr(R=1\mid I=1,G=g)}}
\end{equation}
\end{proposition}

\begin{proof}[\textbf{Proof of Proposition \ref{prop-Tg-opt}}]
Consider the estimator $\widehat{\E}[Y\mid R=1] = \sum_{g=1}^{L}\overline{Y}_{g1}\Pr(G=g\mid R=1)$. 
This estimator is now only unbiased if $Y$ satisfies $Y\indep S\mid G,R=1$. This means that 
$\Var(\overline{Y}_{g1})=\Var(Y\mid I=1,G=g,R=1)$ is not necessarily the same as $\sigma_g^2=\Var(Y\mid G=g,R=1)$, but otherwise, it does not affect the variance calculation and we obtain the same form of the variance in equations~\ref{var-approx}~and~\ref{var-approx-full} above, only with a new expression for $\E[\sampledrareg]$ which is derived in Proposition~\ref{prop-pois-counts}. 
\begin{equation}\label{eq-var-pois}
    \Var(\widehat{\mu}_{\text{strat}}) \approx \sum_{g=1}^{L}\frac{\Var(\overline{Y}_{g1})\Pr^2(g\mid R=1)}{\targetg\Pr(R=1\mid I=1,G=g)}
\end{equation}
\noindent We then apply Lemma \ref{lem-min} to obtain that the optimal $\targetg$ is, assuming homogeneity of the $\Var(\overline{Y}_{g1})$ across $g$,
\begin{equation}\label{eq-pois-intermediary}
\targetg^{\mathrm{pois}} =\targetoverall \frac{\Pr(G=g\mid R=1)}{\sqrt{\Pr(R=1\mid I=1,G=g)}}\left(\sum_{g'}\frac{\Pr(G=g'\mid R=1)}{\sqrt{\Pr(R=1\mid G=g',I=1)}}\right)^{-1}
\end{equation}
\end{proof} 

\noindent Next, we relate the Poisson allocation formula to the stratified sampling one.

\begin{proposition}[\textbf{Connection to stratified sampling allocation}]\label{lem-connection} 
If $\Pr(R=1\mid I=1,G=g)=c\Pr(R=1\mid G=g)$ for some $c\in [0,1]$ and all $g$, then $\targetg^{\mathrm{pois}}=\targetg^{\mathrm{strat}}$ where ${\targetg}^{\text{strat}}$ is the allocation under equation~\ref{eq-optimal_ng}.
\end{proposition}

\begin{proof}[\textbf{Proof of Proposition \ref{lem-connection}}]
If $\Pr(R=1\mid I=1,G=g)=c\Pr(R=1\mid G=g)$ for some $c\in [0,1]$ for all $g$, we can instead plug this into equation~\ref{eq-pois-intermediary} in the proof of Proposition~\ref{prop-Tg-opt} to get
\begin{equation}\label{Topt-supplement}
\targetg^{\mathrm{pois}}\propto  \frac{\Pr(G=g\mid R=1)}{\sqrt{c\Pr(R=1\mid G=g)}} \propto  \sqrt{\Pr(R=1\mid G=g)}\Pr(G=g)
\end{equation}
which is the stratified sampling allocation from before.
\end{proof}

Finally, we remark on why the optimal allocation formula \ref{eq-Tg_opt_poststrat-app} makes sense and is more similar to the disproportionate allocation formula than it may initially seem. The numerator of this expression clearly up-weights strata representing a larger proportion of the Jewish population. However, from equation~\ref{eq-pR1_gvn_I1_g}, strata with many large $\Pr(R=1\mid S=s,G=g)$ will have larger denominator $\Pr(R=1\mid I=1,G=g)$, which brings $\targetg^{\mathrm{pois}}$ down. 
The Poisson allocation formula therefore represents a trade-off. If a stratum forms a large proportion of the $R=1$ population (i.e., high $\Pr(G=g\mid R=1)$), we want to sample from it more, but if it also has many large $\Pr(R=1\mid S=s,G=g)$ (call this Scenario 1), this pushes the allocation down because the probability needs to be spread out across more individuals. Conversely, if a stratum has a small $\Pr(G=g\mid R=1)$ with only a few high $\Pr(R=1\mid S=s,G=g)$ and many small ones (call this Scenario 2), the stratum allocation will be up-weighted a bit by the denominator term.  For individuals in this kind of stratum who \textit{do} have a high $\Pr(R=1\mid S=s,G=g)$ probability, we really want to sample them!  Overall, this logic shows that like $\sampledg^{\mathrm{strat}}$, the optimal Poisson allocation is related to $\Pr(R=1\mid G=g)$. Scenario 1 could happen if $\Pr(R=1\mid G=g)$ is large while Scenario 2 could happen if $\Pr(R=1\mid G=g)$ is small.

\subsubsection{Accounting for filtering in Poisson target calculations}\label{app-filtering-pois}

Next, we discuss how to account for filtering in the sampling probability and allocation calculations. Define {\it filtering function} $H(s)$, which equals $1$ if surname $s$ is kept and $0$ if a surname is to be filtered out. We define filter-updated sampling probabilities
\begin{equation}
\pi_i^* := \frac{\targetgi}{\probnormgistar}H(S_i)\Pr(R_i=1\mid S_i,G_i) \quad\quad\quad \probnormgstar := \sum_{i:G_i=g}H(S_i)\Pr(R_i=1\mid S_i, G_i)
\end{equation}
\noindent where $H_i=H(S_i)$ indicates whether unit $i$ is kept. Filtering alters sampling frame counts, replacing $\popoverallg$ by a post-filtering count $\popoverallgstar = \sum_{i=1}^{N}\ind\{G_i=g,H_i=1\}$, and other key quantities.

\begin{proposition}[\textbf{Poisson sampling optimal allocation with filtering}]\label{prop-filtering}
Under a Poisson sampling scheme with filtering function $H(s)$ and sampling probabilities $\pi_i^*$:
\begin{align}
    \Pr(G=g\mid R=1,H=1)&= \frac{\probnormgstar}{\sum_g\probnormgstar}
    \\
   \Pr(R=1\mid I=1, G=g,H=1) &=\frac{\sum_{s:H(s)=1}\popoverallgsstar\Pr(R=1\mid S=s,G=g)^2}{\sum_{s:H(s)=1}\popoverallgsstar \Pr(R=1\mid S=s,G=g)}
\end{align}    
\end{proposition}
These formulas can be used in the allocation formula for $\targetg^{\mathrm{pois}}$ to obtain an allocation that accounts for filtering. This results in
\begin{equation}\label{eq-pois-filterupdate}
\frac{\targetg^{\mathrm{pois}}}{\targetoverall} \propto  \frac{\Pr(G=g\mid H=1, R=1)}{\sqrt{\Pr(R=1\mid I=1,H=1,G=g)}} \propto \frac{\probnormgstar^{3/2}}{\sqrt{\sum_{s:H(s)=1}\popoverallgs \Pr(R=1\mid S=s,G=g)^2}}
\end{equation}

\begin{proof}[\textbf{Proof of Proposition \ref{prop-filtering}}] 
Because the filter variable $H$ is a deterministic function of $S$, we must have $R\indep H\mid S,G$ and $R\indep H\mid S,G,I$. We then have:
\begin{align}
    \Pr(G=g, R=1\mid H=1) &= \sum_{s=1}^{|\cS|}\Pr(S=s,G=g,R=1\mid H=1) \nonumber \\
    &=\sum_{s=1}^{|\cS|}\Pr(S=s,G=g\mid H=1)\Pr(R=1\mid S=s,G=g,H=1)\nonumber\\
    &= \sum_{s=1}^{|\cS|}\frac{\popoverallgsstar}{\popoverallstar}\Pr(R=1\mid S=s,G=g) \nonumber \\
    &= \frac{\probnormgstar}{\popoverallstar}\label{eq-important-filtering}  
\end{align}
Normalizing this over $g$ results in the expression for $\Pr(G=g\mid R=1,H=1)$. Similarly, we can get the probability conditional on $I=1$ and $H=1$ using:
\small
    \begin{align}
    &\Pr(R=1, I=1,G=g\mid  H=1) = \sum_{s=1}^{|\cS|}\Pr(S=s,G=g,R=1,I=1\mid H=1)\\ 
        &=\sum_{s=1}^{|\cS|} \Pr(S=s,G=g\mid H=1)\Pr(I=1\mid S=s,G=g,H=1)\Pr(R=1\mid S=s,I=1,G=g,H=1)\\
        &=\frac{\targetg}{\probnormgstar\popoverallstar} \sum_{s:H(s)=1}\popoverallgsstar\Pr(R=1\mid S=s,G=g)^2
    \end{align}
and
    \begin{align}
    \Pr(I=1,G=g\mid  H=1) 
        &=\frac{\targetg}{\probnormgstar\popoverallstar} \sum_{s:H(s)=1}\popoverallgsstar\Pr(R=1\mid S=s,G=g)\\
    \Pr(R=1\mid I=1,G=g,H=1)&=\frac{\sum_{s:H(s)=1}\popoverallgsstar\Pr(R=1\mid S=s,G=g)^2}{\sum_{s:H(s)=1}\popoverallgsstar \Pr(R=1\mid S=s,G=g)}
    \end{align}
\end{proof}
\normalsize

\subsubsection{Accounting for filtering in stratified sampling}\label{app-filtering-alone}

By normalizing equation~\ref{eq-important-filtering} in the proof above using $\Pr(G=g\mid H=1)=\frac{\popoverallgstar}{\popoverallstar}$, we obtain
\begin{equation}\label{eq-est-filter-impact}
    \Pr(R=1\mid G=g,H=1) = \frac{\probnormgstar}{\popoverallgstar}
\end{equation}
However, without access to $\Pr(R=1\mid S=s,G=g)$, we cannot calculate this post-filtering formula for use in a post-filtering version of the disproportionate allocation formula for $\sampledg^{\mathrm{strat}}$ from equation~\ref{eq-optimal_ng_pgversion}. A simpler filtering adjustment that uses only the updated $\popoverallgstar$ is to calculate the disproportionate allocation formula (equation~\ref{eq-optimal_ng}) using updated
\begin{equation}\label{filter-best-case}
\Pr^{\text{best.case}}(R=1\mid G=g,H=1)=\frac{\popoverallg \Pr(R=1\mid G=g)}{\popoverallgstar}
\end{equation}

The numerator of this formula is exactly correct only in the best-case scenario where filtering removes no minority population members, but if this is mostly true under conservative filtering, it may be an acceptable approximation that up-weights strata where filtering indicates we have been able to remove more people with surnames not indicative of belonging to the minority population.

\subsubsection{Sample size expectations and variances under the Poisson sampling}\label{sample-size-variance}

For all calculations in this section, we implicitly condition on (treat as fixed) the observed stratum memberships and surnames in the sampling frame.
By construction, $\E[\sampledg]=\targetg$ and $\E[\samplednum]=\targetoverall$. The variance expression for $\sampledg$ in main text equation~\ref{eq-sampesizevar} follows simply from the variance of a Bernoulli $I_i$ and the fact that the Poisson sampling is independent across units:
\begin{align}
    \Var(\sampledg) &= \Var\left(\sum_{i:G_i=g}I_i\right)= \sum_{i:G_i=g}\Var(I_i) = \sum_{i:G_i=g}\pi_i(1-\pi_i) \\
    &= \targetg - \sum_{i:G_i=g}\pi_i^2  \\
    &= \targetg - \frac{\targetg^2}{\probnormg^2}\sum_{i:G_i=g}\Pr(R_i=1\mid S_i,G_i=g)^2 \label{eq-linetoflag}
\end{align}
where the second-to-last expression follows from the fact that $\sum_{i:G_i=g}\pi_i = \targetg$ by construction. The independence in the Poisson sampling also means the overall variance of $n$ is the sum of the variances for each $g$:
\begin{equation}
  \Var(n) = \sum_{g=1}^{L} \targetg - \sum_{g=1}^{L}\sum_{i:G_i=g}\pi_i^2 =  \targetoverall- \sum_{i=1}^{\popoverall}\pi_i^2 
\end{equation}

\medskip 

\noindent For the minority population counts, we have
   \begin{align}
\E[\sampledrare ] &= \sum_{i=1}^{\popoverall} \E[I_i R_i]= \sum_{i=1}^{\popoverall}
    P(I_i = 1 \mid S_i,G_i)\, P(R_i = 1 \mid S_i,G_i) \\
&= \sum_{i=1}^{\popoverall} \pi_i P(R_i = 1 \mid S_i,G_i) \\
&= \sum_{g=1}^{L}\sum_{s=1}^{|\cS|} \popoverallgs
    \left(\frac{\targetg \Pr(R=1\mid S=s,G=g)}{\probnormg}\right) \Pr(R=1\mid S=s,G=g) \\
&= \sum_{g=1}^{L} \frac{\targetg}{\probnormg} \sum_{s=1}^{|\cS|} \popoverallgs \Pr(R=1\mid S=s,G=g)^{2}.\label{eq-expectedsampledrarecount}
\end{align}
We also calculate the variance of the number of minority population members sampled in each state and overall.
\begin{align}
    &\Var(\sampledrareg) = \sum_{i:G_i=g}\Var(I_iR_i) = \sum_{i:G_i=g}\pi_i\Pr(R_i=1\mid S_i,G_i)(1-\pi_i\Pr(R_i=1\mid S_i,G_i))  \\
        &= \frac{\targetg}{\probnormg}\sum_{s=1}^{|\cS|} \popoverallgs\Pr(R=1\mid S=s,G=g)^2 - \frac{\targetg^2}{\probnormg^2}\sum_{s=1}^{|\cS|} \popoverallgs\Pr(R=1\mid S=s,G=g)^4
\end{align}
\begin{align}
    \Var(\sampledrare) 
        &= \sum_{g=1}^{L}\left(\frac{\targetg}{\probnormg}\sum_{s=1}^{|\cS|} \popoverallgs\Pr(R=1\mid S=s,G=g)^2 - \frac{\targetg^2}{\probnormg^2}\sum_{s=1}^{|\cS|} \popoverallgs\Pr(R=1\mid S=s,G=g)^4 \right) 
\end{align}

\subsubsection{Role of BISG probability precision}\label{app-precisionanalysis}

This section provides supporting calculations for main paper Section~\ref{sec-estquality}, equation~\ref{eq-expectedyieldcasestudy}. Consider a simplified scenario in which, in each geography $g$, $a\%$ of surnames have BISG probability $\Pr(R=1\mid S=s,G=g)=p$ and $(1-a)\%$ of surnames have BISG probability $\Pr(R=1\mid S=s,G=g)=q$. We assume $q\leq p$. Applying the formula above, we then have,
\begin{align}
    \E[\sampledrare] &= \sum_{g=1}^{L}\targetg \frac{\popoverallg ap^2 + \popoverallg (1-a)q^2}{\popoverallg ap + \popoverallg (1-a)q}= p\targetoverall\left(\frac{1 + \frac{(1-a)}{a}\frac{q^2}{p^2}}{1 + \frac{(1-a)}{a}\frac{q}{p}}\right)
\end{align}
If $q\leq p$, then $\frac{q^2}{p^2}\leq \frac{q}{p}$ and hence $\E[\sampledrare]\leq p\targetoverall$. Dividing through by $(1-a)q$ instead gives the other formula in equation~\ref{eq-expectedyieldcasestudy}. By similar reasoning, $\E[\sampledrare]\geq q\targetoverall$. 

\subsubsection{Sensitivity to error in BISG probabilities}\label{app-sensitivitytoerror}

This section provides supporting calculations for main paper Section~\ref{sec-estquality}.  In the calculations of Section~\ref{sample-size-variance}, we use the theoretical form of sampling probability $\pi_i$, which is a function of $\Pr(R_i=1\mid S_i,G_i)$. When, in practice, $\pi_i$'s are calculated from $\widehat{\Pr}(R=1\mid S=s,G=g)$, the $\Pr(R=1\mid S=s,G=g)^2$ 
and $\Pr(R=1\mid S=s,G=g)^4$ appearing in the means and variances of $\sampledrareg$ and $\sampledrare$ become
\begin{equation}\label{eq-varexp2}
    \widehat{\Pr}(R=1\mid S=s,G=g)\Pr(R=1\mid S=s,G=g)
\end{equation}
and
\begin{equation}\label{eq-varexp2}
    \widehat{\Pr}(R=1\mid S=s,G=g)^2\Pr(R=1\mid S=s,G=g)^2
\end{equation}
respectively where the right term concerns the true rate of $R_i=1$ units among $S_i=s,G_i=g$ and the left term is used in constructing the sampling probability used. In calculations for the mean and variance of $\sampledg$ and $\samplednum$, all $\Pr(R_i=1\mid S_i,G_i=g)$ become $\widehat{\Pr}(R_i=1\mid S_i,G_i=g)$. 

\paragraph{Consequences of additive mean zero error:}

We briefly explore a simple but illustrative example of what happens to the expected yield and its variance if there is error in the BISG probabilities. Suppose
\begin{equation}
    \widehat{\Pr}(R=1\mid S=s,G=g) = \Pr(R=1\mid S=s,G=g) + \epsilon_{s,g}
\end{equation}
where we assume the true probability is non-zero and where $\boldsymbol{\epsilon} := (\epsilon_{s,g})_{s\in \cS,g\in \cG}$ are drawn at random from some distribution such that $\E[\epsilon_{s,g} \mid G=g]=0$ in each geography $g$. We let this error be mean zero because the effect of any systematic shift by a constant $c\neq 0$ is reduced by the normalization of the final sampling probabilities. 

The mean zero has the theoretically useful consequence that the normalizing sum under this error is in expectation the true one.
\begin{equation}
    \E[\widehat{\pi}(g)] := \E\left[\sum_{i:G_i=g} \widehat{\Pr}(R=1\mid S=s,G=g)\right]=\E\left[\sum_{i:G_i=g} \Pr(R=1\mid S=s,G=g)\right] = \pi(g)
\end{equation}
We will therefore continue to use $\pi(g)$ in the expressions below, making these expressions first-order approximations of the actual expectation.

The primary impact of this additive error lies in whether or not $\epsilon_{s,g}$ is correlated with the true value $\Pr(R=1\mid S=s,G=g)$. Let $\E_{\text{oracle}}[\sampledrareg]$ be the expected minority population count from equation~\ref{eq-expectedsampledrarecount} when sampling with the true BISG probabilities and consider the expectation when sampling using the BISG probabilities with additive error. Then, conditional on the $\epsilon_{s,g}$, we have
\begin{align}
  \E[\sampledrareg \mid \boldsymbol{\epsilon} ] &\approx  \E_{\text{oracle}}[\sampledrareg] + \sum_{g=1}^{L} \frac{\targetg}{\probnormg} \sum_{s=1}^{|\cS|} \popoverallgs \epsilon_{s,g}\Pr(R=1\mid S=s,G=g)
\end{align}
This expression indicates that if errors tend to be negatively correlated with the true BISG probabilities, the yield decreases. If the errors are uncorrelated, the yield in expectation does not change.  If the errors are positively correlated, the expected yield could increase. However, the fact that the BISG probabilities must still be in $[0,1]$ limits the extent of positive correlation possible and intuitively, boosting highest-BISG probability units even higher can increase the variance of the final sample estimator if those highest-BISG units differ systematically from others. 

For the variance, we get
\begin{align}
 & \Var(\sampledrare \mid \boldsymbol{\epsilon} ) \\
        &\approx   \Var_{\text{oracle}}(\sampledrare)  + \sum_{g=1}^{L}\left(\frac{\targetg}{\probnormg}\sum_{s=1}^{|\cS|} \popoverallgs\epsilon_{s,g}\Pr(R=1\mid S=s,G=g)\right)- \sum_{g=1}^{L}\left(\frac{\targetg^2}{\probnormg^2}\sum_{s=1}^{|\cS|} \popoverallgs\xi_{s,g}\right) \nonumber \\ 
&\xi_{s,g}:=2\epsilon_{s,g}\Pr(R=1\mid S=s,G=g)^3+\epsilon_{s,g}^2\Pr(R=1\mid S=s,G=g)^2 \nonumber 
\end{align}
and the total variance is
\begin{equation}
\Var(\sampledrare) = \E[\Var(\sampledrare \mid \boldsymbol{\epsilon} )] + \Var(\E[\sampledrare \mid \boldsymbol{\epsilon}])
\end{equation}
where the second term is non-negative.

The overall effect of the error on the variance depends on various factors. The first term of $\xi_{s,g}$ depends on the covariance between the error and BISG probabilities and could contribute to increasing the variance. The second is non-negative and hence decreases the variance. If errors are uncorrelated with the BISG probabilities, the first term will be zero in expectation so that the second part of $\xi_{s,g}$ works to decrease $\E[\Var(\sampledrare \mid \boldsymbol{\epsilon} )]$, but this might still be outweighed by the increase from $\Var(\E[\sampledrare \mid \boldsymbol{\epsilon}])$. Moreover, a variance decrease does not always mean improvement here. If estimated inclusion probabilities become more unequal, this will tend to decrease the variance of the yield but is also compatible with lower expected yield and increased variance of the final estimator.

\subsubsection{Details of Table \ref{tab-samplingsuccerates} Calculations}\label{sec-tabledetails} 

Table \ref{tab-samplingsuccerates}  illustrates the limited expected success of some existing sampling strategies and compares it to the estimated expected success of our proposed method. Each row is an estimate of $\E[\sampledrare]/\E[\samplednum]$. For stratified sampling, $\samplednum$ is fixed to $50,000$ and the Poisson sampling satisfies $\E[\samplednum]=50,000$ by design. The expected number of Jewish people sampled was calculated in each case as described below. We then estimate the fraction sampled that will be Jewish by dividing this by $\E[n]$. We describe the rows of the table in order of complexity.

\begin{enumerate}

   \item Row 2 is based on taking a simple random sample of size $\samplednum$ form the voter file of size $\popoverall$. The expected fraction of minority population members sampled is then $\popoverallr/\popoverall$, which we estimate using
   \begin{equation}
\E\left[\frac{1}{\samplednum}\sum_{i=1}^{\popoverall}I_iR_i\right]
=
\frac{1}{\samplednum}\sum_{g=1}^{L} \sum_{i:G_i=g}\E[I_iR_i]
=
\frac{1}{\samplednum}\sum_{g=1}^{L} \sum_{i:G_i=g} R_i \frac{\samplednum}{\popoverall}
=
\frac{1}{N}\sum_{g=1}^{L} \popoverallg \Pr(R=1\mid G=g)
\end{equation}
The resulting value of 1.9\% is slightly lower than the reported 2.4\% of U.S. adults who identify as Jewish in the 2021 Pew research report \cite{pew_jewish_population_2021}. This could come from differences and error in the AJPP and Pew estimates or from the fact that the U.S. adult population and the voter file are not exactly the same populations. Using standard SRS results given in \cite{Lohr1998}, the variance of the fraction of minority population units sampled is
\begin{equation}\label{eq-srs-var}
\Var\left(\frac{\sampledrare}{\samplednum}\right) = \Var\left(\frac{1}{\samplednum}\sum_{i=1}^{\popoverall}I_iR_i\right)= \frac{1}{\samplednum}\frac{\popoverallr}{\popoverall}\left(1-\frac{\popoverallr}{\popoverall}\right)\left(\frac{\popoverall-n}{\popoverall-1}\right)
   \end{equation}
   \item Row 3 is calculated as $\sum_{g=1}^{L} \sampledg\Pr(R=1\mid G=g)$ where $\sampledg$ is allocated according to the optimal disproportionate allocation formula given in Section~\ref{sec-samplemethod}, equation~\ref{eq-optimal_ng}. Here the counts $\popoverallg$ come from the voter file and $\Pr(R=1\mid G=g)$ is estimated from AJPP data. The variance is

   \begin{equation}
        \Var\left(\frac{\sampledrare}{\samplednum}\right) =  \frac{1}{\samplednum^2}\sum_{g=1}^{L}\sampledg^2 \Var\left(\frac{\sampledrareg}{\sampledg}\right) 
   \end{equation}
   where each  $\Var\left(\frac{\sampledrareg}{\sampledg}\right)$ is as in equation~\ref{eq-srs-var} with $\popoverallgr, \popoverallg,\sampledg$ in place of $\popoverallr,\popoverall,\samplednum$. 

    \item Row 4 is calculated as $\sum_{g=1}^{L} \sampledg^*\Pr(R=1\mid G=g,H=1)$ with the updated $\Pr(R=1\mid G=g,H=1)$ from equation~\ref{eq-est-filter-impact} in Section~\ref{app-filtering-alone}. However, in order to create a realistic filtering-only allocation method that does not depend on having estimated $\Pr(R=1\mid S=s,G=g)$, we calculate the  $\targetg^*$ using equation~\ref{eq-optimal_ng} with $\Pr_{\text{best.case}}(R=1\mid G=g,H=1)$ and $\popoverallgstar$ as described in Section~\ref{app-filtering-alone}. The variance formula has the same form as in row 2. 

\item Row 1 uses a plug-in estimator  $\E[\sampledrare]/\E[\sampledrare]$ under Poisson sampling as calculated in Section~\ref{sample-size-variance}. This is a first-order Taylor approximation of the expectation $\E[\frac{\sampledrare}{\samplednum}]$. It also relies on our estimated surname, first name, and geography based probabilities of being Jewish. Here we used the $\targetg^{\mathrm{pois}}$ allocation with filtering adjustment described in Section~\ref{app-filtering-pois}. 
For the variance of the sampling fraction, we approximate the variance under a first-order Taylor approximation that will usually be conservative because we drop a $-\Cov(\samplednum, \sampledrare)$ term that will usually be negative.
\begin{equation}
 \Var\left(\frac{\sampledrare}{\samplednum}\right) \approx   \frac{\Var(\sampledrare)}{\targetoverall^2} + \frac{\E[\sampledrare]^2}{\targetoverall^4}\Var(\samplednum)
\end{equation}
The formulas for $\Var(\samplednum)$ and $\Var(\sampledrare)$ are in Section~\ref{sample-size-variance}.

\paragraph{Sensitivity to filtering.}

For the row 1 calculation, we also conducted a sensitivity analysis to examine the impact of setting the estimated BISG probabilities for unobserved surnames to a small positive value $\epsilon$ (used in calculating sampling probabilities $\pi_i$) while the true BISG probabilities for these surnames are denoted by $\delta$. In the formulas for $\E[\sampledrare]$ and $\Var(\sampledrare)$, terms like $\Pr(R=1\mid S=s,G=g)^2$ then become $\epsilon\delta$. First, we set both $\epsilon = \delta =  10^{-7}$, which is of the order of magnitude of the average minimum estimated probability across states. This assigns a small positive probability to  $\approx$ 94 million additional people and assumes their 
true probability of being Jewish is also tiny. The result is no meaningful difference in the expected yield compared to filtering (58.0\% instead of 58.6\%).
Second, we set the estimated BISG probabilities for these unobserved surnames to $\epsilon=0.01$ while assuming the true BISG probabilities are still $\delta=10^{-7}$. We then estimate that the expected proportion of people sampled who are Jewish would decrease to 45\%. Overall, if one assigns a positive probability to unobserved surnames, it is mainly important not to set this probability too high.

\end{enumerate}

\subsection{Bayesian Hierarchical Model and Sampler}\label{sec-sampler}

We discuss the details of fitting the model from Section~\ref{sec-learning-surname-dist} for learning $\Pr(S\mid G,R=1)$ from the Jewish obituary data. As a reminder, the model is
\begin{align}
\mgvec &\sim \text{Multinomial}(\mg, \thetagvec), &\text{ where  }& \theta_{gs}:=\Pr(S=s\mid G=g,R=1) \text{ for } s=1,...,|\cS|\\
\thetagvec &\overset{iid}{\sim} \text{Dirichlet}(\eta\alphavec) & \text{ where  } & \alphavec \in \mathbb{R}^{|\cS|} \text{ with } \sum_{s=1}^{|\cS|}\alpha_s = 1 \text{ and } \eta > 0\\ 
\alphavec &\sim \text{Dirichlet}(\gammavec) &\text{ where } & \gammavec \in \mathbb{R}_{+}^{|\cS|}\\
\eta &\sim  \pi & \text{ where } & \pi \text{ is a prior distribution},
\end{align}
where $\mgvec = (m_{g1},...,m_{g|\mathcal{S}|})$ with $\mg = \sum_{s\in \mathcal{S}}\mgs$ and we set the elements of hyperparameter $\gammavec$ to $\gammavec_s = \ms+1$ (to avoid Dirichlet parameters $\leq 1$). We exclude surnames in the voter file that do not occur in the obituaries. For the prior on $\eta$, we used a diffuse gamma distribution $\eta\sim \text{Gamma}\left(\text{shape}=1,\ \text{rate} = \frac{1}{100}\right)$, which worked reasonably well in simulations.

\subsubsection{Posterior calculations}

Let $\Data=\{\mgvec\}_{g\in \mathcal{G}}$ represent the observed data. Using $\sum_{s=1}^{|\cS|}\alpha_s=1$, the full posterior is
\begin{equation}\label{eq-full-posterior}
\pi(\thetavec,\alphavec,\eta\mid \Data, \gammavec) \propto  \pi(\eta) \left(\prod_{s=1}^{|\mathcal{S}|}\alpha_s^{\gammavec_s-1}\right)\left(\frac{\Gamma(\eta)}{\prod_{s}\Gamma(\eta\alpha_s)}\right)^{L}\left( \prod_{g=1}^{L}\prod_{s=1}^{|\mathcal{S}|}\theta_{gs}^{\mgs+\eta\alpha_s-1}\right)
\end{equation}
and hence the conditional posterior of $\theta$ given $\alphavec, \eta$ is a Dirichlet with updated parameters $\mgs+\eta\alpha_s$ and conditional and unconditional posterior means
\begin{align}
\E[\theta_{gs} \mid \alphavec,\eta,\Data,\gammavec]&=\frac{\mgs+\eta\alpha_s}{\mg+ \eta}  = (1-\rho_g) \frac{\mgs}{\mg} + \rho_g \alpha_s, \hspace{.5cm} \rho_g=\frac{\eta}{\mg+\eta}, \label{eq-condpostmean}\\ 
\E[\theta_{gs} |\Data,\gammavec]&=\E\left[\frac{\mgs+\eta\alpha_s}{\mg+ \eta}\right] \label{eq-postmean}
\end{align}
which can be estimated using averages over posterior samples. The marginal posterior of $\alphavec, \eta$ is,
\begin{align}
&\pi(\alphavec,\eta\mid \Data, \gammavec) \nonumber \\
&\propto \pi(\eta) \left(\prod_{s=1}^{|\mathcal{S}|}\alpha_s^{\gammavec_s-1}\right) \left(\frac{\Gamma(\eta)}{\prod_{s}\Gamma(\eta\alpha_s)}\right)^{L} \int \left(  \prod_{g=1}^{L}\prod_{s=1}^{|\mathcal{S}|}\theta_{gs}^{\mgs+\eta\alpha_s-1}\right)d\thetavec\\ 
&\propto  \pi(\eta) \left(\prod_{s=1}^{|\mathcal{S}|}\alpha_s^{\gammavec_s -1}\right) \left(\frac{\Gamma(\eta)}{\prod_{s}\Gamma(\eta\alpha_s)}\right)^{L} \prod_{g=1}^{L} \frac{\prod_{s=1}^{|\mathcal{S}|}\Gamma(\mgs+\eta\alpha_s)}{\Gamma(\sum_{s=1}^{|\cS|}\mgs+\eta\alpha_s)}\\
&\propto  \pi(\eta)\Gamma(\eta)^L \left(\prod_{s=1}^{|\mathcal{S}|}\alpha_s^{\gammavec_s -1}\right) \left(\prod_{g}\prod_{s}\frac{\Gamma(\mgs+\eta\alpha_s)}{\Gamma(\eta\alpha_s)}\right) \prod_{g=1}^{L} \frac{1}{\Gamma(\mg+ \eta)}\\
&\propto  \pi(\eta)\Gamma(\eta)^L \left(\prod_{s=1}^{|\mathcal{S}|}\alpha_s^{\gammavec_s -1}\right) \left(\prod_{g}\prod_{s} \prod_{k=1}^{\mgs}(\eta\alpha_s + \mgs - k)\right) \prod_{g=1}^{L} \frac{1}{\Gamma(\mg+ \eta)}, \label{eq-finalmargpost}
\end{align}
where by convention, we let a product over $k=1,...,0$ be $1$ (this occurs where $\mgs=0$).  Finally, we derive the conditional posterior of $\alpha_i,\alpha_j$ given the data, $\eta$, and other $\alphavec$ elements (denoted $\alphavec_{-ij}$). First, note that the conditional marginal prior for $\alpha_i,\alpha_j$ is a scaled Dirichlet.
\begin{equation}
    \pi(\alpha_i,\alpha_j\mid \alphavec_{-ij},\gammavec) =  \text{Dirichlet}(\gammavec_i,\gammavec_j)*(1-\sum_{l\notin \{i,j\}}\alpha_l)
\end{equation}
which follows from the following lemma with $I=\{i,j\}$.

\begin{lemma}\label{lem-scale-dir}
Let $\alphavec \in \mathbb{R}^k$ be a Dirichlet($\gammavec$) random variable and let $I\subset\{1,...,k\}$. Let $S=\sum_{i\notin I}\alpha_i$. Then $\alpha_{I}\mid \alpha_{-I}\sim (1-S)\mathrm{Dirichlet}(\gammavec_{I})$ where $\sim$ denotes equality in distribution.
\end{lemma}

\begin{proof}[\textbf{Proof of Lemma \ref{lem-scale-dir}}]
Let $G_i \simind \mathrm{Gamma}(\gammavec_i)$ and let $T=\sum_{i=1}^{k}G_i$ and $\alpha_i=G_i/T$. Then $\alphavec \sim \text{Dirichlet}(\gammavec)$. 
We also have $S=\frac{\sum_{i\notin I}G_i}{T}$ and $1-S =\frac{\sum_{i\in I}G_i}{T}$ and hence
\begin{equation}\label{eq-dir-prf}
 \alpha^*_i :=   \frac{\alpha_i}{1-S} = \frac{G_i/T}{\sum_{l\in I}G_l/T} = \frac{G_i}{\sum_{l\in I}G_l} \text{ for any }i\in I \ \ \Longrightarrow \ \ \alpha_I^*\sim  \text{Dirichlet}(\gammavec_I)
\end{equation}
If $\alpha^*_I \indep (1-S)$, then we can multiply both sides of \ref{eq-dir-prf} by $(1-S)$ to obtain the result. 
To show this, note that by a standard Dirichlet-Gamma independence result, $\alpha^*_I \indep \sum_{l\in I}G_l$ and because both are entirely independent of $\sum_{l\notin I}G_l$, we have  $\alpha_I^*\indep (\sum_{l\in I}G_l,\sum_{l\notin I}G_l)$. Since $S$ and $1-S$ are functions of the right side of this independence, the necessary result holds. See also \cite{lin2016dirichlet} for a similar proof.
\end{proof}
\noindent Using the scaled Dirichlet result and prior separability of $(\alphavec, \eta)$, 
\begin{align}
& \pi(\alpha_i,\alpha_j\mid \alphavec_{-ij},\eta,\Data,\gammavec)\\
&\propto \pi(\alpha_i,\alpha_j\mid \alphavec_{-ij},\eta,\gammavec)\pi(\Data\mid \alphavec, \eta,\gammavec)\\
&\propto \left(\text{Dirichlet}(\gammavec_i,\gammavec_j)(1-\sum_{l\notin \{i,j\}}\alpha_l)\right)\Gamma(\eta)  \left(\prod_{g}\prod_{s} \prod_{k=1}^{\mgs}(\eta\alpha_s + \mgs - k)\right) \prod_{g=1}^{L} \frac{1}{\Gamma(\mg+ \eta)}\\
&\propto \alpha_i^{\gammavec_i-1}\alpha_j^{\gammavec_j-1}\Gamma(\eta) \left(\prod_{g}\prod_{s} \prod_{k=1}^{\mgs}(\eta\alpha_s + \mgs - k)\right) \prod_{g=1}^{L} \frac{1}{\Gamma(\mg+ \eta)}.
\end{align}

\subsubsection{Posterior sampling}

We focus on sampling from the marginal posterior for $(\alphavec, \eta)$, after which the posterior mean of $\theta_{gs}$ can be estimated by averaging using equation~\ref{eq-postmean}. Sampling for the higher-dimensional $\theta_{gs}$ in addition to the already high dimensional $\alpha$ is slow, and it has been shown that collapsing to a marginal can be more efficient in terms of the covariance structure of the resulting Markov chain \cite{liu1994collapsed, Liu_Wong_Kong_1994}. 

We use a Metropolis-within-Gibbs sampler  [See Chapter 10.3 of \cite{RobertChristian1999MCSM}] using the conditional posterior for $\alpha_i,\alpha_j$ derived above. This algorithm cycles through the conditional distributions for pairs of $\alpha_i,\alpha_j$ as one would in a Gibbs sampler but uses a Metropolis-Hastings update at each step rather than sampling from the conditional distribution directly. To achieve irreducibility of the Markov chain, the $(i,j)$ pairings cycled over are randomly selected at each step. 
Specifically, let $L=(l_1,...,l_d)$ be an ordered partition of $\{1,...,|\mathcal{S}|\}$ into $d$ pairs of indices of the form $l =(i,j)$. If $|\cS|$ is odd, let one index appear in two pairs. At each step, we randomly pick a new set of $\alpha_i,\alpha_j$ pairings. We only add to the final chain after each full cycle over pairs and the $\eta$ update. Algorithm 1 gives the procedure. 

\newpage 
\vspace{1cm}
\begingroup
\linespread{1}\selectfont
\par\medskip
\hrule
\smallskip
\noindent\textbf{Algorithm 1: Metropolis-within-Gibbs Sampler}
\par\smallskip\hrule
\begin{description}
  \item[] \textbf{Initialize}
  $\alpha^{(0)}, \eta^{(0)}, L^{(0)}$ vector of $d$ pairings
  \item [] \textbf{Set} $\gammavec$ with $\gammavec_s=\ms+1$;
  \textbf{set} $B \in \mathbb{N}$.
  \item[] \textbf{for} $t \gets 1$ \textbf{to} $B$
  \textbf{do}
  \begin{description}
    \item[]
    $(\alphavec,\eta) \gets
    (\alpha^{(t-1)},\eta^{(t-1)})$;
    $L \gets L^{(t-1)}$.

    \item[] \textbf{for} $r \gets 1$ \textbf{to} $d$
    \textbf{do}
    \begin{description}
      \item[]
      $i \gets l_r[1]$, $j \gets l_r[2]$.
      \hfill \textit{($\alpha$-proposal)}
      \item[]
      Sample $(\alpha'_i,\alpha'_j)
      \sim q(\cdot \mid
      \alpha_i,\alpha_j,\alphavec_{-ij},\eta)$
      and compute
      \[
        p_r = \min\!\left(1,\;
        \frac{
          \pi(\alpha'_i,\alpha'_j \mid
            \alphavec_{-ij},\eta,\Data,\gammavec)\,
          q(\alpha_i,\alpha_j \mid
            \alpha'_i,\alpha'_j,\alphavec_{-ij},\eta)
        }{
          \pi(\alpha_i,\alpha_j \mid
            \alphavec_{-ij},\eta,\Data,\gammavec)\,
          q(\alpha'_i,\alpha'_j \mid
            \alpha_i,\alpha_j,\alphavec_{-ij},\eta)
        }
        \right).
      \]

      \item[]
      With probability $p_r$, set
      $\alpha_i \gets \alpha'_i$,
      $\alpha_j \gets \alpha'_j$.
    \end{description}

    \item[] \textbf{end for}

    \item[]
    Sample $\eta' \sim q(\cdot \mid \eta,\alpha)$
    and compute \hfill \textit{($\eta$-proposal)}
    \[
      p_\eta = \min\left(1,\;
      \frac{
        \pi(\alphavec,\eta' \mid \Data,\gammavec)\,
        q(\eta \mid \eta',\alpha)
      }{
        \pi(\alphavec,\eta \mid \Data,\gammavec)\,
        q(\eta' \mid \eta,\alpha)
      }
      \right).
    \]

    \item[]
    With probability $p_\eta$, set $\eta \gets \eta'$.

    \item[]
    $(\alpha^{(t)},\eta^{(t)}) \gets (\alphavec,\eta)$.
    \hfill \textit{(Add to chain)}

    \item[]
    Randomly permute $\{1,\dots,|\mathcal{S}|\}$ and form
    $L^{(t)}$ as adjacent index pairs.
    \hfill \textit{($L$-update)}
  \end{description}

  \item[] \textbf{end for}
\end{description}
\hrule
\medskip
\par
\endgroup
\vspace{1cm}
\noindent The posterior ratios in $p_r$ and $p_\eta$ in Algorithm 1 are:
\begin{align}
\frac{\pi(\alpha'_i,\alpha_j'\mid \alphavec_{-ij},\eta,\Data, \gammavec)}{\pi(\alpha_i,\alpha_j\mid \alphavec_{-ij},\eta,\Data, \gammavec)} &=  \left(\frac{\alpha'_i}{\alpha_i}\right)^{\gammavec_{i}-1}\left(\frac{\alpha'_j}{\alpha_j}\right)^{\gammavec_{j}-1}   \frac{ \left(\prod_{g}\prod_{s\in\{i,j\}} \prod_{k=1}^{\mgs}(\eta\alpha_s'+\mgs-k)\right) }{ \left(\prod_{g}\prod_{s\in \{i,j\}} \prod_{k=1}^{\mgs}(\eta\alpha_s + \mgs - k)\right)} 
\end{align}
\begin{align}
\frac{\pi(\alphavec, \eta'\mid \Data,\gammavec)}{\pi(\alphavec,\eta\mid \Data,\gammavec)} &= \frac{\pi(\eta')\Gamma(\eta')^L}{\pi(\eta)\Gamma(\eta)^L} * \frac{\left(\prod_{g}\prod_{s} \prod_{k=1}^{\mgs}( \eta'\alpha_s+\mgs-k)\right) \prod_{g=1}^{L} \frac{1}{\Gamma(\mg+ \eta')}}{\left(\prod_{g}\prod_{s} \prod_{k=1}^{\mgs}(\eta\alpha_s + \mgs - k)\right) \prod_{g=1}^{L} \frac{1}{\Gamma(\mg+ \eta)}}\label{eq-peta-accept}
\end{align}

\subsubsection{Proposal distribution for $\alpha$}\label{app-alphaproposaldist}

\begin{figure}[h]
\centering    \includegraphics[width=.6\linewidth]{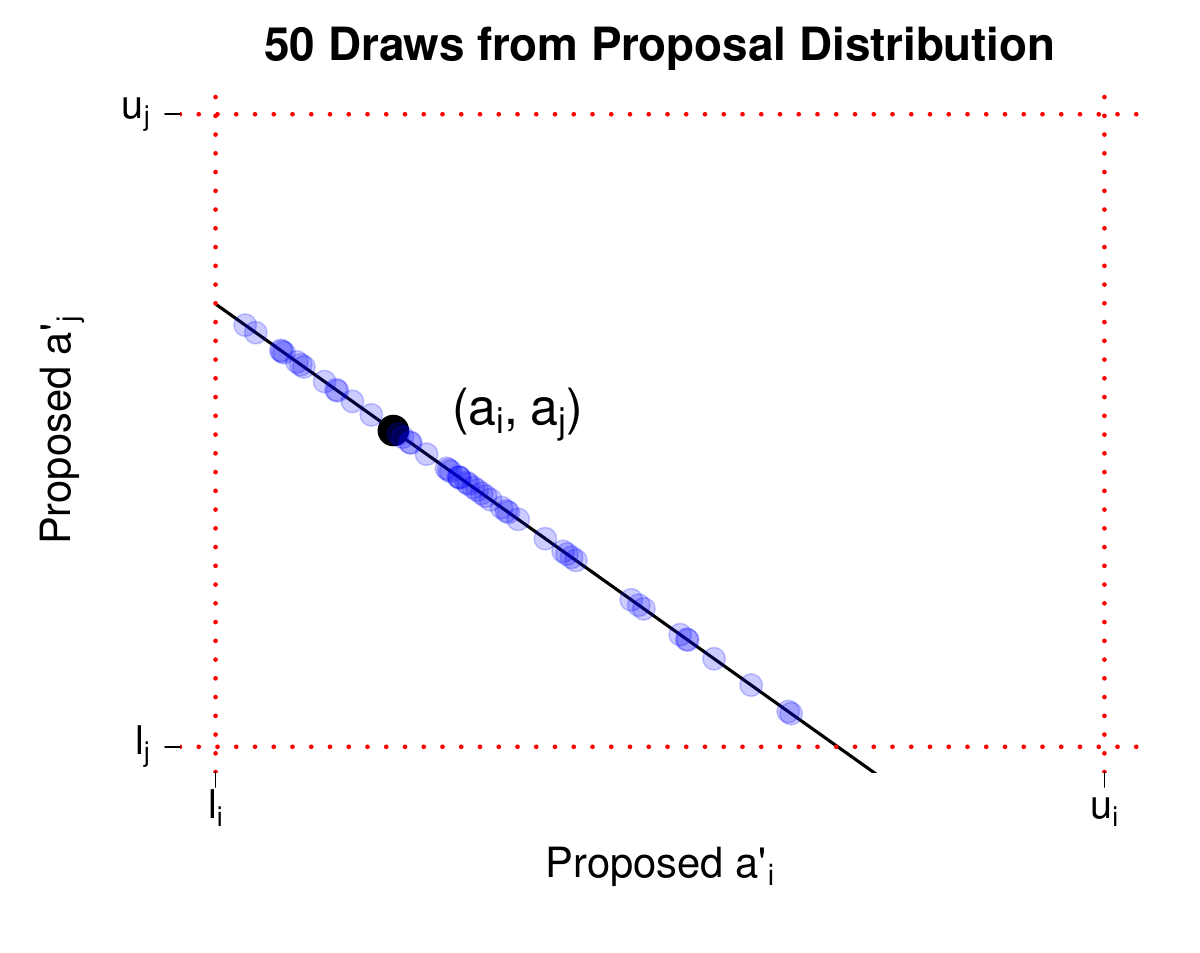}
    \caption{\textbf{Illustration of proposal distribution}. Shows 50 simulated proposals for $\alpha_i',\alpha_j'$ from the proposal distribution described in the Section \ref{app-alphaproposaldist} of the text when the limits on each parameter's values are $[l_i,u_i]$ and $[l_j,u_j]$ respectively and when $\sigma = \frac{1}{2}(\alpha_i+\alpha_j)$. Each proposal satisfies $\alpha_i'+\alpha_j' = \alpha_i+\alpha_j$.}
    \label{fig-alpha-proposal-distribution}
\end{figure}

The proposal distribution for $\alpha_i',\alpha_j'$ given $\alpha_i,\alpha_j,\alphavec_{-ij}$ must propose values such that $\alpha_i'+\alpha_j'=\alpha_i+\alpha_j$ and avoid proposing $\alpha_j'=0$ or $\alpha_i'=0$ where the Dirichlet density is 0 or infinite (proposing very close to $0$ can also lead to bad behavior). To achieve this, we use geometry. Suppose we are given bounds $\alpha_i,\alpha_i'\in[l_i,u_i],\alpha_j,\alpha_j'\in[l_j,u_j]$ and constraint  $\alpha_i+\alpha_j = \alpha_i' + \alpha_j'$. With no further information, the bounds are $[0,1],[0,1]$. Then the possible values of $\alpha_i',\alpha_j'$ lie on a line segment within a box. If the distance from the left point of the line segment to $(\alpha_i,\alpha_j)$ is $L_1$ and the distance from $(\alpha_i,\alpha_j)$ to the right end point is $L_2$, then we can sample $\epsilon$ from a truncated normal $N(0,\sigma^2)$ distribution on $[-L_1, L_2]$ and then add $\epsilon$ to $(\alpha_i,\alpha_j)$ along the direction of the line to obtain a new point. figure~\ref{fig-alpha-proposal-distribution} shows $50$ independent draws from this proposal distribution for a given $(\alpha_i,\alpha_j)$ starting point. 
Calculating the proposal probabilities requires a bit of algebra. Let $\|\bullet\|_2$ be the Euclidean norm. Note that the line of possible  $(\alpha_i',\alpha_j')$ values has the same slope as the line from $(\alpha_i,\alpha_j)$ to $(l_i,\alpha_i+\alpha_j-l_i)$ 
and therefore has slope and intercept,
\begin{equation}
    \beta_1 =\frac{\alpha_j-\alpha_i-\alpha_j+l_i}{\alpha_i-l_i}=\frac{l_i-\alpha_i}{\alpha_i-l_i}=-1 \hspace{1cm}\beta_0 = \alpha_j + \alpha_i
\end{equation}
Hence, the intersection points with limits (red lines in figure~\ref{fig-alpha-proposal-distribution}) are:
\begin{align}
    &\text{Intersect } y=u_j : (\alpha_i+\alpha_j-u_j,u_j) && \text{Intersect } x=l_i :(l_i,\alpha_i+\alpha_j-l_i) \\ 
    &\text{Intersect } y=l_j : (\alpha_i+\alpha_j-l_j,l_j) && \text{Intersect } x=u_i :(u_i,\alpha_i+\alpha_j-u_i)
\end{align}

Some of these points do not actually lie in the box. The left intersection point with the box is either the point intersecting $x=l_i$ or $y=u_j$ and the right intersection point with the box is either the intersection with $y=l_j$ or $x=u_i$. This can be calculated as:
\begin{align}
 p_{left}  &= (\max(\alpha_i+\alpha_j-u_j,l_i), \min(\alpha_i+\alpha_j-l_i,u_j))\\
 p_{right}  &= (\min(\alpha_i+\alpha_j-l_j,u_i), \max(\alpha_i+\alpha_j-u_i,l_j)).
\end{align}
Using these, we have
\begin{align}
    L_1 &= \left\lVert
\begin{pmatrix} \alpha_i \\ \alpha_j \end{pmatrix}
- p_{left}
\right\rVert_{2} \hspace{1cm} L_2 = \left\lVert
\begin{pmatrix} \alpha_i \\ \alpha_j \end{pmatrix}
- p_{right}
\right\rVert_{2} \hspace{1cm} \epsilon =\text{sign}(\alpha_i'-\alpha_i)\left\lVert
\begin{pmatrix} \alpha_i \\ \alpha_j \end{pmatrix}
-
\begin{pmatrix} \alpha_i' \\ \alpha_j' \end{pmatrix}
\right\rVert_{2},
\end{align} 
\noindent Finally, 
\begin{align}
q(\alpha'_i,\alpha_j'\mid \alpha_i,\alpha_j)&=\text{TruncNorm}(\epsilon;0,\sigma^2,-L_1,L_2).
\end{align}
It is straightforward to swap the roles of $\alpha_i',\alpha_j'$ and $\alpha_i,\alpha_j$ in these calculations.

Choosing $\sigma$ requires navigating the usual MCMC trade-off between achieving a high acceptance rate and achieving good mixing. In practice, we found using $\sigma = \frac{1}{2}(\alpha_i+\alpha_j)$ gave a reasonably large amount of movement and an acceptance rate of on average about 1/3 of the $(\alpha_i,\alpha_j)$ pair proposals per iteration (meaning that in most cases, the overall $\alphavec$ vector changed in about a third of its entries). Further fine-tuning could be done to optimize this for a given application.

\subsubsection{Prior and proposal distribution for eta}

Since $\eta$ is a positive scaling parameter, we use a proposal that naturally proposes jumps between different orders of magnitude by drawing $\eta' = \eta\exp(X)$ with $ X \sim N(0,\sigma^2)$. If $\sigma$ is too large, this can quickly explode in magnitude with many rejections. Based on simulations, $\sigma=1$ works fairly well. Calculating the proposal probability yields
\begin{equation}\label{eq-eta-proposal}
    q(\eta'\mid \eta)= \frac{1}{\sigma \eta'}\phi\left(\frac{1}{\sigma}\log\left(\frac{\eta'}{\eta}\right)\right)
\end{equation}
where $\phi$ is the density of the standard Gaussian. Note that in the acceptance ratio, by symmetry, most of this cancels leaving
\begin{equation}
    \frac{q(\eta\mid \eta')}{q(\eta'\mid \eta)} = \frac{\eta'}{\eta}
\end{equation}

Next, suppose that we know a bound $\eta\in [l,u]$. A simple way to enforce this is to set $\eta'' = \max(\min(\eta', u), l)$. Calculating the probability of proposing $\eta''$ given $\eta$ is slightly more complicated. We have:
\begin{equation}
q(\eta''\mid \eta) =
\begin{cases}
    1-\Phi(\frac{1}{\sigma}\log(u/\eta)) & \text{if $\eta'' = u$}  \\
    \Phi(\frac{1}{\sigma}\log(l/\eta)) & \text{if $\eta'' = l$}  \\
   \frac{1}{\sigma \eta'} \phi(\frac{1}{\sigma}\log(\eta''/\eta)) & \text{otherwise} 
\end{cases}
\end{equation}

\subsubsection{Incorporating bounds}\label{sec-bounds}

When estimates of $\Pr(R=1\mid G)$, $\Pr(S\mid G,R=1)$ and $\Pr(S\mid G)$ required for Bayes' rule in equation~\ref{eq-bayes-rule} come from different data sources, they may be incoherent with each other and give estimated probabilities outside of $[0,1]$. To address this, we can calculate bounds on the estimate of $\Pr(S\mid G,R=1)$ using the estimates of $\Pr(R=1\mid G)$ and $\Pr(S\mid G)$, which are generally more reliable. By the Law of Total Probability, 
\begin{equation}
\Pr(S\mid G,R=1) = \frac{\Pr(S\mid G)-\Pr(S\mid G,R=0)\Pr(R=0\mid G)}{\Pr(R=1\mid G)}
\end{equation}
Setting $\Pr(S\mid G,R=0)=0$ and $\Pr(S\mid G,R=0)=1$ leads to the following bounds 
\begin{equation}\label{eq-bounds}
\Pr(S\mid G,R=1) \in  \left[\max\left(\frac{\Pr(S\mid G)-\Pr(R=0\mid G)}{\Pr(R=1\mid G)},0\right) ,  \min\left(\frac{\Pr(S\mid G)}{P(R=1\mid G)}, 1\right)\right]
\end{equation}

In our application, estimating these bounds always resulted in the lower bounds of zero, but the upper bounds were much less than one.  A simple post-hoc option to enforce these bounds is to, after fitting the MCMC model, set any estimates that are outside them to their nearest bound. Although the estimates of $\Pr(S\mid G,R=1)$ no longer sum to $1$, this is not a problem if only using them to obtain the Poisson sampling probabilities.

A more principled approach is to incorporate the bounds into Bayesian estimation by enforcing them in the priors and proposal distributions. We designed the proposals above with such bounds in mind for this reason. Suppose $\theta_{gs}\in [l_{gs},u_{gs}]$. Then equation~\ref{eq-condpostmean} for the posterior mean also has this bound
\begin{equation}
\E[\theta_{gs}\mid \alphavec, \eta,\Data,\gammavec] = \frac{\mgs+\eta\alpha_s}{\mg+ \eta}\in[l_{gs},u_{gs}]
\end{equation}
and solving this implies a bound on $\alpha_s$ given $\eta$, 
\begin{align}
\alpha_s \in \left[\frac{l_{gs}(\mg+ \eta)-\mgs}{\eta}, \frac{u_{gs}(\mg+ \eta)-\mgs}{\eta}\right]
\end{align}
\noindent which we can combine with $\min$ and $\max$ for an overall bound.
\begin{align}
\alpha_s \in \left[\max_{g}\left(\frac{l_{gs}(\mg+ \eta)-\mgs}{\eta}\right), \min_{g}(\frac{u_{gs}(\mg+ \eta)-\mgs}{\eta})\right].
\end{align}
\noindent We can also calculate bounds on $\eta$ given $\alpha_s$. For example, if $\alpha_s-l_{gs}\geq 0$ and $\alpha_s-u_{gs}\geq 0$, then 
\begin{align}
\eta \in \left[\frac{l_{gs}\mg - \mgs}{\alpha_s - l_{gs}},\frac{u_{gs}\mg - \mgs}{\alpha_s - u_{gs}}\right].
\end{align}
\noindent an overall bound can be found by taking a maximum of the lower bounds and a minimum of the upper bounds.

In principle, these bounds on $\alpha$ and $\eta$ could be incorporated into a sequential update strategy, with the bounds on $\alpha_s$ being calculated from the latest $\eta$ and vice versa and each imposed via the proposal distributions as described above. However, in practice, because of the $\max$ and $\min$, estimation errors can accumulate to distort these bounds and did so in our application. These estimation errors arise from the fact that each $l_{gs}$ and $u_{gs}$ must be estimated.  For this reason, in our application, we chose the more ad-hoc solution of imposing bounds on $\theta_{gs}$ directly after estimation. We leave exploration of whether the bounds above can be made useful to future work.

\subsubsection{Extension: Incorporating surname features}\label{app-surnamefeatures}

Surname data are skewed, with a small number of highly common surnames and a long tail of rare surnames. A limitation of using only the per-surname counts as we did in our main model is that it does not leverage patterns within surnames that are indicative of group membership. For example, while the surname Berkowitz appears hundreds of times in our data, Berlowitz is rare. Our multinomial-Dirichlet model does not use this kind of surname similarity. Below, we briefly show how to extend the Bayesian hierarchical model to use surname features or embeddings that can capture such common structures \cite{DasanaikeImai2026}. 

Let $V=V(s)$ represent a vector of characteristics of each surname, where the frequency $\Pr(S\mid R=1) = \frac{\ms}{\msampsize}$ could be included as one of the features. Then we might consider the model 
\begin{align}
\mgvec &\sim \text{Multinomial}(\mg, \thetagvec), &\text{ where  }& \theta_{gs}:=\Pr(S=s\mid G=g,R=1) \text{ for } s=1,...,|\cS|\\
\thetagvec &\overset{iid}{\sim} \text{Dirichlet}(\eta\alphavec) & \text{ where  } & \alphavec \in \mathbb{R}^{|\cS|} \text{ with } \sum_{s=1}^{|\cS|}\alpha_s = 1 \text{ and } \eta > 0\\ 
\alpha_s &\propto \exp(V(s)^\top \beta) & \text{ where } & \sum_{s=1}^{|\cS|}\alpha_s = 1, \ \ \beta\in \mathbb{R}^{\text{dim}(V)}\\
\eta & \sim \pi_\eta, \beta \sim \pi_\beta & \text{ where } & \pi_\eta,\pi_\beta \text{ are prior distributions.}
\end{align}

Under this model, the conditional posterior mean of $\theta_{gs}$ given $\beta,\eta$ and data $\Data = \{\mgvec : g\in \cG\}$ and known function $V(s)$ has the same form as equation~\ref{eq-condpostmean},
\begin{equation}
    \E[\theta_{gs} \mid \beta,\eta,\Data] = (1-\rho_g)\frac{\mgs}{\mg} + \rho_g  \left(\frac{\exp(V(s)^\top \beta)}{\sum_{s'}\exp(V(s')^\top \beta)}\right),\hspace{.5cm} \rho_g = \frac{\eta}{\eta+\mg}
\end{equation}
where $\rho_g$ again controls a trade-off between location-specific data and a surname-only-based model while $\beta$ represents information-sharing across surnames in which surnames with certain characteristics consistently have higher or lower probabilities. The marginal posterior of $\beta$ and $\eta$ also has the same form as before. Letting $\tau = \sum_{s=1}^{|\cS|} \exp(V(s)^\top\beta)$, it is
\begin{align}
\pi(\beta,\eta|\Data) \propto \pi_\eta(\eta)\pi_\beta(\beta)\Gamma(\eta)^L \left(\prod_{g=1}^{L}\prod_{s=1}^{|\cS|}\prod_{k=1}^{\mgs}\left(\mgs - k + \frac{\eta}{\tau}\exp(V(s)^\top \beta) \right)\right)\prod_{g=1}^{|\mathcal{G}|}\left(\frac{1}{\Gamma\left(\mg + \eta \right)}\right)\nonumber 
\end{align}
Posterior mean estimates for $\theta_{gs}$ from this model can be used in Bayes' rule (equation~\ref{eq-bayes-rule}) as before.

Our model is not the only approach to learning from name features or embeddings. Indeed, there is an extensive literature on name-based classification that provides a variety of possible models and features. 
However, most of this work assumes access to labeled training data [See, for example, \cite{ambekar2009name, ye2017nationality,ye2019secret, Jain2022Importance}]. The purpose of the model in this section is to learn $\Pr(R=1\mid S,G)$ when only a dataset on the minority population of interest (like the obituary data) and an unlabeled dataset with both minority and non-minority population members (here, the sampling frame) are available.

\subsection{Simulations}\label{sec-sim}

We provide a simulated version of both the sampling and the surname distribution estimation. Both are designed to mimic the structure of our Jewish sampling application. All code to run them is available as part of the replication code for this paper and in a public GitHub repository available at \url{https://github.com/kchaz/BISG_sampling_public}.

Our simulations of the sampling process confirm that simply taking random samples of size $\targetg$ in each state yields a low success rate even when the $\targetg$ are allocated using known $\Pr(R=1\mid G=g)$ probabilities and that sampling using the true $\Pr(R=1\mid S=s,G=g)$ can improve the success rate. Unsurprisingly, the degree of improvement depends greatly on the informativeness of the surname probabilities, reflected in how close the $\Pr(R=1\mid S=s,G=g)$ are to $0$ and $1$. 
Our simulations of estimating $\Pr(S=s\mid G=g,R=1)$ indicate that, in line with theory on Bayesian shrinkage, the partial pooling estimation reduces total variation distance of the estimated probability distribution relative to the true one, as compared to using only the raw proportions $\frac{\mgs}{\mg}$ without shrinkage.

\subsubsection{Simulating Jewish obituary data}

In our simulated set-up, we generate $|\cS|$ random strings of 6 characters as ``surnames.'' Then, for each state $g$, we generate a true distribution over surnames by  
\begin{enumerate}
    \item Drawing $\gammavec_{s} \simiid \text{exponential}(1)$ for $s=1,...,|\cS|$ and letting $\gammavec=(\gammavec_1,...,\gammavec_{|\cS|})$
    \item Drawing $\thetagvec \sim \text{Dirichlet}(\gammavec)$ where $\theta_{gs}=\Pr(S=s\mid G=g,R=1)$
\end{enumerate}
\noindent The first step creates skew so that some surnames are more prevalent than others, and the second step ensures variation in their prevalence by state (see figure~\ref{fig-sim-prob}).  We then generate a simulated dataset of $\msampsize$ rows from the minority population by allocating $\msampsize$ among the states using a multinomial draw with AJPP probability estimates $\Pr(G=g\mid R=1)$. For each state $g$ with $\mg$ observations, we draw state-surname counts $(\mgs)_{s\in \cS} \sim \text{multinomial}(\mg,\thetagvec)$. This creates a dataset of $(S_i,G_i,R_i=1)$ with a distribution over states similar to the real U.S. Jewish adult population.

\begin{figure}[h]
    \centering
    \includegraphics[width=\linewidth]{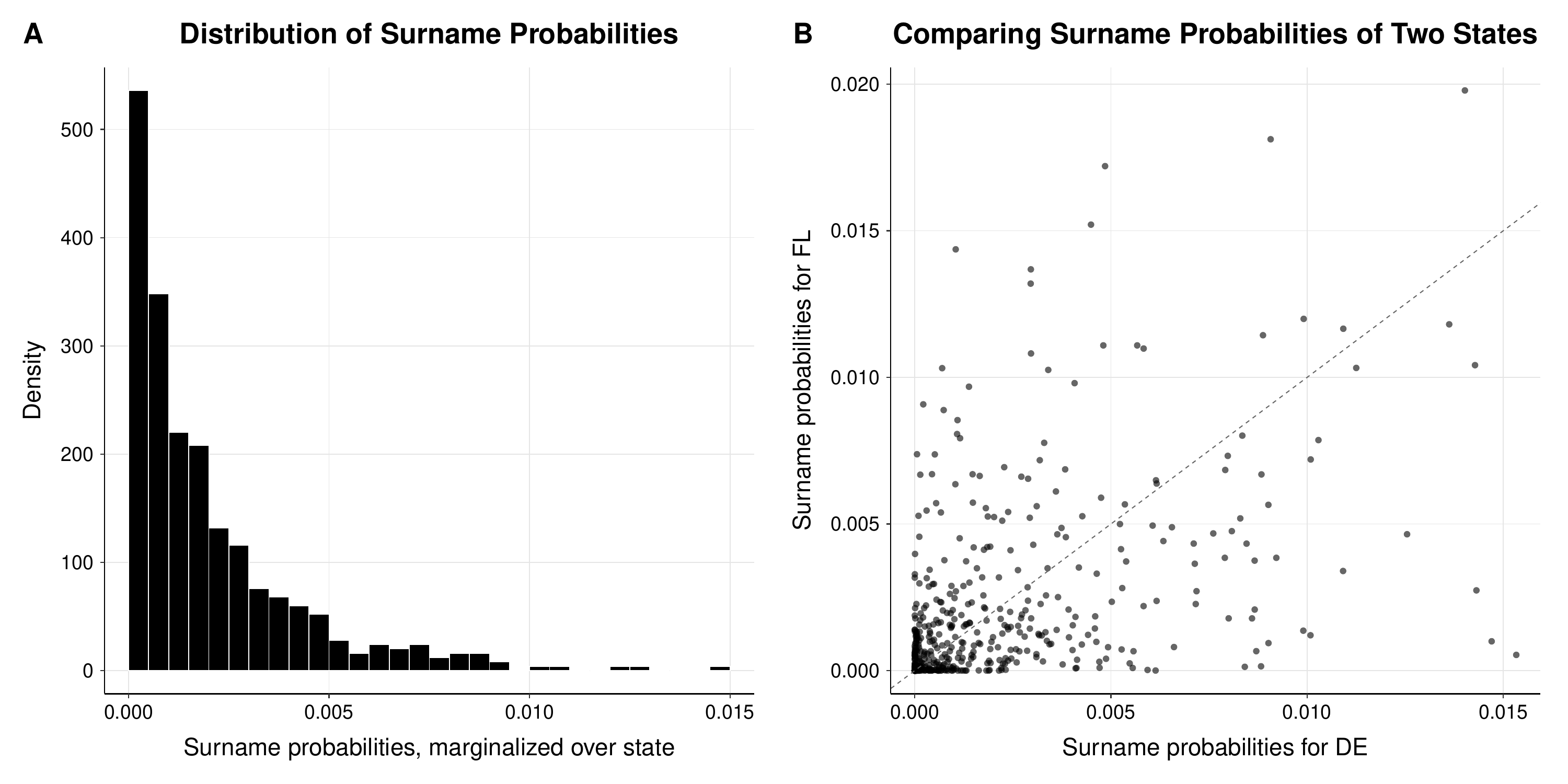}
    \caption{\textbf{Illustration of simulated surname probabilities for 500 surnames.} Both plots are for one example draw from the data generating process for 500 simulated surnames. \textbf{(A)} distribution of $\Pr(S=s\mid R=1)$ overall, calculated using AJPP estimates with $\Pr(S=s\mid R=1)=\sum_{g=1}^{L}\theta_{gs}\Pr(G=g\mid R=1)$. \textbf{(B)} correlation between $\theta_{gs}$ and $\theta_{g',s}$ for two example states. Each point corresponds to a simulated surname $s$.}
    \label{fig-sim-prob}
\end{figure}

\subsubsection{Simulation of sampling methods}

To simulate a sampling frame, we also simulate data for non-Jewish people with a different distribution over the same set of surnames from the obituary data. We take each state's $\thetagvec$ and transform it by
drawing $\epsilon_{gs}\simiid N(0,1)$ and then defining $\nu_{gs} =\frac{\exp(-\beta\theta_{gs}+\epsilon_{gs})}{\sum_{s'=1}^{|\cS|} \exp(-\beta\theta_{gs'}+\epsilon_{gs'})}$ as a new distribution over surnames. This transformation ensures that surnames with higher probability in the $R=1$ group have lower probability in the $R=0$ group and vice versa while also creating some smoothing so that the $R=0$ group is less peaked at a few particular surnames (it is a more mixed group of people). figure~\ref{fig-R1probsR1vsR0} (top) shows an example comparison of the resulting $\Pr(S=s\mid R=1)$ and $\Pr(S=s\mid R=0)$ probabilities for $\beta=1000$. The $\beta$ parameter plays a critical role in controlling how much $\Pr(S=s\mid G=g,R=1)$ and $\Pr(S=s\mid G=g,R=0)$ diverge and hence how close to 0 or 1 the probabilities $\Pr(R=1\mid S=s,G=g)$ are. 
figure~\ref{fig-R1probsR1vsR0} (bottom) shows the distribution of $\Pr(R=1\mid S=s,G=g)$ probabilities in the $R=0$ and $R=1$ groups for four sampling frames generated under four different magnitudes of $\beta$. We see that as $\beta$ increases, we get closer to a situation in which surnames are strongly informative about group membership. 

\begin{figure}[H]
    \centering
    \includegraphics[width=0.62\linewidth]{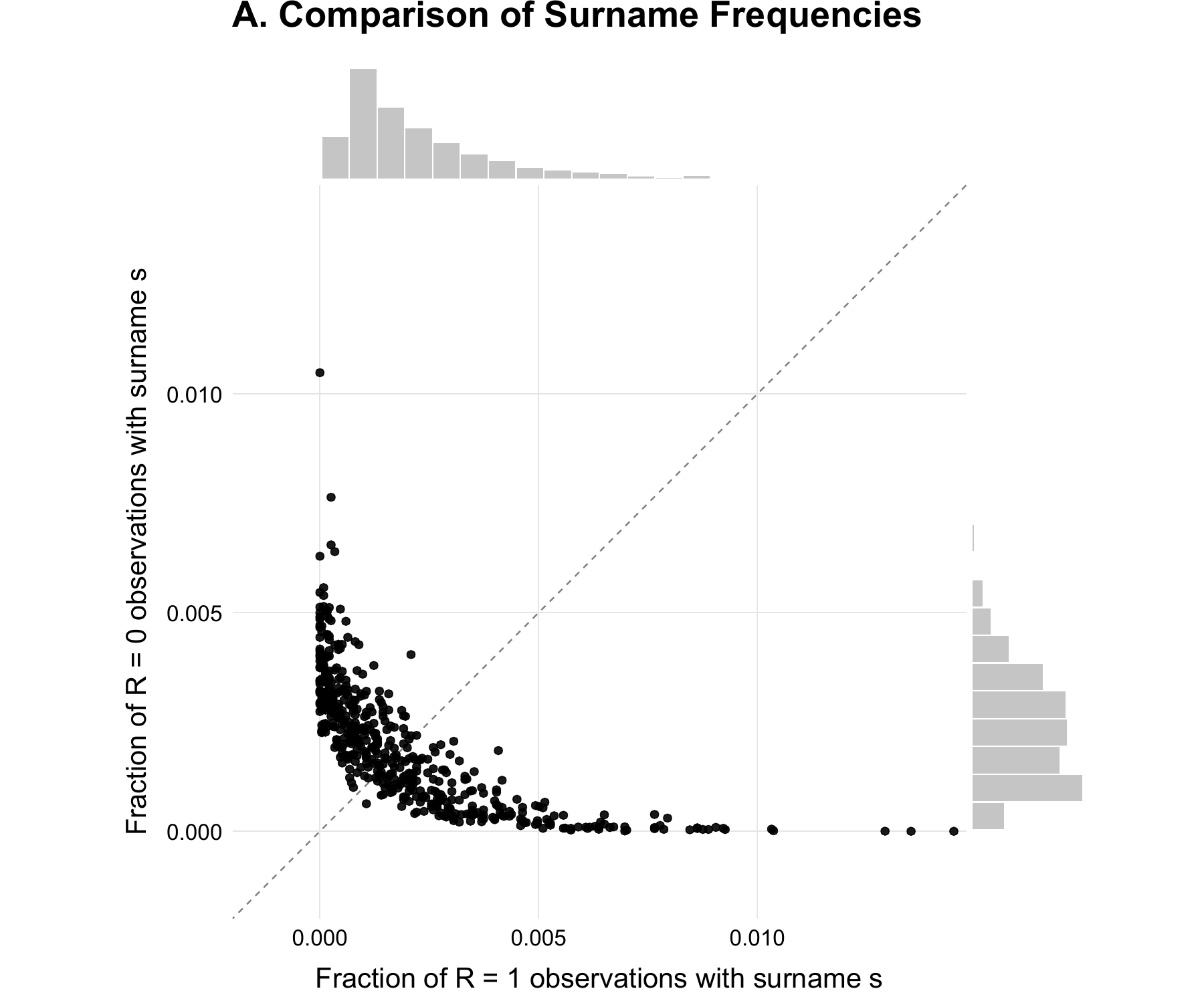}
    \par\vspace{.5cm}
    \includegraphics[width=.90\linewidth]{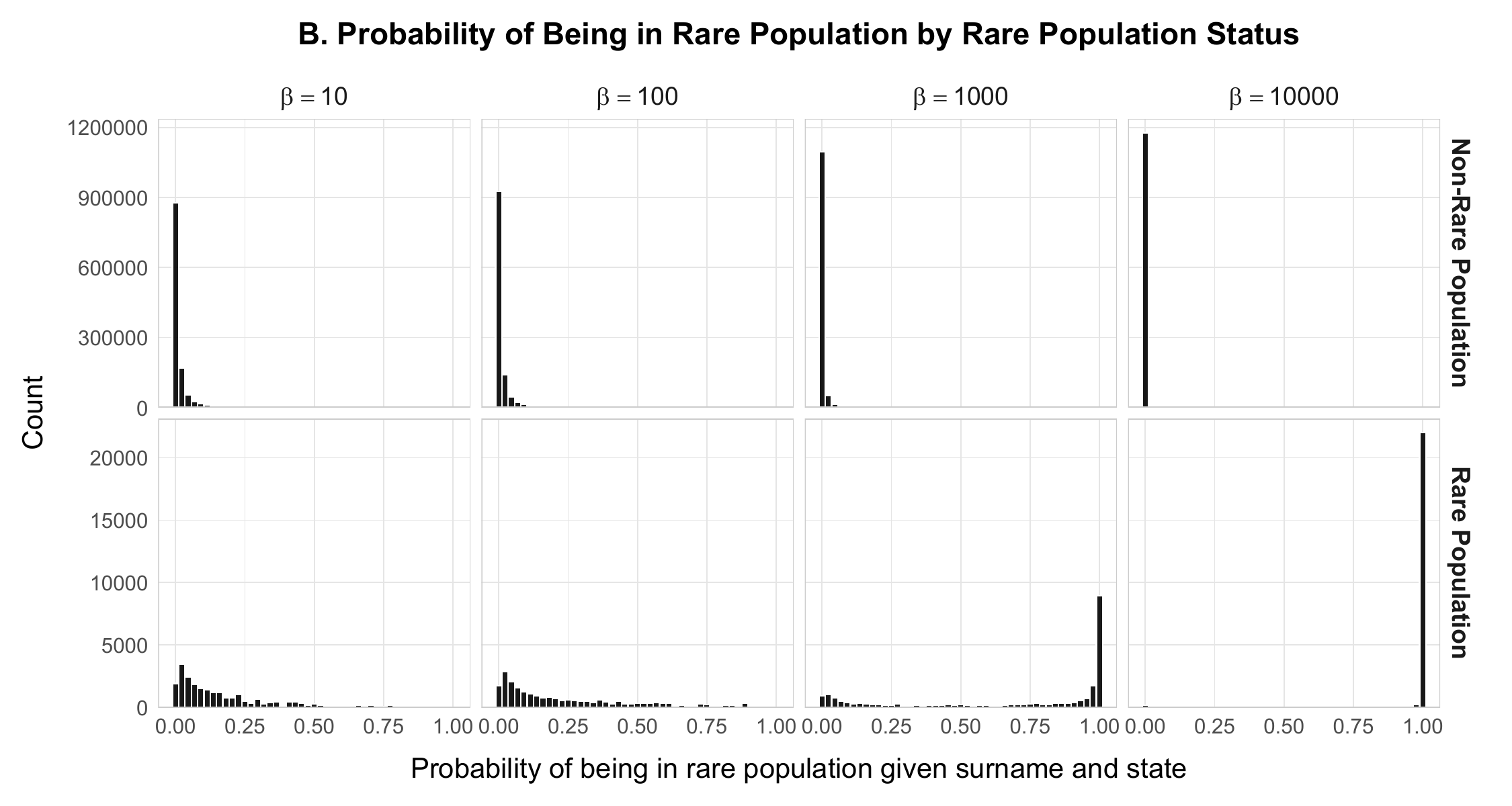}
    \caption{\textbf{Illustration of differences in distribution between groups.} Both plots are for one example draw from the data generating process for 500 simulated surnames.  \textbf{(A)} Distributions of and relationship between $\Pr(S=s\mid R=1)$ and $\Pr(S=s\mid R=0)$ (marginalized over geography) for $\beta=1000$. \textbf{(B)} Distribution of $\Pr(R=1\mid S=s,G=g)$ for $R=0$ units in sampling frame (top) and $R=1$ units in sampling frame (bottom). The impact of $\beta$ is contingent on the number of possible surnames, which affects the magnitude of the probabilities involved.}
    \label{fig-R1probsR1vsR0}
\end{figure}
\newpage 

To create a sampling frame with true minority population labels, we:
\begin{enumerate}
    \item Generate $N_{\cdot 1}$ members of $R=1$ population using the generating process above
    \item Generate $N_{\cdot 0}=N_{\cdot 1} \frac{0.98}{0.02}$ members of $R=0$ population (so the minority population fraction is 0.02) using the generating process above only with $\theta_{gs}$ replaced by $\nu_{gs}$ and $\Pr(G=g\mid R=1)$ replaced by $\Pr(G=g\mid R=0)$ also calculated from the AJPP estimates.
\end{enumerate}

In the real application, we have about 50,000 unique surnames with 120 million observations. We therefore set our simulated sampling frame to have 1.2 million units and 500 unique surnames. We set $N_{\cdot 1}=24,000$ and $N_{\cdot 0}=1,176,000$. We set our target to a sample of $\targetoverall=1,000$ observations. For all methods, we allocate state-specific targets $\targetg$ using the Poisson allocation equation~\ref{eq-Tg_opt_poststrat} from the main text with the true simulated $\Pr(R=1\mid S=s,G=g)$ and the actual $\Pr(G=g\mid R=1)$ from the AJPP data. We compare:
\begin{enumerate}
    \item \textbf{Random sample baseline}: within each state, take a random sample of size $\targetg$
 \item \textbf{Surname-state targeted probabilities}: Poisson sampling as described in the main text. 
\end{enumerate}

We repeatedly generate the sampling frame under the data generating process, and for each frame, apply each sampling method. The results are in figure~\ref{fig-mainsimresult}. The success rate of sampling improves as the $\Pr(R=1\mid S=s,G=g)$ become more informative so that the mean BISG probability in the $R=1$ group is increasingly larger than that of the $R=0$ group. Even when the $R=1$ group has only slightly larger BISG probabilities on average, there is some improvement over random sampling. As precision grows, the growth in yield looks quite linear. This is consistent with the simplified theoretical example in Section \ref{app-precisionanalysis}.

\begin{figure}[H]
    \centering
    \includegraphics[width=\linewidth]{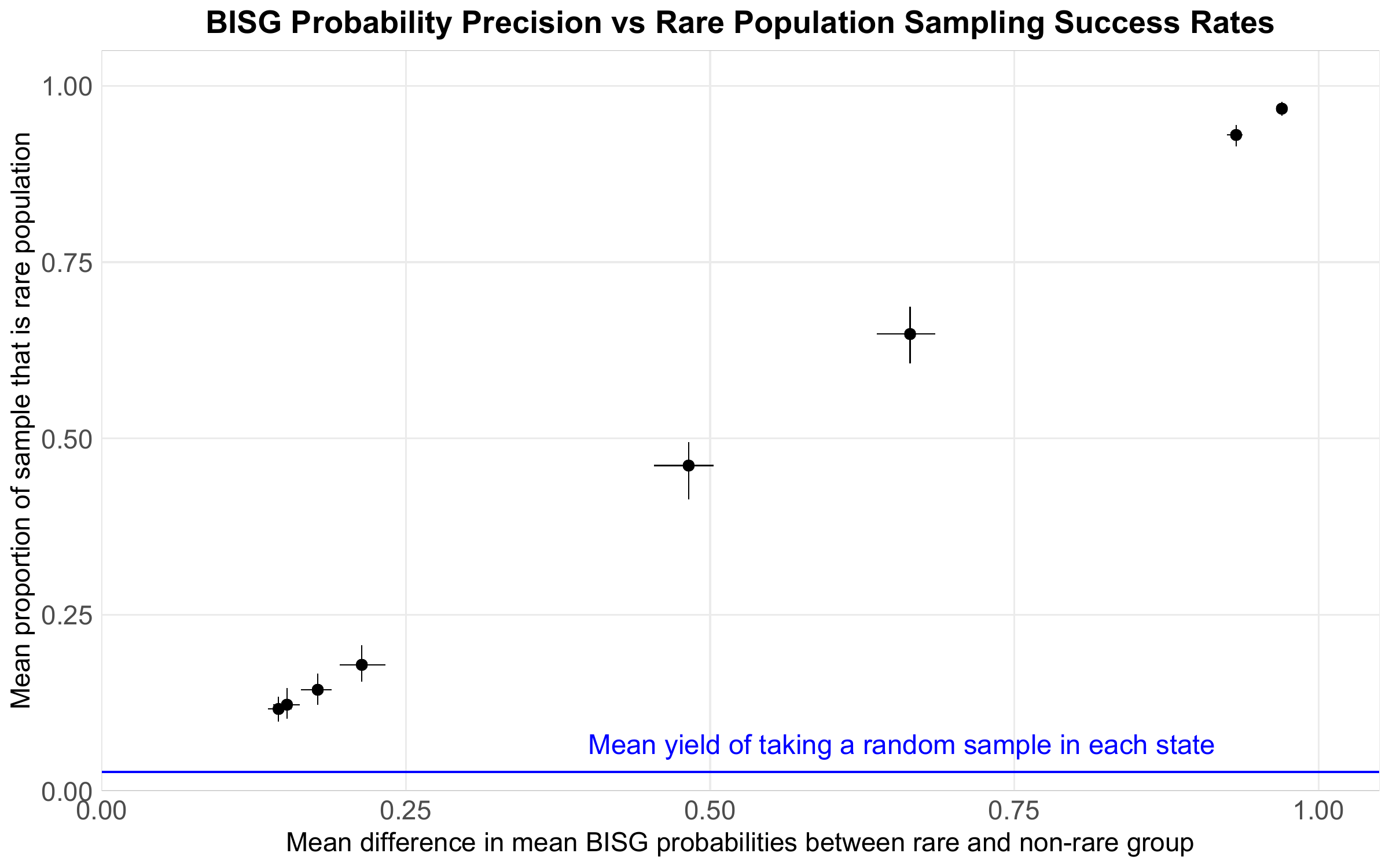}
    \caption{\textbf{Sampling success rates under different BISG probability precision.} Plot of success rate in repeat sampling using Poisson sampling against the mean difference in $\Pr(R=1\mid S,G)$ probabilities between the $R=1$ and $R=0$ groups. Variation in the mean BISG probabilities and yields comes from varying $\beta\in \{1,10,50,100,500,1000,5000,10000\}$ and repeatedly running the data generating process and sampling method 50 times for each $\beta$. Intervals shown are the 0.975 and 0.025 quantiles of the simulated values. For comparison, the horizontal line gives the mean yield of taking a simple random sample in each state of the same size as the target size used in the Poisson sampling. Note that the difference in mean BISG probabilities on the x axis of this plot is mainly driven by the BISG probabilities of the $R=1$ group. Those in the $R=0$ group range in $[0.00054, 0.017]$ while the rare group probabilities range in $[0.15, 0.97]$}
    \label{fig-mainsimresult}
\end{figure}

\subsubsection{Estimating surname distributions}

We applied the MCMC sampler described in Section~\ref{sec-sampler} to sample from the Bayesian hierarchical model described in Section~\ref{sec-learning-surname-dist} for the simulated obituary data while varying the number of surnames ($|\cS|\in \{100, 1000, 10000\}$), the number of observations ($\msampsize\in \{50000, 100000\}$), and method for initializing $\alpha$ (uniform over the possible surnames $\alpha_s^{(0)}=\frac{1}{|\cS|}$ or using a smoothed version of the MLE $\alpha_s^{(0)}=\frac{\ms + 1}{m + |\cS|}$). The most similar to our actual case in terms of surname-to-sample size-ratio is the 10,000 surnames with 50,000 observations case, though the fact that these are still spread over the 51 states in our simulation makes our simulated case still more data-scarce than our actual case with over 200,000 observations spread over 51 states. 

For the 100, 1000, and 10000 surname cases, we draw a fixed set of true $\Pr(S=s\mid G=g,R=1)$ parameters. For each $\msampsize$, we then simulate 5 datasets under these probabilities and fit the model to each, running each chain for 2000 iterations. In general, we see good mixing in the $\alpha$ parameter chains while for the $\eta$ chain, after some initial large fluctuations, fairly poor mixing. See Figures \ref{fig-achain_plot} and \ref{fig-eta_chain_plot} for examples. This low mixing is unsurprising given our proposal distribution, which could easily propose to double or halve the value of $\eta$. The initial fluctuation reflects the chain converging to the right order of magnitude for $\eta$, and after that, only small fluctuations are accepted. Decreasing $\sigma$ in the proposal distribution could improve the proportion of accepted proposals but with the usual trade-off of smaller movements. 

\begin{figure}[H]
    \centering
    \includegraphics[width=\linewidth]{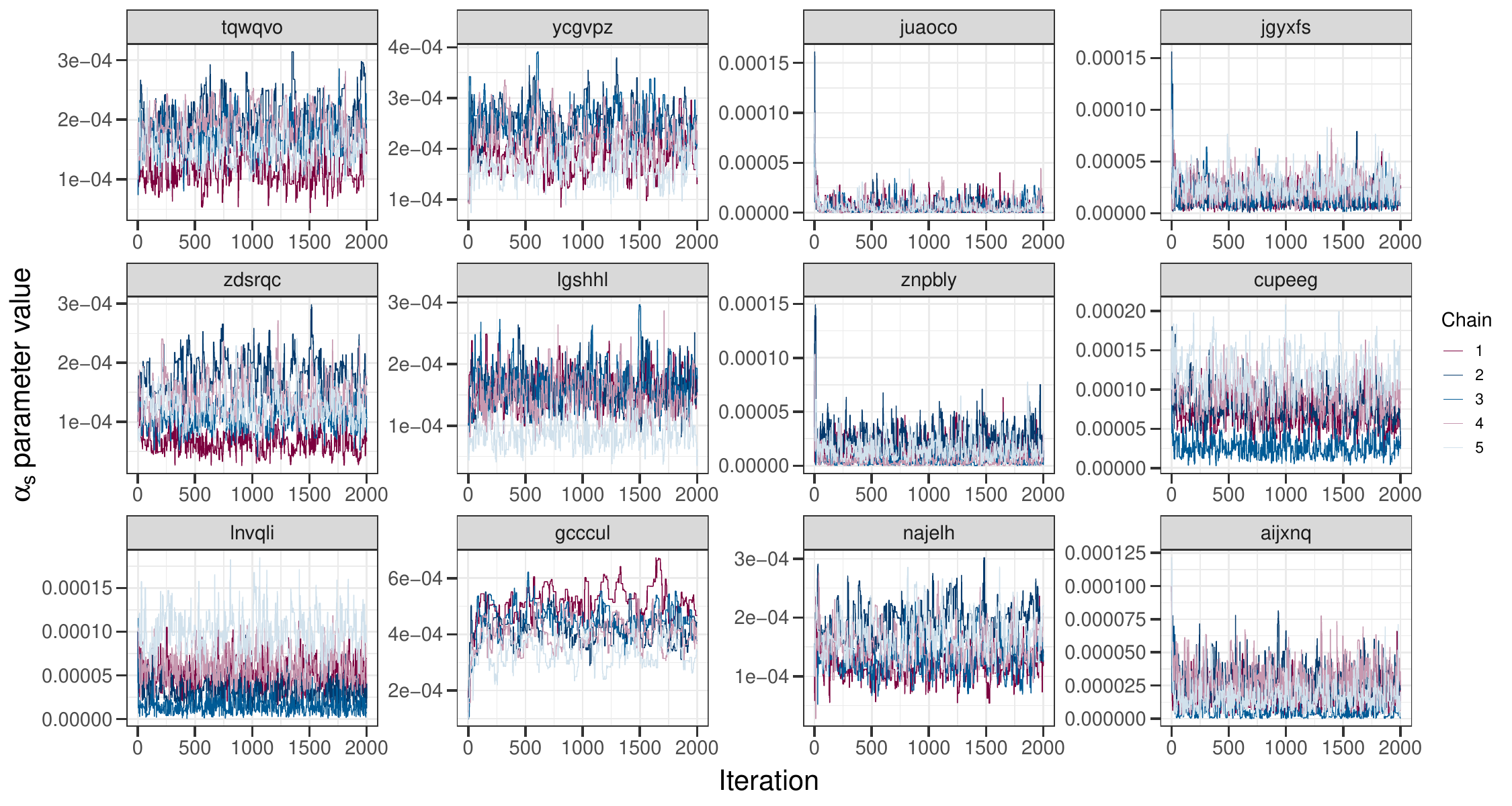}
    \caption{\textbf{Trace plot for $\alpha_s$ from simulation}. Shown for a random subset of 12 fake surnames and 5 independent Markov chains for the $\msampsize=100,000$, $|\cS|=10,000$ simulation setting with $\alphavec$ initialized as uniform over $\cS$. In any single iteration, there are $5000$ accept-reject steps for 5000 pairs of $(\alpha_i,\alpha_j)$. Across runs, on average 34\% of these proposals were accepted. }  \label{fig-achain_plot}
\end{figure}

\begin{figure}[H]
    \centering
    \includegraphics[width=\linewidth]{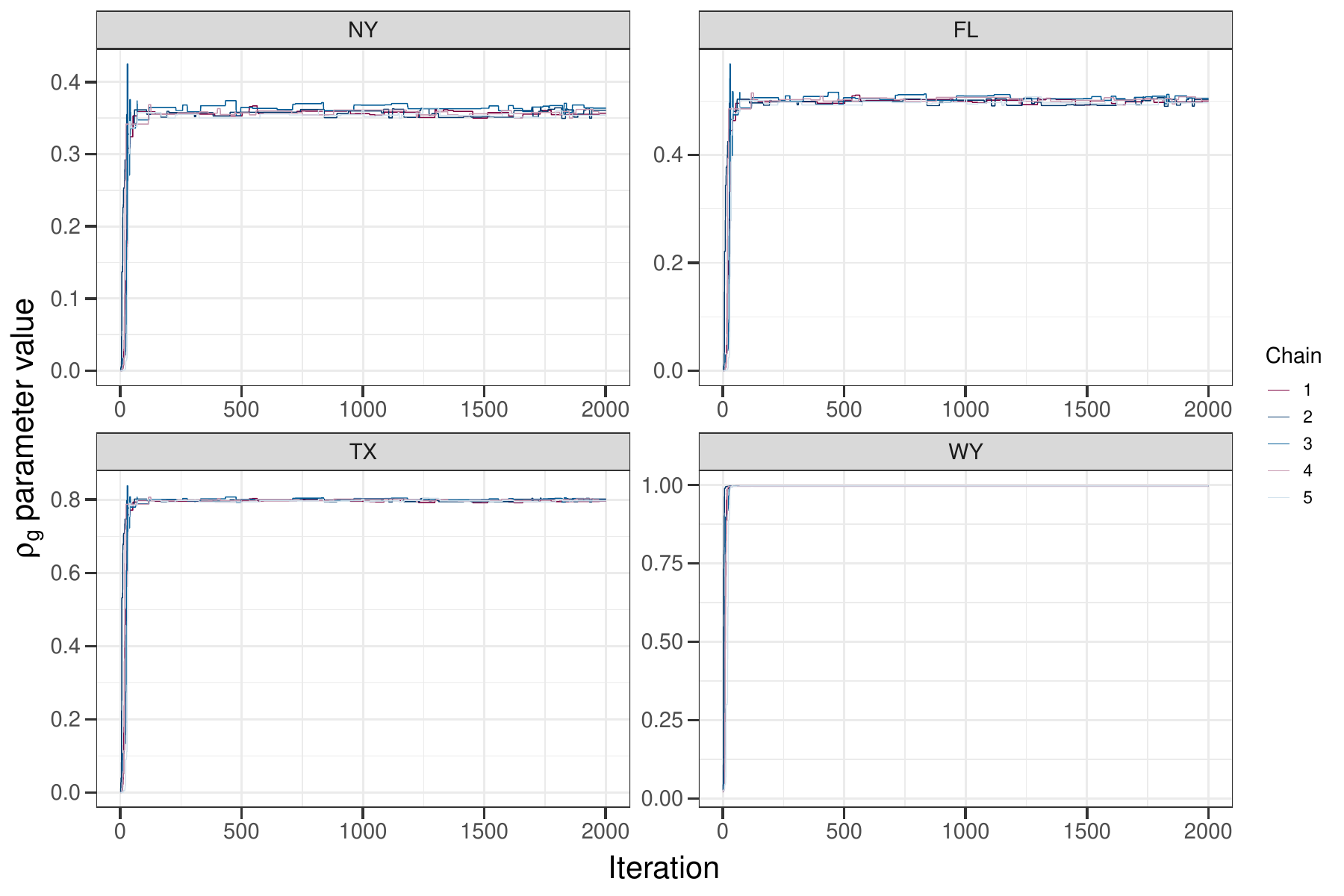}
    \caption{\textbf{Trace plot for $\rho_g$ from simulation}. Shown for four states selected to highlight different degrees of partial pooling. Plots are for 5 independent Markov chains for the $\msampsize=100,000$, $|\cS|=10,000$ case with $\eta$ initialized to 1 and each $\rho_g$ calculated from $\eta$ and $\mg$. Note the different scales on the y axes. On average across the runs, 2.7\% of the proposals were accepted.}
    \label{fig-eta_chain_plot}
\end{figure}

For each model fit, we also estimate the posterior means of each $\alpha_s$ and each $\theta_{gs}=\Pr(S=s\mid G=g,R=1)$ parameter. Figures \ref{fig-sampler-sim-theta} and  \ref{fig-sampler-sim-alpha} compare these parameters to their true values in terms of the average total variation distance (TV), which is calculated as follows.
\begin{align}
\text{TV}(\alphavec)&=\frac{1}{2}\sum_{s=1}^{|\cS|}|\widehat{\alpha}_s-\alpha_s|\\
   \text{meanTV}(\thetavec) &= \frac{1}{L}\sum_{g=1}^{L}TV(\thetagvec)=\frac{1}{L}\sum_{g=1}^{L}\frac{1}{2}\sum_{s=1}^{|\cS|}|\widehat{\theta}_{gs}-\theta_{gs}|
\end{align}

We prefer TV to root mean squared error here because it is invariant to scale as the number of surnames increases (and hence magnitude of $\alpha_s$ and $\theta_{gs}$ decreases) and thereby captures the greater relative estimation error when using the same amount of data to estimate parameters for a larger number of surnames.
We also calculate this metric for the simple observed proportions, $\frac{\ms}{\msampsize}$ and $\frac{\mgs}{\mg}$, which are the MLEs for $\alpha_s$ and $\theta_{gs}$ respectively. 

As expected from shrinkage theory, we see that for $\theta$, the Bayesian model with partial pooling yields a smaller TV than the MLE does (for a fixed number of surnames, it also yields a smaller RMSE). It achieves this even after only a few iterations, suggesting that the benefit is largely related to the shrinkage relative to using the raw proportions. The $\alpha$ initialization method makes little difference, while the number of surnames does matter, with the TV growing with the dimension of the parameter space. Unsurprisingly, the TV also decreases as sample size increases.

For the $\eta$ parameter, there is not a true value to compare to. However, we can examine how its value changes with $|\cS|$ and the sample size $\msampsize$. Importantly, the second form of the posterior mean in equation~\ref{eq-conditpostmean-main} indicates that the relevant quantity to consider is not $\eta$ but $\rho_g = \eta/(\eta+\mg)$, which indicates the degree to which $\alpha_s$ is favored over $\frac{\mg}{\msampsize}$ and is comparable across settings because it is always in $[0,1]$. At each iteration $b$ in the chain, we calculate
\begin{align}
    \bar{\rho}^{(b)}&=\frac{1}{L}\sum_{g=1}^{L}\rho_g^{(b)}\\
    {\rho}^{(b)}_{max}&=\max_{g}\rho_g^{(b)}\\
    {\rho}^{(b)}_{min}&=\min_{g}\rho_g^{(b)}
\end{align}
and then we average these over 2000 iterations.
In figure~\ref{fig-rhoplot} we plot these quantities against the number of surnames and by sample size.

The figure shows that, as expected, the smoothing is stronger when the sample size $\msampsize$ is smaller and when the number of surnames is greater -- both corresponding to greater sparsity. This happens even as the magnitude of $\eta$ increases with $\msampsize$ to match the greater magnitudes of the $\mg$. We see that the range of smoothing is large, which is in line with some states having much higher overall counts $\mg$ than others in the simulated (and real) data. Note that the average value of $\rho_{g}$ over individuals, which is $\frac{1}{\msampsize}\sum_{g=1}^{L}\mg \rho_g$ (not shown), is generally much lower ($\leq 0.1)$ because more observations come from states with less shrinkage.

\begin{figure}[H]
    \centering
\includegraphics[width=\linewidth]{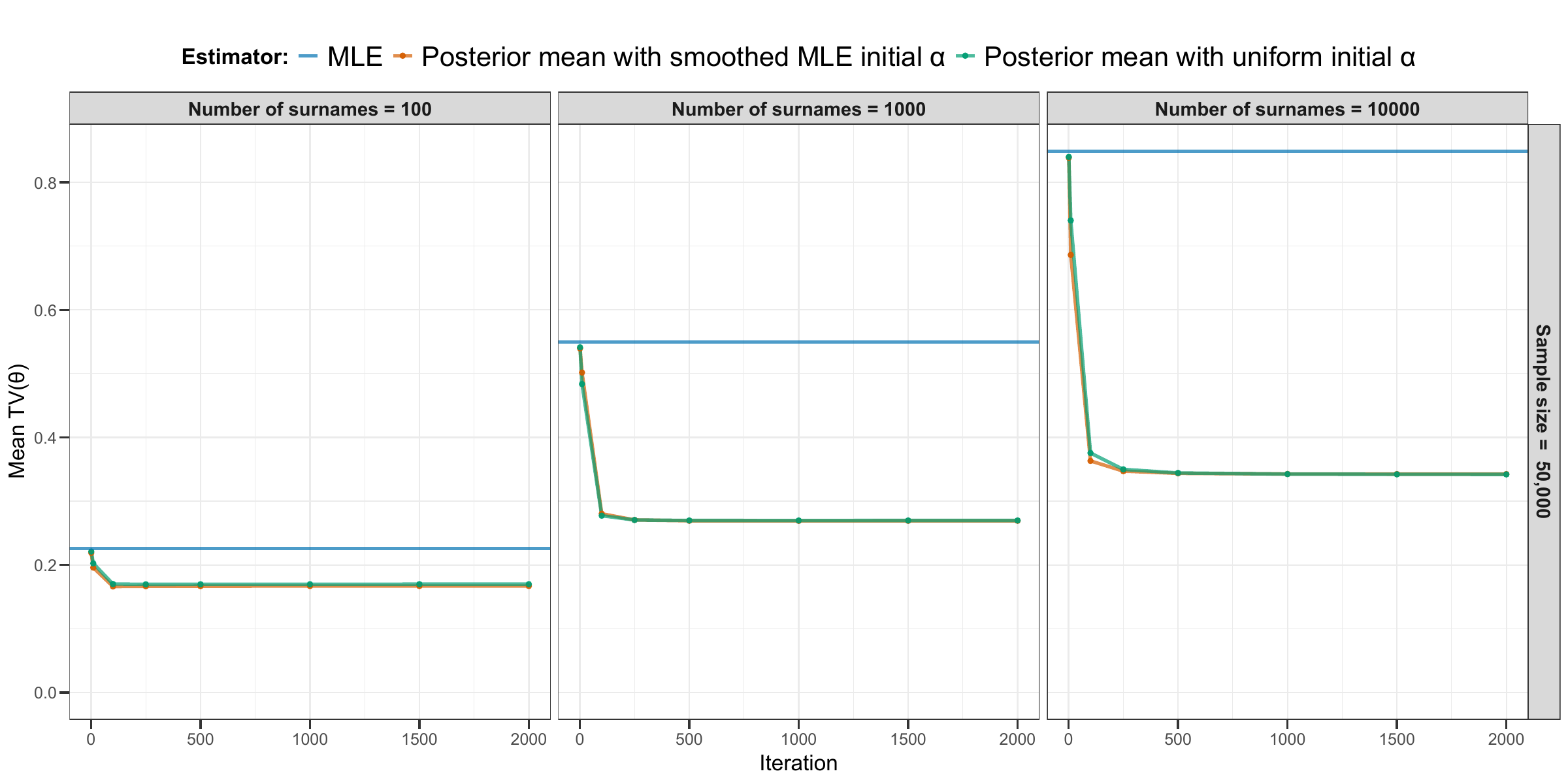}\newline 
\includegraphics[width=\linewidth]{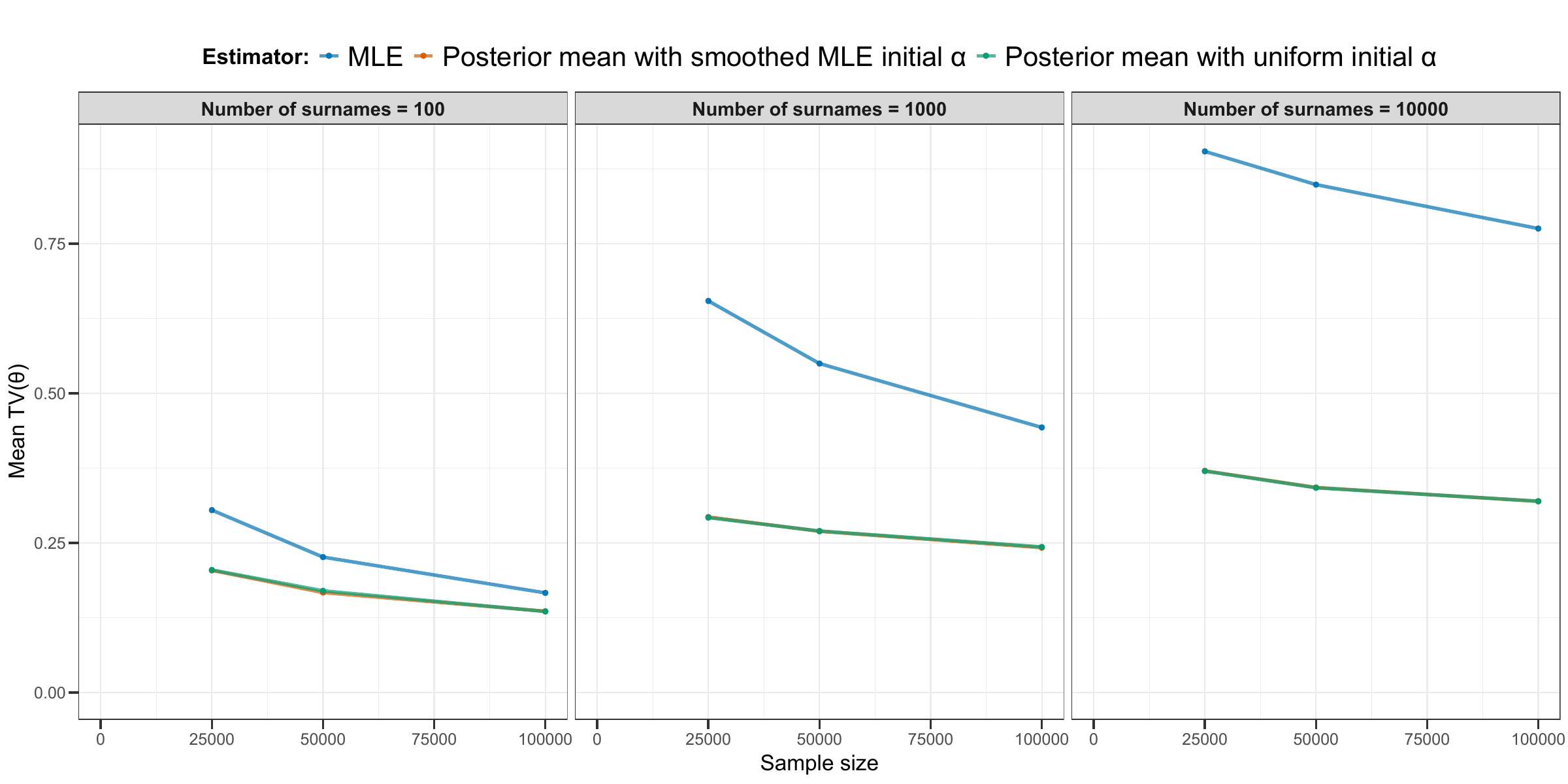}
    \begin{picture}(0,0)
  \put(-230,475){\textbf{A}}
  \put(-230,240){\textbf{B}}
\end{picture}
    \caption{\textbf{Accuracy of estimated BISG probabilities in simulation.} Total variation distance of  estimators of $\thetavec$ versus true value. \textbf{(A)} Relationship to number of iterations and number of surnames, fixing $\msampsize=50,000$. \textbf{(B)} Relationship to sample size $\msampsize$ and number of surnames, fixing number of iterations to 2000.  Partial pooling estimators out-perform MLE. Each point is an average over 5 independent simulation replicates.}\label{fig-sampler-sim-theta}
\end{figure}

\begin{figure}[H]
    \centering
\includegraphics[width=\linewidth]{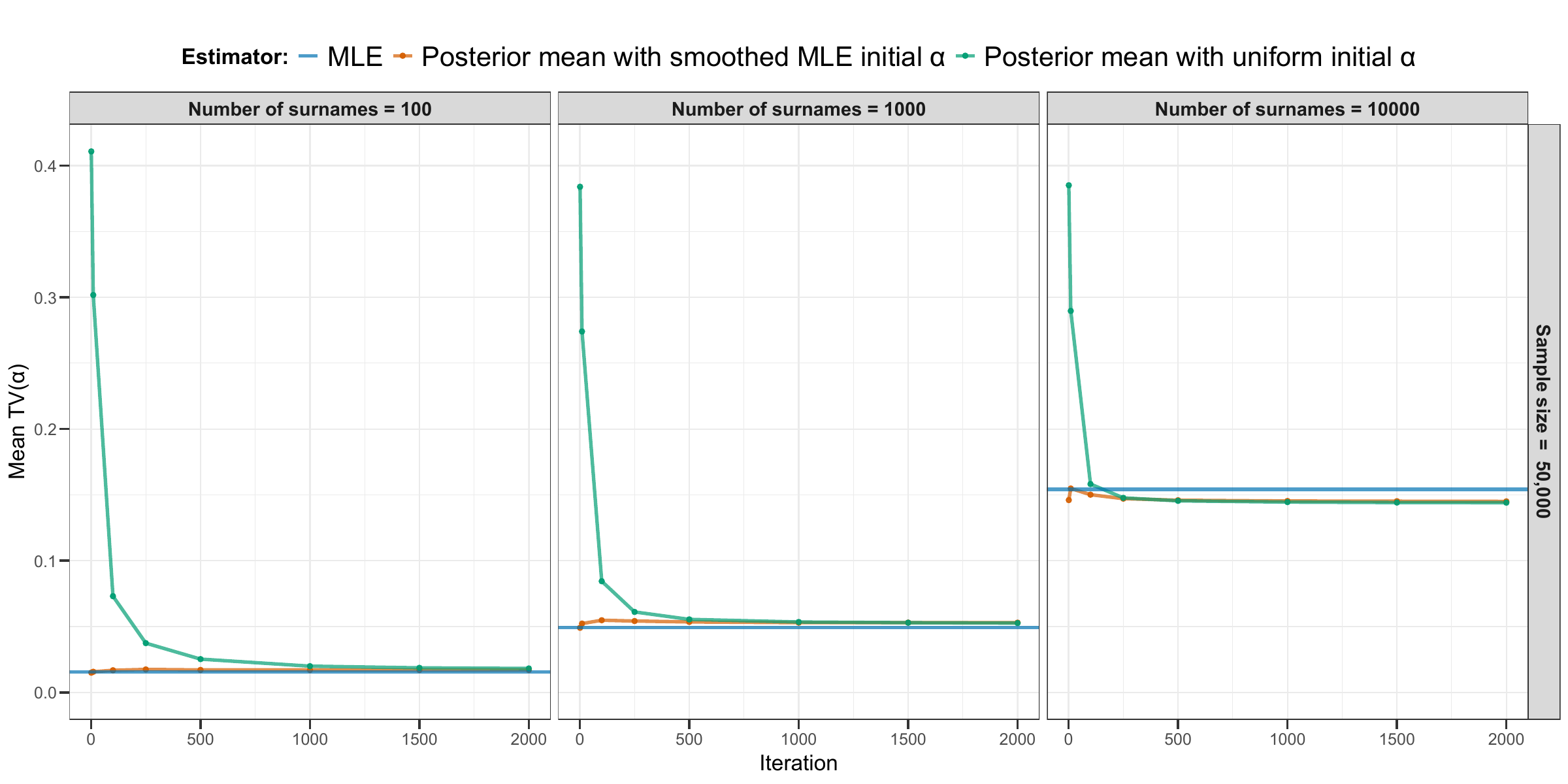}\newline 
\includegraphics[width=\linewidth]{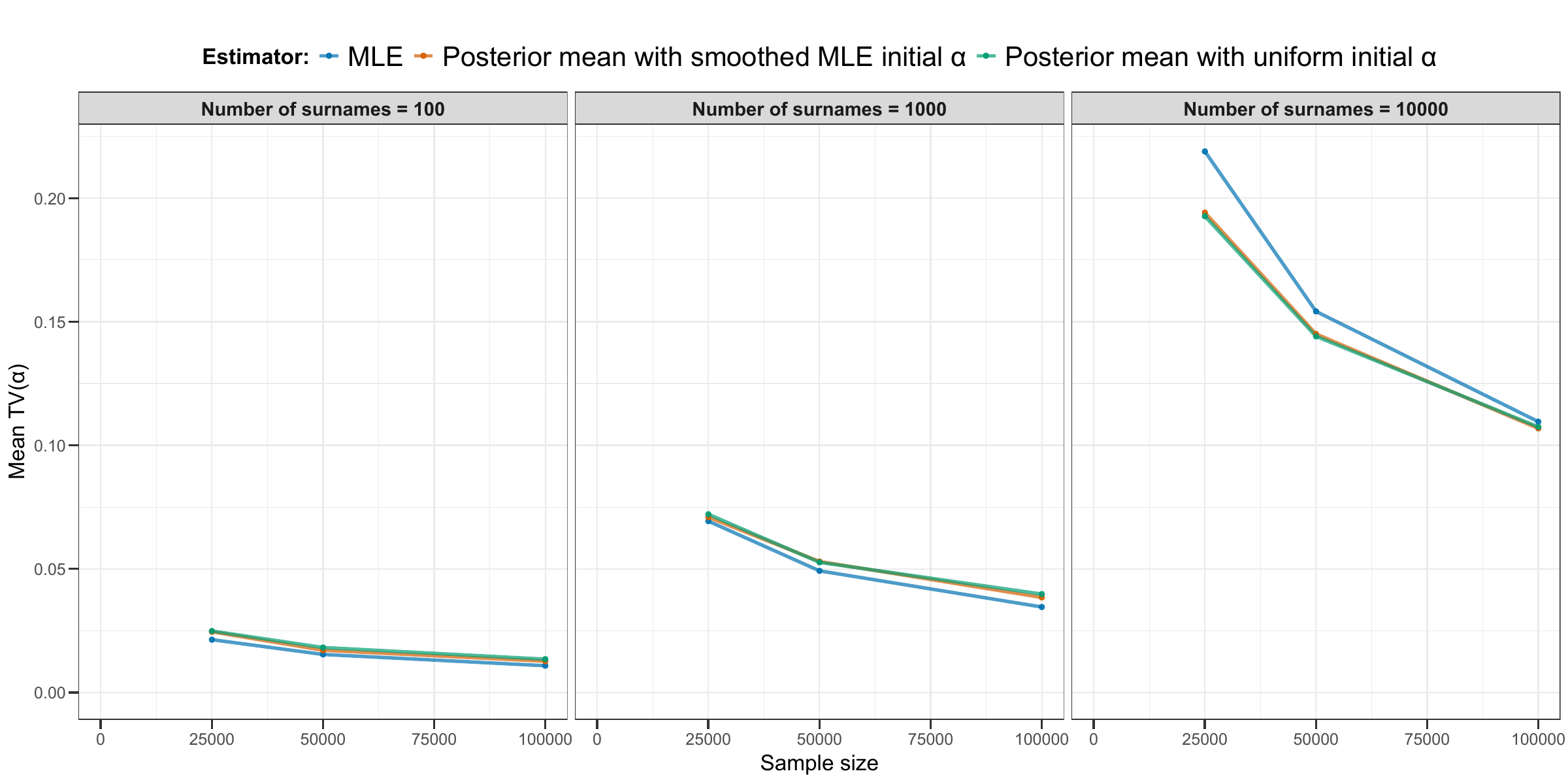}
   \begin{picture}(0,0)
  \put(-230,475){\textbf{A}}
  \put(-230,240){\textbf{B}}
\end{picture}
    \caption{ \textbf{Accuracy of estimated marginal distribution parameter in simulation.} Total variation distance of estimators of $\alphavec$ versus true value of $\Pr(S=s\mid R=r)$. \textbf{(A)} Relationship to number of iterations and number of surnames, fixing $\msampsize=50,000$. \textbf{(B)} Relationship to sample size $\msampsize$ and number of surnames, fixing number of iterations to 2000. MLE generally out-performs. This makes sense because there is no partial pooling happening for the $\alpha$ estimation and the MLE is the MVUE. Each point is an average over 5 independent replicates.}\label{fig-sampler-sim-alpha}
\end{figure}

\begin{figure}[h]
    \centering
    \includegraphics[width=\linewidth]{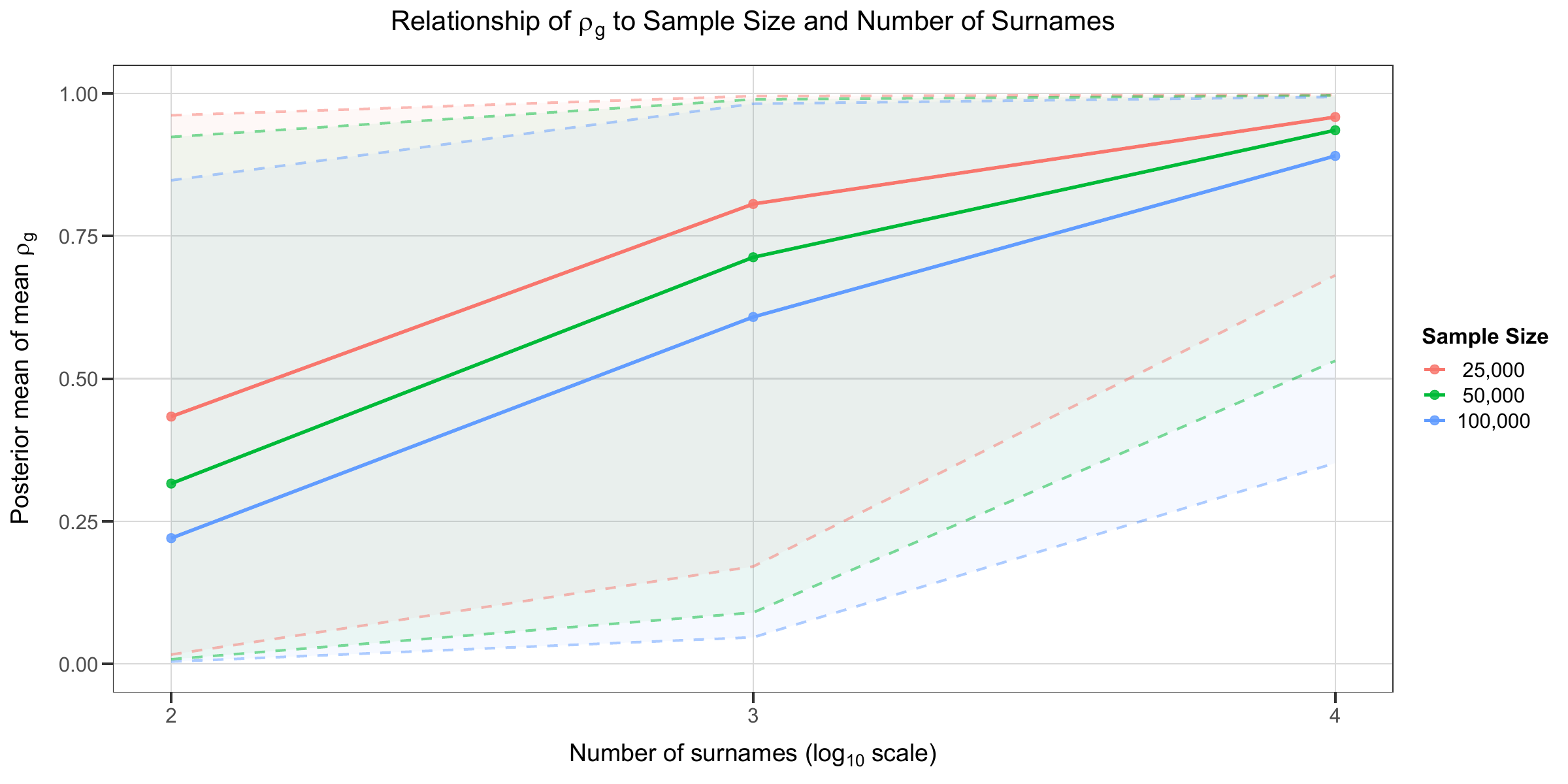}
    \caption{\textbf{Relationship of partial pooling parameter to number of surnames and sample size in simulation.} Summaries of $\widehat{\rho}_g$ posterior quantities and their relationship to sample size $\msampsize$ and number of surnames $|\cS|$. Bold lines are the mean of the $\bar{\rho}^{(b)}$ over 2000 iterations and over the different $\alpha$ initialization options. Dashed lines are the mean values of ${\rho}^{(b)}_{min}$ and ${\rho}^{(b)}_{max}$.}
    \label{fig-rhoplot}
\end{figure}
   
\clearpage 
\section{Materials}\label{supplement-data-detail}

This supplement provides detail on the data processing and data validation we ran prior to running our survey and then on the final survey. Section~\ref{app-presurveydatasource} gives details on each data source and how we processed it. In Section~\ref{sec-validation}, we describe our validation checks that are used to evaluate whether the obituary data contain the expected trends. In Section~\ref{app-samplingproc-summary}, we summarize the sampling procedure. Section \ref{sec-designimplementation} gives more detail on the design and results of our survey of Jewish Americans.

\subsection{Pre-Survey Data Sources}\label{app-presurveydatasource}

\subsubsection{Data from American Jewish Population Project}\label{app-AJPP}

Brandeis University's American Jewish Population Project (AJPP) provides regional estimates of the Jewish population. It has been used by Pew and others for identifying the geographic distribution of the American Jewish community. For its 2020 estimates, Brandeis aggregated 266 surveys of U.S. adults from the previous 5 years. In total, these surveys contained 1.3 million respondents and among them 32,300 Jewish identifying respondents. The AJPP uses this sample to model Jewish identity based on a variety of geographic (e.g., population density) and demographic (e.g., education, age) variables. As detailed in its reporting, the AJPP uses multilevel regression and post-stratification (MRP) methods to estimate Jewish people by geographic area \cite{tighe21}.

The AJPP data provide per-county estimates of the number of Jewish adults, but we aggregated these to the state level. This aggregation made learning the already high-dimensional surname by geography distribution more stable and reduced ambiguity around how to assign a funeral home to a geography. For example, if two neighboring counties contain only one Jewish funeral home between them, we did not want to act as if we had no observations of Jewish people from the county with no Jewish funeral home. 

We used the point estimates of the number of Jewish adults in each state to estimate $\Pr(G=g\mid R=1)$ as the proportion of total Jewish adults that are in each state $g$. To estimate  $\Pr(R=1\mid G=g)$, we used the AJPP data's ``all adults" point estimate to calculate the proportion of Jewish adults among all adults in state $g$. Although the AJPP data also include upper and lower bounds on these counts, we do not incorporate these into designing the sampling probabilities because it is not clear what role the interval should play in calculating the final sampling probability. We also expect that these estimates are not the main source of error in our probability estimates, and figure~\ref{fig-AJPPjewpropestimates} indicates that the bounds do not lead to dramatically different proportion estimates. 

\begin{figure}[!b]
    \centering
    \includegraphics[width=.77\linewidth]{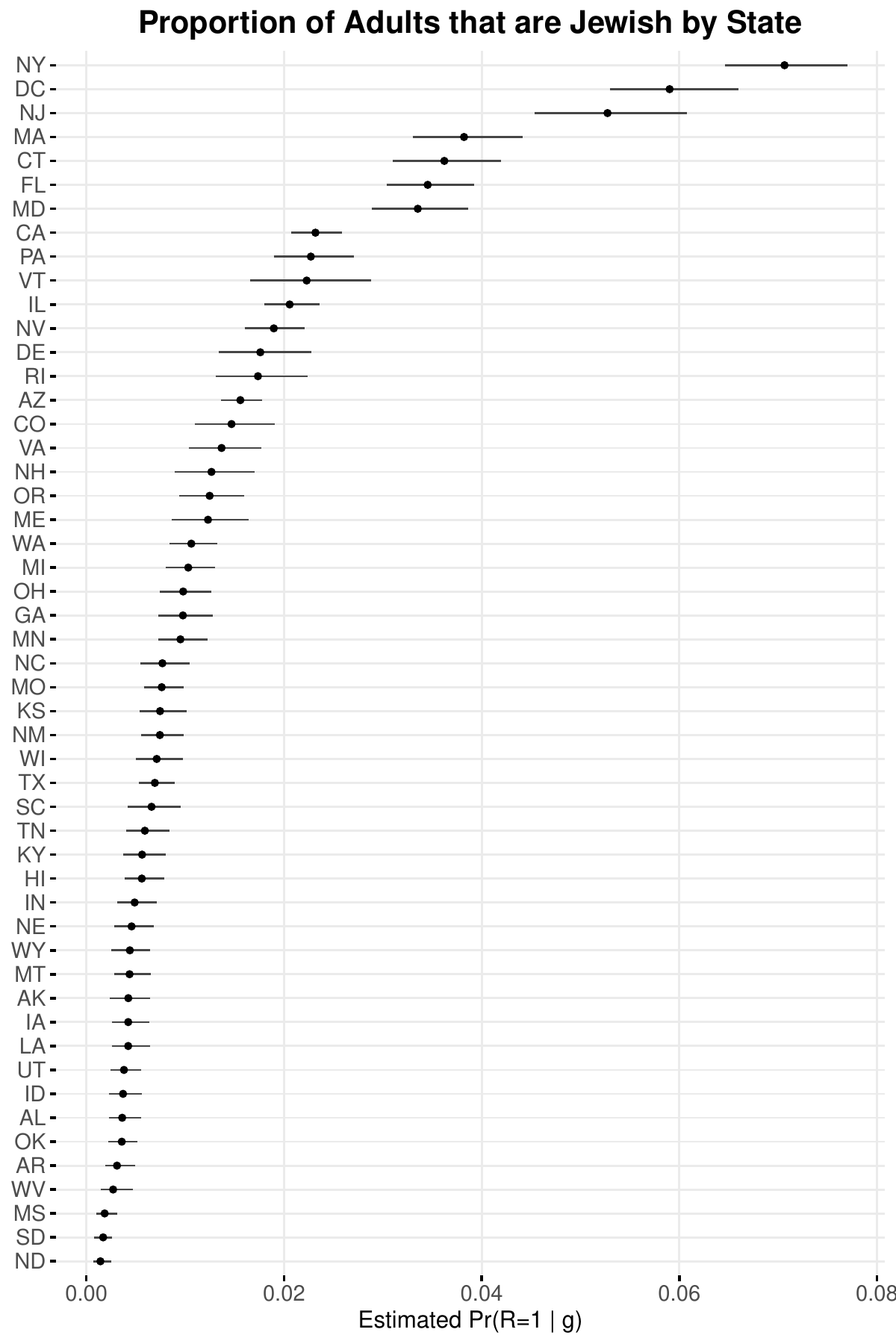}
    \caption{\textbf{AJPP estimates of the proportion of Jewish adults in each state.} Intervals are the upper and lower AJPP estimates  of the number of Jewish adults divided by the AJPP all adults estimate for that state.}
    \label{fig-AJPPjewpropestimates}
\end{figure}

\subsubsection{Obituary Data}\label{app-obit-data}

\subsubsection*{Data collection}

No official list of Jewish funeral homes exists. We thus compiled names of Jewish funeral homes from the following sources:  the Jewish Funeral Directors Association (\url{https://iccfa.com/jfda/}), the Independent Jewish Funeral Chapels (\url{https://www.nijfd.org/}), the National Association of Chevra Kadisha (\url{https://nasck.org/}), the Dignity Memorial Trusted Jewish Funeral Homes (\url{https://www.dignitymemorial.com/funeral-homes/jewish}), and Google Maps. We investigated each funeral home and excluded those that serve multiple faith traditions or make no mention of Jewish rites. 

Some funeral homes had to be excluded because they had generic names that also corresponded to names of funeral homes not used (or not exclusively used) by Jewish families. Other funeral homes were excluded because they represented facilities that appeared to be no longer operative during the time of our data collection. Still other funeral homes did appear to be operative but did not have records with those facilities listed in the PBI database, either because their data was not accessible to PBI or PBI did not collect it. For funeral homes with multiple branches within the same metropolitan area, we consolidated the branches for the purpose of our searches in the PBI database.  Ultimately, we drew obituary data from 107 funeral homes that were mentioned in obituaries in the PBI database.

The geography of these funeral homes corresponds to the geography of the Jewish population, according to the AJPP statistics. For instance, the plurality of the funeral homes in our sample (28) are in New York. The next most common states are New Jersey (14), California (8), and Florida (8), which correspond to the four states with the largest Jewish populations. These four states contain approximately 54\% of the Jewish population in the United States and 48\% of 107 funeral homes in our sample. Other states with moderate-sized Jewish populations are also well represented in the funeral home database, such as Connecticut (6), Pennsylvania (5), Illinois (4), Massachusetts (4), Maryland (3), Michigan (3), Missouri (3), Ohio (3),  and Texas (3). The remaining states with a funeral home in our study are Arizona, Colorado, Washington DC, Delaware, Georgia (2), Maine (2), Minnesota, Nebraska, Nevada, Oregon, Rhode Island (2), and Washington State.

The obituary data were compiled from the PBI Research Services database of obituaries for individuals who died between January 1, 2000 and December 31, 2023 from one of the 107 Jewish funeral homes.  We assume that these individuals were Jewish. From this database, we were able to directly obtain (without any text processing) the date of death, date of birth, first name, surname, and city and state at time of death. We also obtained the text of the obituary. These variables form what we refer to below as \textbf{the dataset on deceased individuals}.

In addition, we used the ChatGPT API (GPT-4 model) in October 2024 to process the text of the obituaries to extract additional information. The text was fed to ChatGPT along with the instructions below.  The returned JSON files were then processed and combined into one large file, resulting in a spreadsheet in which each row represents a name listed in an obituary that is associated with a relationship (e.g., grandchild) and the deceased person's information. For the purpose of this project, the main variable of interest was the first names of all people mentioned in the obituary text. We refer to this dataset as \textbf{the dataset of extracted first names.}

For ease of parsing and to only have ChatGPT analyze a single text block for each obituary, we prepended the above variables into the obituary text as five sentences: "The deceased is named \texttt{First name, Surname}. The deceased died on \texttt{Date of death}. The deceased was age \texttt{Date of death - Date of birth}. The deceased was born on \texttt{Date of birth}. The deceased's location is \texttt{City, State}." 

We used the following prompt for ChatGPT to extract the first names:

\begin{quote}
\texttt{Extract the first names and relationships of the family members mentioned in the following obituary entries, along with the deceased's name, date of death, age at death, birth date, and location. The deceased's name, date of death, age, date of birth, and location can be found in the first five sentences. If the date of birth and age are missing from the first five sentences, check to see if either is in the obituary text, and extract it if so. Format the date of death and date of birth as YYYY-MM-DD if appropriate.
}
\vspace{.5cm}

\texttt{Remove (late) from all entries. Names in parentheses are spouses of family members. For example -- 'David (Andrea Bozoki) Annis' is two people -- David, son, and Andrea Bozoki, daughter-in-law. If there is only one name in parentheses, that is the first name of the spouse who shares the same last name as the listed family member. For example, 'Larisa (Vlad) Zaslavsky' should be Larisa - daughter and Vlad - son in law. Remove all nicknames for the deceased -- for instance, if someone is listed as Morton 'Mel' Smith, just list them as Morton Smith. Remove all titles like Dr, Mr, Mrs, Rabbi, etc from relatives -- those are not first names. Please include all family members mentioned, including in laws. If any of the information is missing, please put 'MISSING' in quotes.
}
\vspace{.5cm}

\texttt{If there is more than one relationship for an obituary, make sure the JSON is valid and it is nested. In other words, use [] to delineate. Validate and ensure the JSON is valid for each entry. Ensure each JSON output has the deceased name, location, and date listed. Make sure only single quotation marks (') are used around strings. Make sure there is one relationship listed per person -- it is very important to not list 'grandson, grandson' for example, and there should be no nested information under `relationship` -- each relative should be nested on the same level, as in the example given. There should only be one person listed per entry -- if, for example, an obituary has Michelle and Alison listed as daughters, put them in two rows. No person should have multiple relationships -- someone cannot be a wife and daughter of the same person, for example. Check this carefully. Make sure there is only a single JSON file for each entry, not multiple JSON files. Remove all apostrophes from names -- eg, Ma'ayan should be Maayan.
}
\vspace{.5cm}

\texttt{If there are no relatives in an obituary, put 'MISSING' for the person name and 'MISSING' for the relationship, again in quotes.
}
\begin{verbatim}
Output your results as JSON in the following format: {
     'deceased': ,
     'death_date': ,
     'age': ,
     'birth_date': ,
     'location': ,
     'person': [{
       'first_name':,
       'relationship':
     }]  
     }    
\end{verbatim}
\end{quote}

\subsubsection*{Surname processing}

Only the dataset on deceased individuals was used as a source of surname data. The following steps were then taken to clean the surname data. The results are summarized in Table \ref{tab-obitdataprocessing}.

\begin{table}[H]
    \centering
    \begin{tabular}{l|cc}
       &  \textbf{Rows} & \textbf{Unique Surnames} \\
     \hline 
    Before processing &  232,585    & 50,238 \\ 
    After processing & 233,365 & 49,198  \\ 
    \hline 
    \end{tabular}
    \caption{\textbf{Summary of processing of surnames of deceased people.}}
    \label{tab-obitdataprocessing}
\end{table}

\begin{enumerate}
    \item \textbf{Surname length}: We found 25 instances of the surname ``AH" which is short for ``alav hashalom," a Jewish honorific frequently found in obituaries that was sometimes accidentally scraped as the surname. Because of the small number of cases, we were able to look up and manually replace AH with the actual surname. There were no 1-letter names. Other 2-letter surnames and instances of very long surnames were real names.

    \item \textbf{Hyphenated names}:  The data contained $1,107$ unique hyphenated names representing $1,134$ observations. $1,129$ observations contained a single hyphen while 5 contained two hyphens. Splitting on hyphens results in $1,688$ unique name ``components'' of which 1,284 (76\%) were names that already occurred elsewhere in the dataset. To increase our counts for these surnames, we added `pseudo-rows' for each surname component to the dataset but only if the components already occurred in the dataset.  This was done to avoid adding names that might not be distinctively Jewish. Specifically, given a name ``AAA-BBB'', if ``AAA'' and ``BBB'' occurred elsewhere in the dataset, we replaced this row with a row with just ``AAA'' and just ``BBB.'' If the components did not exist in the data, we left just the name ``AAA-BBB.'' This is why the number of rows increases in Table \ref{tab-obitdataprocessing} 
\end{enumerate}

\subsubsection*{First name processing}

For first names, the dataset on the deceased individuals and the dataset of extracted first names were pooled to create a single dataset on first names. Though this dataset does not really represent a random sample of Jewish first names, the hope is that it still gives a data-driven indication of more and less common names. We applied the following steps to both the first names of the deceased people and the text-extracted first names. 

\begin{enumerate}
    \item \textbf{NA}: dropped 97,816 rows from the ChatGPT processed data where it retrieved no first name 
    \item \textbf{Initials:} replaced any first names with a single character and a space with just the name portion. For example, ``E Adele'' or ``Adele E'' would become Adele. This reduced the number of unique names by 533 among names of deceased people and did not affect the extracted first names. This was also applied to names from the voter file. 
    
    \item \textbf{Non-Names}: By exploring examples, we built a list of common accidental extractions in the ChatGPT extracted names and removed them. These included honorifics, words signaling family relationships, and obituary related words. The list of filtered out non-names was:
\small 
  \begin{verbatim}
c("MOTHER","FATHER","STEPSON", "DAUGHTER","^PARENT$", "^SONS$", "STEPFATHER", 
"STEPMOTHER", "STEP-", "GRANDFATHER","GRANDPA","GRANDMA", "^GREAT-", "^GREAT$", 
"^UNCLE$","^AUNT$","^COUSIN$", "^COUSINS$", "^CANTOR$", "^NEE[^a-zA-Z]$", 
"FIANCE","FIANCé", "SPOUSE", "^FRIEND", "^DOCTOR", "^PROFESSOR$", 
"^MR$", "^MS$", "^MISS$","^MRS$","^JR$", "^DR$", "^PROF$", "^SR$", "^JR$", "^MD$",  
"^PHD$", "^REV$", "^HON$", "^ADM$",   "^GEN$", "^LT$", "^COL$", "^SGT$",   "^CPT$",
"^ESQ$", "^ST$", "^FR$", "^PR$", "^CEO$", "^CFO$", "^COO$", "^HUSBAND$", 
"^DEVOTED$", "^CHERISHED$", "^LOVING$", "^DEAREST$", "^DEARLY$", "^GRIEF$", "^MOURN", 
"^BRIEF$", "^ONE$", "^TWO$", "^THREE$", "^FOUR$", "^FIVE$", "^SIX$", "^SEVEN$", 
"^EIGHT$", "^NINE$", "^TEN$")
\end{verbatim}
\normalsize

For the extracted first names, this  resulted in removing 99343 rows and 69 unique names. This had no effect on the first names of the deceased people.

\item \textbf{Nicknames}:  we used a manually created lookup table of 127 common nicknames to consolidate counts. For example, we grouped ``Ben" under ``Benjamin," ``Tommy," ``Tommie," and ``Tom" under ``Thomas" and ``Zach", ``Zack," ``Zak," and ``Zac" under ``Zachary". This processing was applied to all obituary and voter file first names.
\end{enumerate}

\noindent Table \ref{tab-firstnameprocessing1} summarizes the result. After processing, the first names of the deceased and the first names extracted from the obituaries were combined into a single \textbf{pooled first name dataset} resulting in 38,814 unique first names. Of these, all but 3,064 occur as first names at least once in the voter file, indicating they are real first names. Based on manual review, we find the 3,064 are a mix of extraction mistakes of names likely to really be surnames, hyphenated names (336 cases), nicknames, and spelling errors of common names.

\begin{table}[H]
    \centering
    \begin{tabular}{l|cc}
       &  \textbf{Rows} & \textbf{ Unique First Names} \\
     \hline 
    Data on deceased, before processing & 233,365 & 9,657  \\ 
    Data on deceased, after processing & 233,320  & 9,108    \\ 
    \hline 
    Text extractions before processing &  2,315,403    &  36,056 \\ 
    Text extractions after processing &  2,216,036 &  35,965  \\ 
    \hline 
    \end{tabular}
    \caption{\textbf{Processing summary for datasets of first names. } The two datasets are first names of deceased people in the obituaries and first names od people mentioned in obituary text.}
    \label{tab-firstnameprocessing1}
\end{table}

\subsubsection{The Voter File}\label{app-voterfile}

We used the processed name information above to create sampling probabilities for individuals in the voter file. We used voter files from the research firm L2. Their state-level files combine publicly available voter registration data with matched commercial data and other modeled attributes. Only individuals who are registered to vote are in the voter file, which in the U.S. only includes citizens over the age of 18. Hence our target population for the survey excludes children, non-voters, and non-citizens. We used the latest version of each state's voter file that was available from L2 on May 2, 2025. This resulted in a sampling frame of approximately 214 million individuals spanning all 50 states and Washington DC.

\subsection{Pre-Survey Data Validation}\label{sec-validation}

In this section, we describe a few different strategies that we used to validate the data that went into estimating our BISG probabilities. 

\subsubsection{Geographic distributions}

The dataset on deceased individuals contains data from 232,585 obituaries (processed to 233,365 rows as described above) from 107 Jewish funeral homes in 30 states, 26 of which contain over 100 observations and 18 of which contain over 1000 observations with the highest number of observations (37,424) coming from Florida. The left plot of figure~\ref{fig-comparison_of_prop} indicates that the correlation between the distribution over states in the Jewish population estimated by AJPP and the state distribution of the obituaries is high in low-prevalence states and moderate in higher prevalence states, with California and New York slightly under-represented and Florida over-represented. This indicates that the obituary data are not drastically under or over representing any state where the Jewish population is high. 

The right plot of figure~\ref{fig-comparison_of_prop} compares the distribution over states in the voter file filtered to only people with surnames present in the obituary data to the AJPP estimates. This voter file distribution is not an unbiased estimate of the geographic distribution of the Jewish population but it should be correlated with it. Specifically, major under-estimation for some states (and hence over-estimation for others) would suggest that the obituary data might be missing an important sub-category of Jewish names specific to those states. That said, we should keep in mind here that because the obituary data contain some common names like ``Smith'' that are not particularly indicative of being Jewish, we might expect some shrinkage effect whereby the voter file based estimates under-estimate for the states with the highest $\Pr(G=g\mid R=1)$ and over-estimate for the states with the lowest $\Pr(G=g\mid R=1)$. We see this in the plot with the high-prevalence states NY, NJ, CA, FL all under-estimated and many lower-prevalence states a bit over-estimated relative to the AJPP estimates. Still, there is moderate agreement. 

Overall, this check did not flag any major concerns about the datasets' alignment in terms of their  implied geographic distributions for the Jewish population.

\begin{figure}[H]
    \centering
    \includegraphics[width=\linewidth]{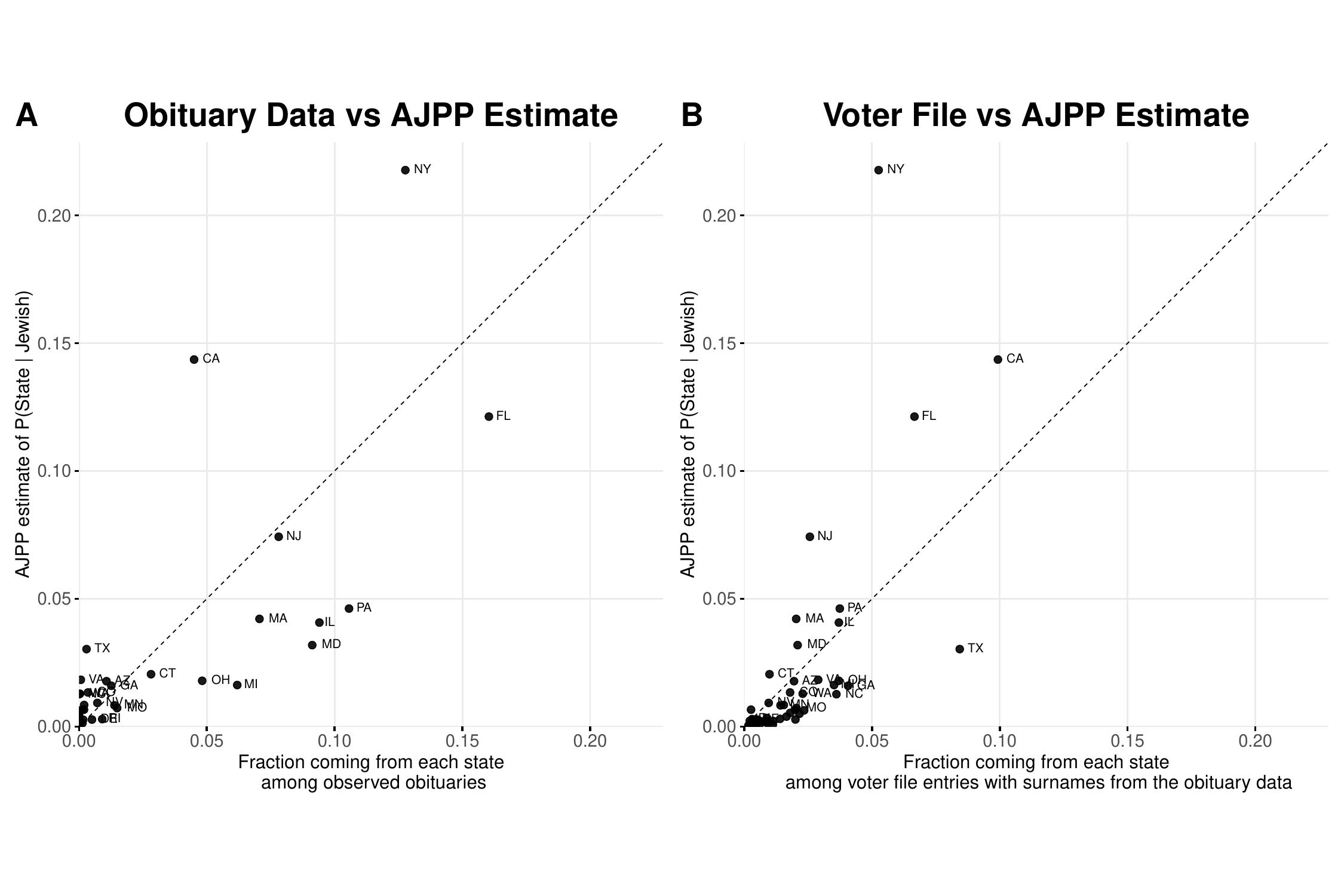}
    \caption{\textbf{Checking agreement between data sources on geographic distribution of Jewish population.} \textbf{(A)} Comparing distribution of Jewish population by state as estimated by AJPP to the allocation of Jewish obituaries by state. \textbf{(B)} Comparing the AJPP distribution to the distribution over states among people in the voter file whose surname appears in the obituary data.} 
\label{fig-comparison_of_prop}
\end{figure}

\subsubsection{Age distribution}

The over-representation of Florida among obituaries likely reflects the fact that Florida is a popular retirement destination, and the people in the obituary data tend to be much older than the general population. Their mean age at death is 82 with a standard deviation of 13 and the median age at death is 86. The median year of birth is 1926, so they are also an older cohort compared to people in the voter file who we would like to sample today. We expect the shifts in age to be more relevant for first names than surnames, as surnames tend to persist across generations while first names reflect more variable generational trends. For example, among the 9,108 unique first names from the obituary data, 2,849 of them occur only in the obituary data and not in the dataset of extracted first names. This is a main motivation for supplementing the data on deceased individuals with the dataset of extracted first names for studying first names, as these are likely to reflect younger generations. This dataset adds 29,157 unique first names that were not in the obituary dataset.

\subsubsection{Name informativeness}

To partially validate the surnames from the obituary data, on November 5, 2025, we scraped the list of ``Surnames of Jewish origin'' from Wikipedia.
This resulted in 1,460 names. Of these, 1,009 (69\%) appear in the obituary data while 451 are not. This indicates that the obituary covers a sizable proportion of names widely recognized as of Jewish origin. We also found that as one might expect, these 1,009 names have higher prevalence than average in our obituary dataset (Table \ref{tab-lastvalidation}). Although the incomplete overlap suggests we may have missed some names, we believe it also reflects the fact that the Wikipedia page includes names from all over the world and over time, some of which may not be current in the U.S. The obituary data also contain 48,189 additional surnames not on Wikipedia, reflecting that it is a much more expansive list of names that at least sometimes belong to Jewish individuals. 

\begin{table}[H]
  \centering
 \renewcommand{\arraystretch}{1.5}\begin{tabular}{p{7cm}|cc}
  & Mean Count & Median Count  \\
  \hline 
    All 49,198 surnames & 4.74 & 1 \\ 
    The 1,009 surnames also on Wikipedia & 60.85 & 12 \\
    \hline 
    \end{tabular}
   \caption{\textbf{Validating obituary surnames against Wikipedia list}. Count statistics refer to counts in the obituary data}
    \label{tab-lastvalidation}
\end{table}

We ran a similar check on the pooled dataset of first names. Many of the first names in this dataset represent family of the deceased. Though not all of these people will be Jewish, we assume that enough are so that names which are more prevalent in Jewish obituaries are also more prevalent among Jewish people in general. To help evaluate that assumption, on November 5, 2025, we scraped the names on the Wikipedia pages of Jewish given names. 
This resulted in only 388 first names, but of those 388 names, 340 appear in the first name dataset (87\%). These names also tend to have higher than average counts compared to all first names in the dataset (Table \ref{tab-firstvalidation}).

\begin{table}[H]
    \centering
\renewcommand{\arraystretch}{1.5}\begin{tabular}{p{5cm}|cc|cc}
  & Mean Count & Median Count & Mean Ratio & Median Ratio \\
  \hline 
    All 35,965 first names & 63.11 & 2 & 3.45 & 1.08\\ 
    The 340  first names also on Wikipedia & 1092.07 & 61 & 7.95 & 5.09 \\ 
    \hline 
    \end{tabular}
      \caption{\textbf{Validating obituary first names against Wikipedia list.} Count statistics refers to counts in obituary data (pooling names from deceased and names scraped from obituary text). First name ratios are the obituary data prevalence compared to the voter file prevalence. First name ratios are non-truncated for this calculation.}
      \label{tab-firstvalidation}
\end{table}

One concern is that the names on Wikipedia and names with high counts in the dataset just tend to both be more common names. To account for this, we can calculate the ratio of each first name's prevalence in the dataset of extracted first names to its prevalence in the voter file.

\begin{equation}\label{ratio-formula}
    \widehat{r}_f = \frac{\text{Proportion of entries with name $f$ in obituary data}}{\text{Proportion of entries with name $f$ in voter file}}
\end{equation}

\noindent where $\widehat{r}_f$ is an estimate of $r_f := \Pr(F=f\mid R=1)/\Pr(F=f)$.  We find that the names found on Wikipedia also have larger ratios, indicating they are relatively more prevalent in the obituary data than in the voter file. That is, based on our data alone, the Wikipedia-scraped names would generally also have been identified as distinctively Jewish. A few Wikipedia names do have ratios $<1$, but this is not necessarily wrong. For example,  Alya ($\widehat{r}_f=0.77$), is listed on the Jewish Wikipedia page but turns out to also have Slavic, Arabic, and Ancient Greek roots.

As described in Section~\ref{sec-first names}, the ratio $\widehat{r}_f$ can also be used to up-weight sampling probabilities for people with distinctively Jewish first names. To reduce instability due to small sample sizes, we only calculated $\widehat{r}_f$ over the set of 28,627 names which appear at least 10 times in the voter file or Jewish first name dataset (or both). 52\% of these yield a ratio $>1$ while 48\% have a ratio less than 1, and figure~\ref{fig-ratio-histogram} shows substantial variation in the ratios. Given that the data are not a perfect random sample of Jewish first names,  we expect some inaccuracy. However, an exploration of the ratios indicates face validity. For example, Table \ref{tab-top-name-table} gives the top and bottom 15 names by ratio among names that appear at least 10 times in the voter file, many of which also seem reasonable. In particular, we find low ratios for names associated with other religions such as Mohammad, Muhammad, and Jesus despite no prior encoding to distinguish these names. 

Ratio estimators do carry the risk of extreme behavior. We see this in some of the rarer names, such as Chickie, which has a ratio of 168 based on only a few observations. This motivates the truncation described in Section~\ref{sec-assembly} when using the ratios in sampling probabilities. 

\begin{figure}[H]
    \centering
    \includegraphics[width=0.7\linewidth]{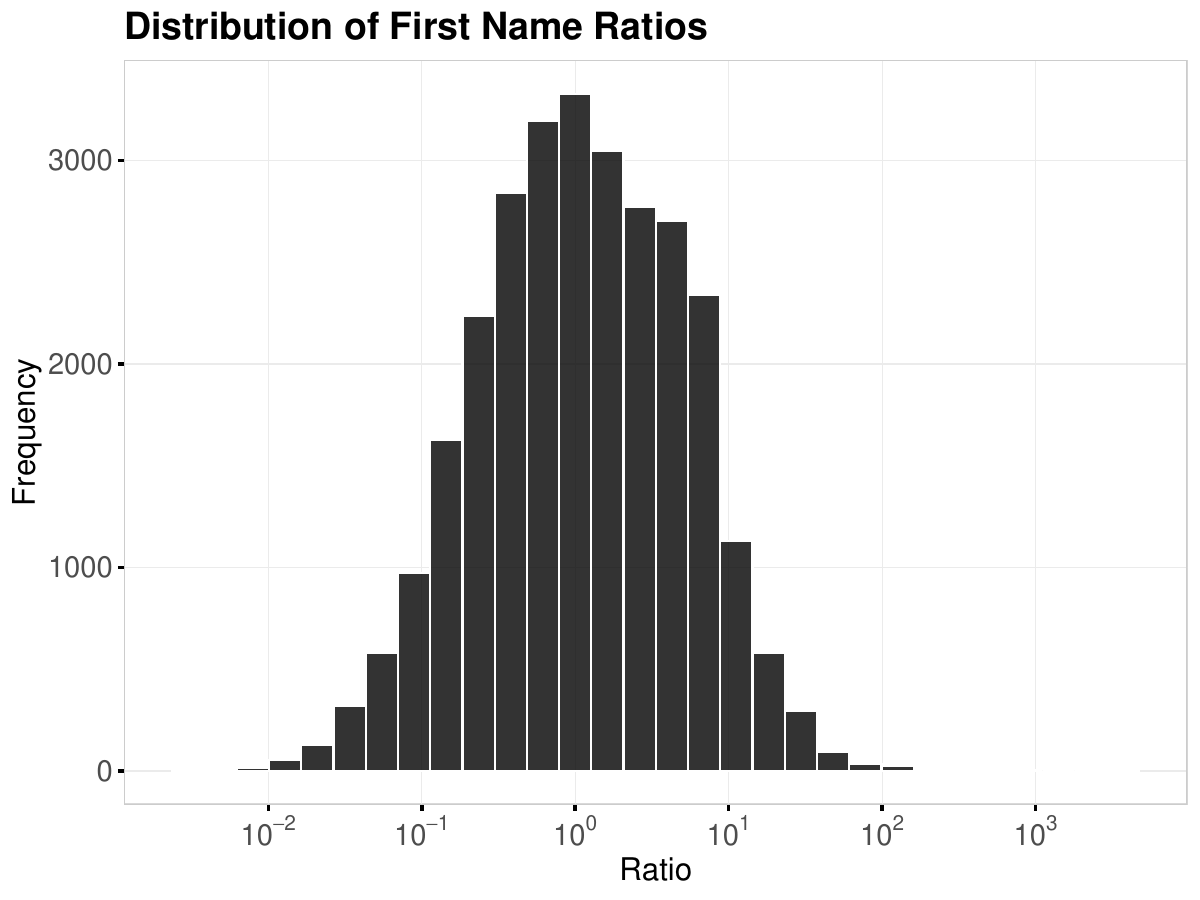}
    \caption{\textbf{Histogram of first name ratios}. The x axis is for ratios on the scale $\log_{10}(\widehat{r}_f)$.}
    \label{fig-ratio-histogram}
\end{figure}

\begin{table}[H]
    \centering
    \begin{tabular}{cc|cc}
    Top 15 First Names    & Ratio  &  Bottom 15 First Names & Ratio \\
    \hline 
    Hyman     & 269 & Mohammed &0.003\\
    Roz &  223  & Latasha & 0.003\\
    Chickie & 168   & Muhammad & 0.004\\
    Goldye   & 161 & Tamika & 0.004 \\
    Itzy  & 150   & Reynaldo  & 0.004 \\
    Myer     & 148 & Latonya & 0.005\\
    Yetta     & 147 & Esmerelda & 0.005\\
    Isadore & 141 & Abdul & 0.006\\     Tybie & 137 & Deandre & 0.006 \\
    Milt &  136 & Guadalupe &0.007 \\
    Seymour &  124 & Gerardo &0.007 \\
    Stu &  122 & Jesus & 0.007 \\
    Meyer &  120 & Lakeisha & 0.007 \\
    Gavi &  113 & Ebony & 0.008 \\
    Morrie &  112 & Fidel & 0.008 \\
    \hline 
\end{tabular}    
\caption{\textbf{Most and least Jewish first names by first name relative frequency ratios.} First names ranked  by their $\widehat{r}_f$ values: top and bottom 15 names among names that appear at least 10 times in the voter file. These are all names that appeared at least once among the deceased individuals or the first names extracted from obituary text.}
    \label{tab-top-name-table}
\end{table}

\clearpage 
\subsection{Survey Sampling procedure}\label{app-samplingproc-summary}

\subsubsection{Summary of steps}\label{sec-assembly}

We assembled our final sampling probabilities using the following steps:

\begin{enumerate}

    \item \textbf{Filtering}: Filter sampling frame to only people with surnames appearing in the obituary data. This was about 57\% of the voter file, resulting in a sampling frame with about 120,000,000 individuals
    
    \item \textbf{Surname probabilities}: estimated as described in Section~\ref{sec-surname_dist_results}. 
    
    \item \textbf{First name ratios}: we calculated ratios $\widehat{r}_f$ as in equation~\ref{ratio-formula} for all first names that occurred 10 or more times in the voter file, Jewish first name file, or both. Ratios were also truncated at 10 to limit their positive impact on the sampling probabilities. This truncation impacted 1,906 names. Ratios for any first names that did not appear in our first name dataset were set to $\widehat{r}_f=1$ to have no impact.
    
    \item \textbf{Unnormalized probability:} we calculated
    \begin{equation}
        \widehat{\pi}_{f,s,g} = \frac{\widehat{\Pr}_{\text{OBIT}}(s\mid G=g,R=1)\widehat{\Pr}_{\text{AJPP}}(R=1\mid G=g)}{\widehat{\Pr}_{\text{VF}}(s\mid G=g)}  \widehat{r}_{f,\text{OBIT}}
    \end{equation}
    where the subscripts note the main data source for estimating each part. For any surnames not appearing in the obituary data, $\widehat{\pi}_{f,s,g}=0$. Because $\widehat{r}_{f,OBIT}$ was truncated to $10$, it was possible to have a $\widehat{\pi}_{f,s,g}$ greater than 1. We truncated these to $1$. This truncation impacted only $.7\%$ of individuals in the filtered sampling frame.
    
    \item \textbf{Sampling probability:} We calculated $\probnormg=\sum_{i:G_i=g}\widehat{\Pr}(R=1\mid F_i,S_i,G=g)$ for each state $g$ and final sampling probabilities $\pi_{f,s,g} = \targetg\frac{\widehat{\Pr}(R=1\mid F=f,S=s,G=g)}{\probnormg}$. Here $\targetg$ was calculated as described in Section~\ref{sec-actual-allocations} with target sample size $\targetoverall=50,000$. The  distribution of the unnormalized and normalized sampling probabilities is summarized in figure~\ref{fig-summary_sample_probs}. As expected, the majority are small. No individual was sampled with a probability higher than $0.10$.

\item  \textbf{Sampling}: We sampled each person independently using their sampling probability. The resulting sample contained 49,546 people.
\end{enumerate}

\begin{figure}[H]
\centering
\includegraphics[width=\linewidth]{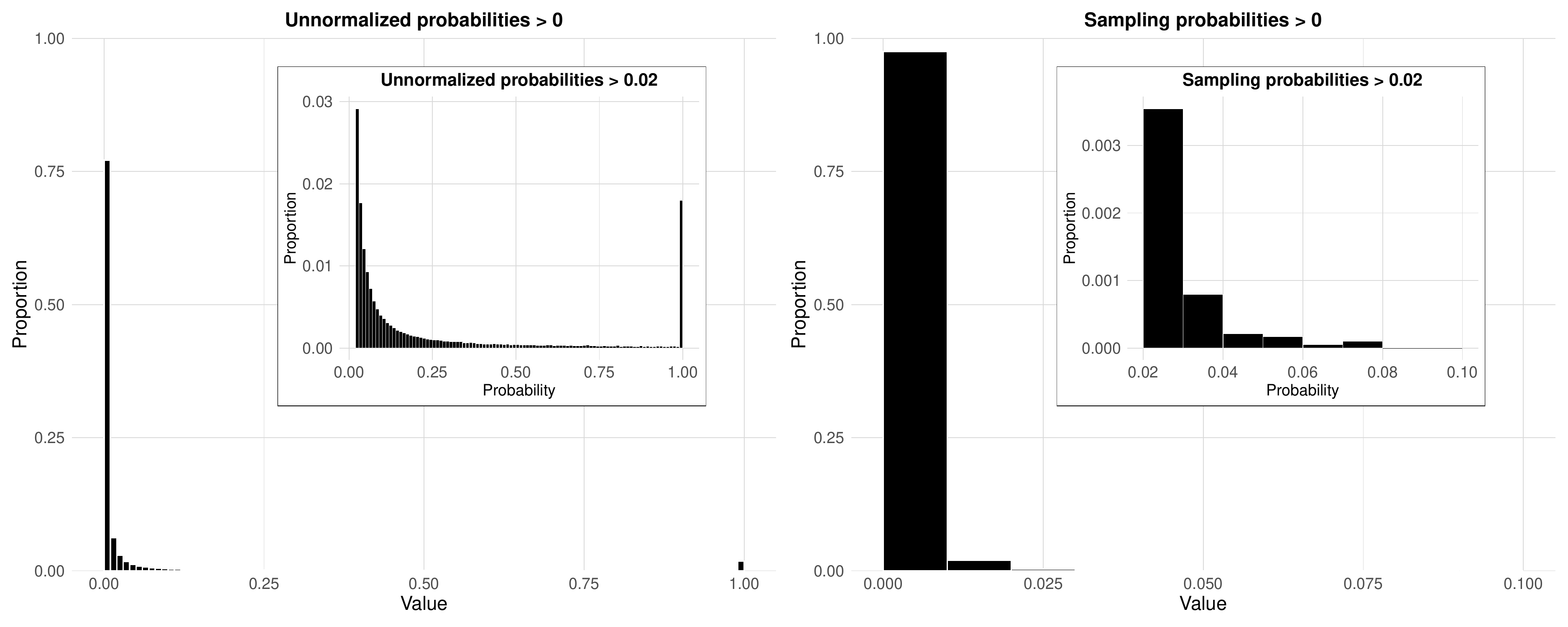}
    \caption{\textbf{Histograms of BISG probabilities and sampling probabilities.} Histograms of unnormalized estimates of $\Pr(R=1\mid F=f,S=s,G=g)$ with those $>1$ from ratio $\widehat{r}_f$ multiplication truncated at 1 and normalized final sampling probabilities. Facets zoom in on tails.}
    \label{fig-summary_sample_probs}
\end{figure}

\subsubsection{Estimating surname distributions}\label{sec-surname_dist_results}

We applied the MCMC sampler described in Section~\ref{sec-sampler} to sample from the Bayesian hierarchical model described in Section~\ref{sec-learning-surname-dist} for the obituary data with a state-level geography variable.  Table \ref{tab-mcmc-choices} summarizes our parameter choices. To determine an appropriate burn-in to discard for estimation and evaluate sensitivity to initialization, we first ran the sampler three times for 12,000 iterations. Convergence was slower than in our simulations, which is unsurprising given the larger number of surnames, but after a few thousand iterations, there was good mixing on a number of randomly selected trace plots. We also calculated the Gelman-Rubin diagnostic statistic for each $\alpha_s$ using the \texttt{coda} package in \texttt{R}, and found that on average over surnames, both the point estimates and upper confidence interval bounds calculated by that package converged towards 1 and were, within about 6000 iterations, under the common rule-of-thumb threshold of 1.1. The same was true for $\eta$, though the diagnostic values were closer to 1.1. 

\begin{table}[H]
    \centering
    \renewcommand{\arraystretch}{1.5}  
    \begin{tabular}{c|p{12cm}}
      \textbf{Setting} & \textbf{Description}\\ 
      \hline 
      $\gammavec$   & Each $\gammavec_s$ set to $\ms+1$  \\
      $\alpha$   & Initialized each $\alpha_s$ by calculating $\gammavec_s+\epsilon$ for $\epsilon \sim\text{Binomial}(100,.5)$ and then normalizing over $s$. \\
      $\eta$     & Initialized by drawing from $\eta^{(0)}\sim \text{Binomial}(100,.5)$\\
      $\eta$ prior & $\eta \sim \text{Gamma}(1,\frac{1}{100})$ (has mean $100$ and standard deviation $100$)\\
      Number of iterations & 45,000\\
      Burn-in & 15,000\\ 
      \hline 
    \end{tabular}
  \caption{\textbf{Summary of choices made when running MCMC sampler.}}
    \label{tab-mcmc-choices}
\end{table}

Based on this initial analysis, we then ran the sampler for 45,000 iterations, with a 15,000 iteration burn-in. Figure~\ref{fig-real-alpha-trace-plot} shows the $\alpha_s$ trace plots for a selection of surnames (selected to represent more and less frequent ones) and for $\eta$ including the burn-in iterations. On average, the mean fraction of accepted $(\alpha_i,\alpha_j)$ pair proposals per single iteration was $0.68$. As shown in figure~\ref{fig-eta-trace-plot}, the mean acceptance proportion for the $\eta$ proposals was, as in our simulations, quite low at $0.022$ after an initial period of larger movement in the chain.

We then calculated the posterior means using the empirical version of equation~\ref{eq-postmean} for the chain of 30,000 samples. Main text figure~\ref{fig-rho-postmean} shows the degree to which these were based on the overall proportion or the state-specific one for each state.

Finally, as described in Section~\ref{sec-bounds}, after this estimation, we imposed an upper bound derived from the AJPP and voter file data on the $\widehat{\theta}_{gs}$ estimates by setting any that exceeded the bound to the bound value. This correction was necessary for 47,767 of the state-surname probability estimates (about $2\%$ of state-surname combinations). Additionally, about 62\% of surname-state combinations did not occur in the voter file, meaning that we knew $\Pr(S\mid G,R=1)=0$ for our sampling frame. Estimates of $\Pr(S\mid G,R=1)$ for these combinations played no role in sampling.

\medskip

\paragraph{Error note.} Our original implementation of the MCMC sampling algorithm accidentally omitted from the acceptance probability for $\eta$ two terms: (1) the $\eta'/\eta$ from the Jacobian in equation~\ref{eq-eta-proposal} and (2) the power $L$ on the ratio $\Gamma(\eta')/\Gamma(\eta)$ in  equation~\ref{eq-peta-accept}. These errors only affected the $\eta$ update part of the sampling algorithm. Because we picked fairly small $\eta$ updates, the $\eta'/\eta$ ratio will often have been fairly close to $1$. 

After noticing this error, we re-ran all affected simulations in Section~\ref{sec-sim} and found that the simulation plots were indistinguishable from before the correction with one exception: the magnitudes of the scaling factor $\eta$ were a bit larger after the correction. Because this error was caught after already running the survey, we have not re-run the sampler or updated the plots in this section. That is, the current section still reflects the actual probabilities we estimated. 

Based on our simulations, we suspect that had we run the entire process with the correction, the only difference would have been slightly stronger smoothing parameters $\rho_g$ and hence slightly more tendency for our state-specific estimates to resemble the overall surname proportions. The downstream validity of our sampling using the estimated probabilities and analysis of the results is not affected by this error.

\begin{figure}[H]
    \centering
    \includegraphics[width=.9\linewidth]{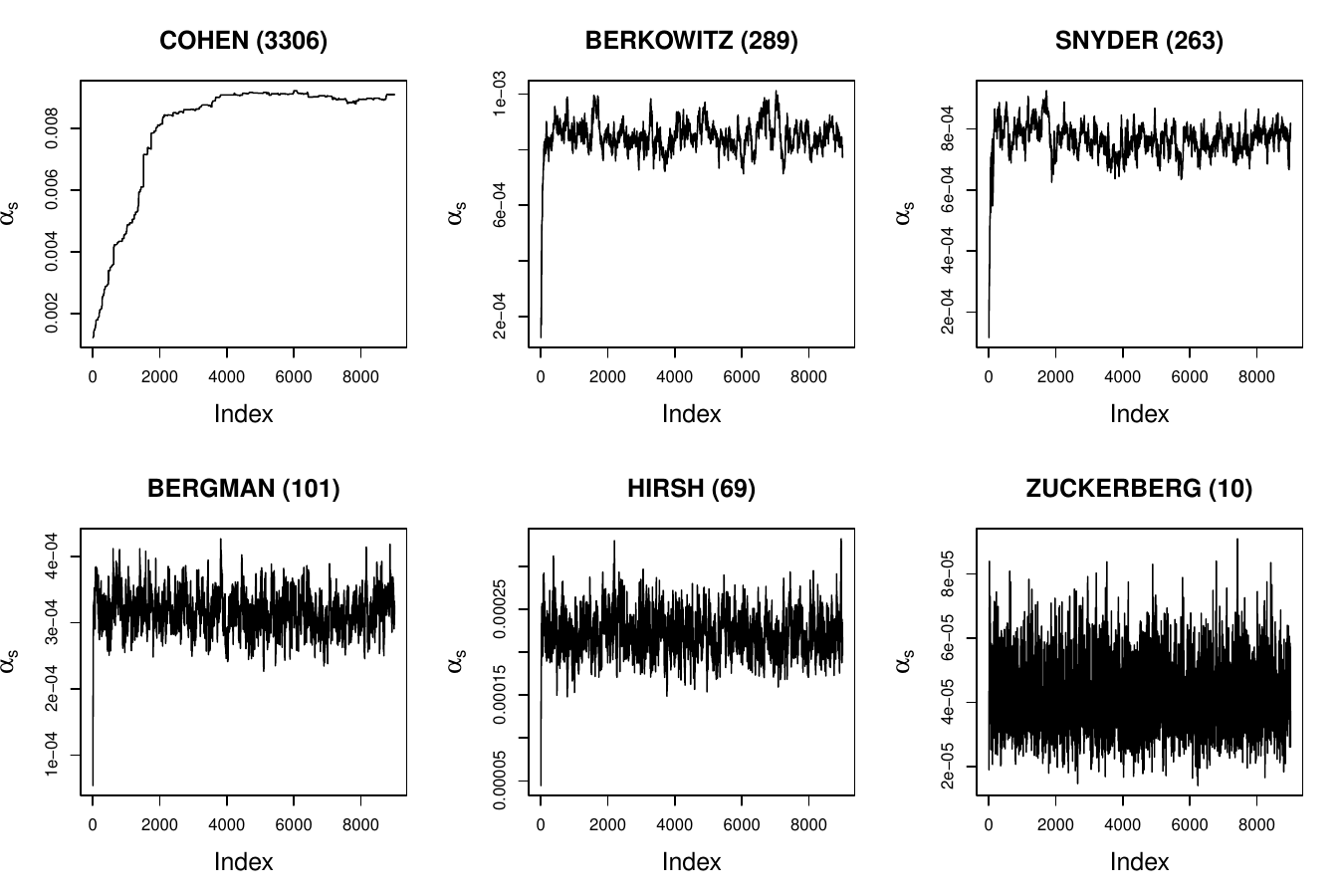}
    \caption{\textbf{Trace plot for a selection of surnames from obituary data MCMC fit}. Plots are from the 45,000-iteration run of MCMC sampler but shows every $5^{th}$ iteration. They exclude the first 500 iterations (the burn-in period used for estimates was 15,000) to allow for a scale which makes fluctuations more visible. Surnames were selected to include both common and rarer examples. Obituary data counts are noted in parentheses next to each surname.}
    \label{fig-real-alpha-trace-plot}
\end{figure}

\begin{figure}[H]
    \centering
 \includegraphics[width=.7\linewidth]{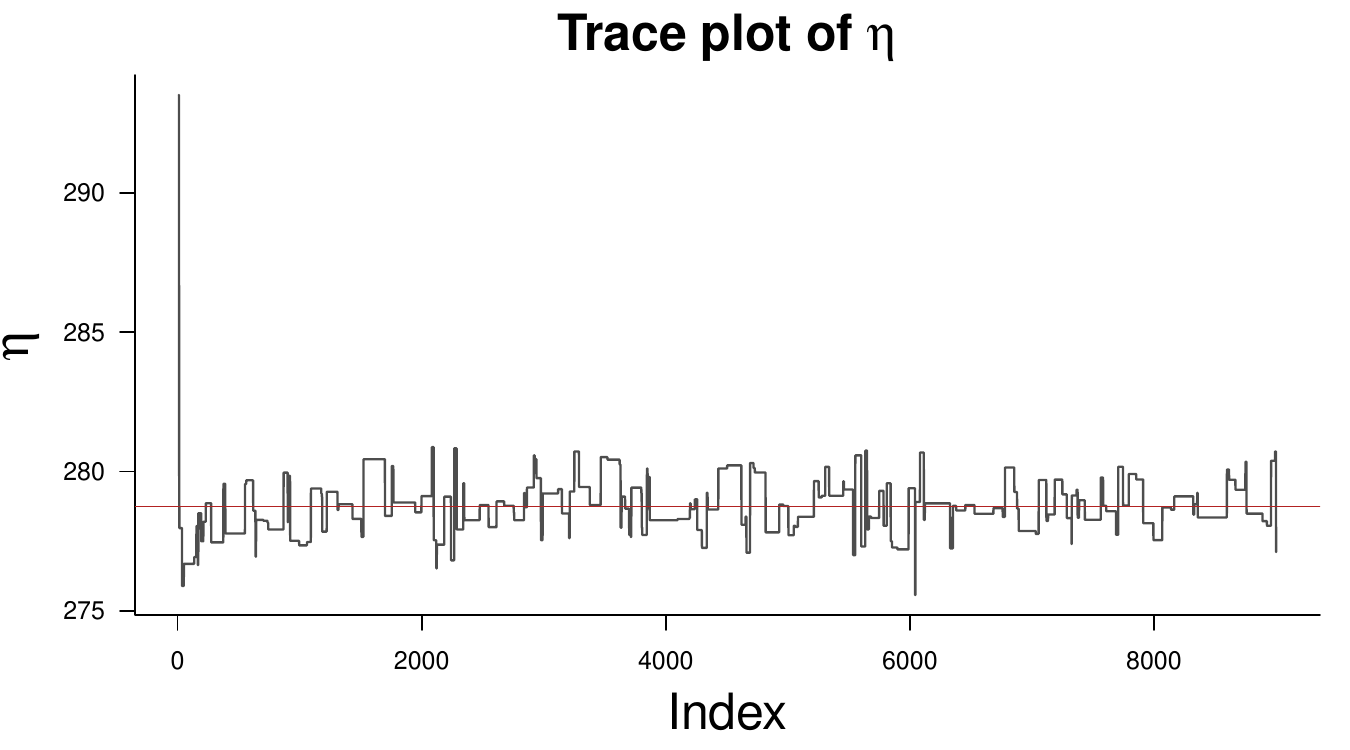}
    \caption{\textbf{Trace plot for $\eta$ parameter from obituary data MCMC fit}. Trace plot for $\eta$ parameter from 45,000 iteration run of MCMC sampler. Shows every $5^{th}$ iteration and includes the burn-in period except for the first 50 iterations, which were very volatile and make the rest of the plot hard to read. The red line marks the posterior mean. These $\eta$ values translate to more interpretable $\rho_g$ values, which are plotted for the final estimates in main text figure~\ref{fig-rho-postmean}.}
    \label{fig-eta-trace-plot}
\end{figure}

\subsubsection{Allocating stratum targets}\label{sec-actual-allocations}

figure~\ref{fig-allocation-comparison} compares the different allocation methods described in Section~\ref{app-sampling} for our application. Starting at the lower right, the last two simply allocate proportional to $\Pr(G=g\mid R=1)$ or $\Pr(R=1\mid G=g)$. The next two are the disproportionate stratified sampling formulas without filtering adjustment (using $\popoverallg$ and $\Pr(R=1\mid G=g)$) or with the filter update with $\Pr^{\mathrm{best.case}}(R=1\mid G=g,H=1)$ described in Section~\ref{app-filtering-alone}. The Poisson sampling option is equation~\ref{eq-pois-filterupdate}.

However, due to a coding error from selecting the wrong column of a data frame, in our application, we accidentally calculated $\targetg$ proportional to $\sqrt{\Pr(R=1\mid G=g)\probnormgstar}$ and did not catch this until later because it still produced reasonable $\targetg$. As shown in figure~\ref{fig-allocation-comparison}, this bug-method was thankfully highly correlated with the other options and particularly with our intended $\targetg^{\mathrm{pois}}$ calculation. We do not recommend the ``actually used'' method for the future.

\begin{figure}[H]
    \centering
    \includegraphics[width=\linewidth]{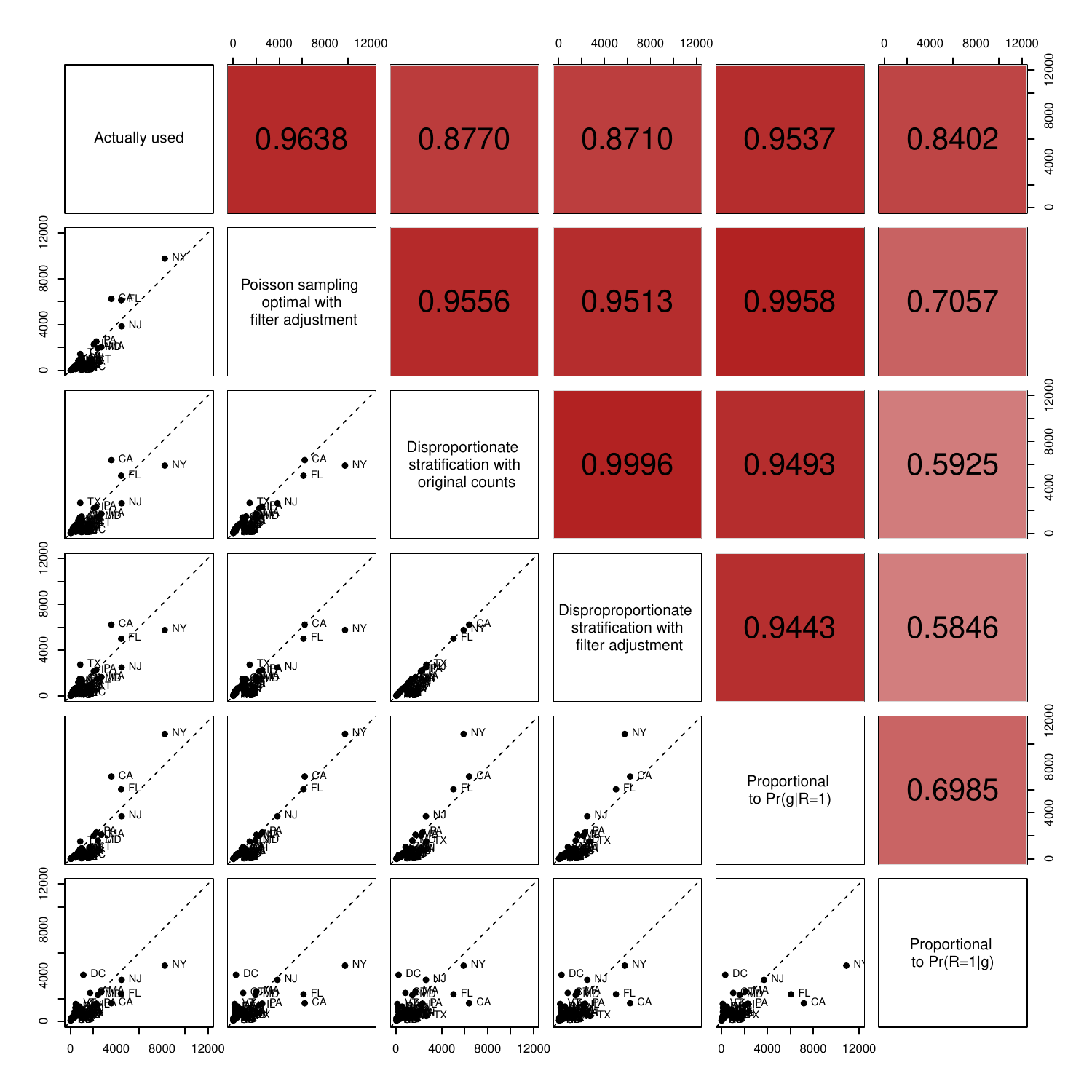}
    \caption{\textbf{Comparison of per-state allocations using different methods}. Displays scatterplots of allocations over states from applyinig the methods discussed in Section~\ref{sec-opt-pois} to Jewish sampling application. }
    \label{fig-allocation-comparison}
\end{figure}

\clearpage 
\subsubsection{The choice to use state-level rather than county-level geography}\label{app-statelevelchoice}

\begin{figure}[h]
    \centering
    \includegraphics[width=.67\linewidth]{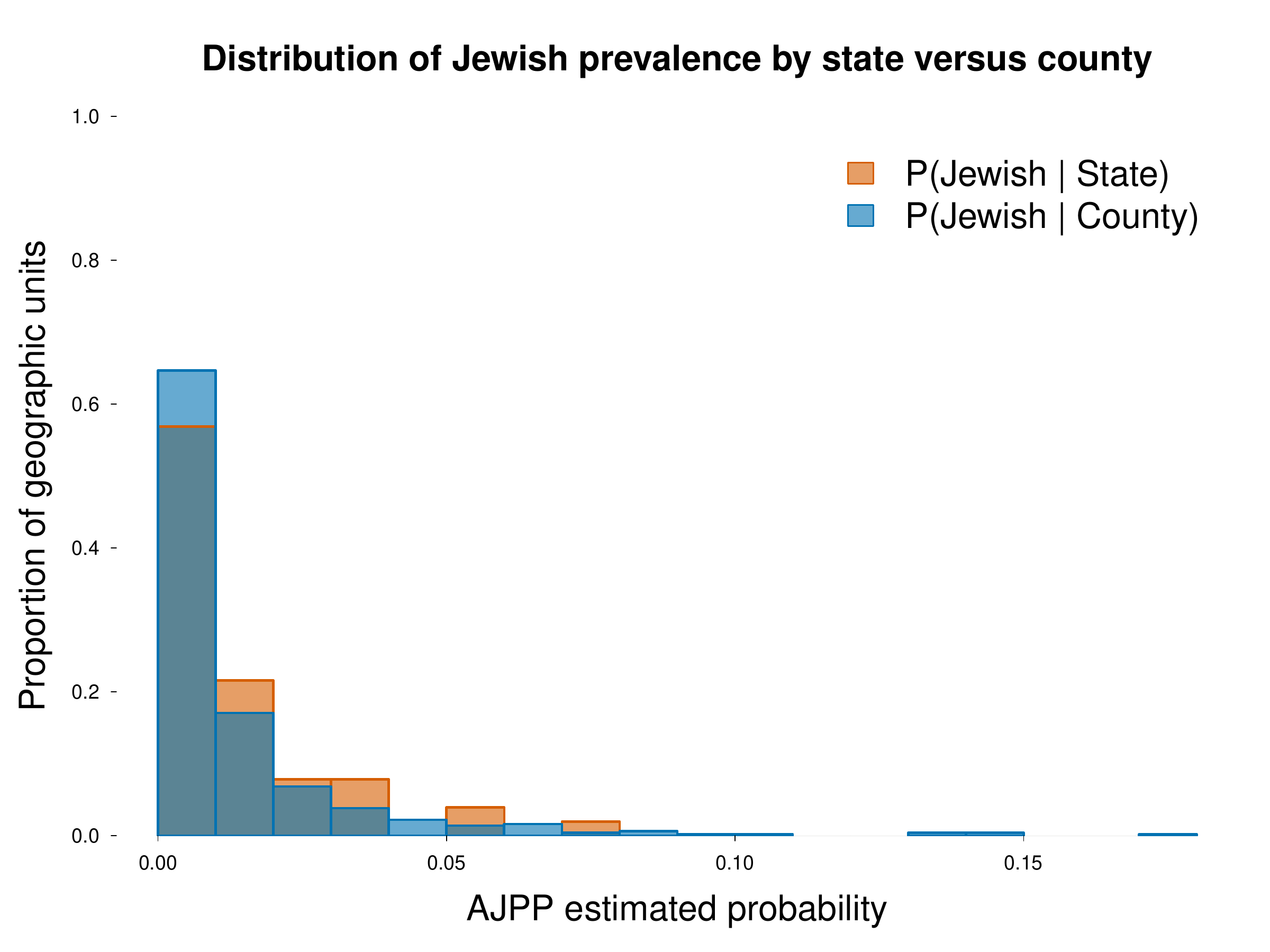}
    \caption{\textbf{State vs county Jewish prevalence comparison.} Comparison of the distributions of AJPP's estimated Jewish prevalence probabilities across states or across counties (or groups of counties). }
    \label{fig-countyvsstate}
\end{figure}

In principle, it would have been possible to do all of the above at the county level. The AJPP estimates also include estimates for counties (or sometimes groups of neighboring counties), the obituary data also sometimes contain information about the person's specific city of residence, and the funeral homes can be geolocated at higher resolutions.  We decided not to do this county-level analysis for three reasons.

\begin{enumerate}
    \item \textbf{Estimation challenges:} a higher resolution geographic variable would increase the number of parameters to learn from the data. We already have fairly few surname observations from some states, and splitting them across smaller units would make these even more sparse. The AJPP estimates are also likely to have more estimation error at the county level.
    
    \item \textbf{Migration concerns}: the people in the obituary data by definition cannot be in our sampling frame, and there is a concern that family members may not live in the same places as the deceased, hurting our ability to estimate their BISG probabilities. Looking at the state level means any within-state movements of family members do not matter. For example, \cite{daw2019beyond} study white and Black families in nationally representative longitudinal study and find that most people live in the same state as their parents and other kin, while fewer live in the same county or tract or household \cite{daw2019beyond}. \cite{Choietal_2020} also find a tendency for children to live in close spatial proximity to at least some of their parents, though this varies by sociodemographic. 

    State-level analysis does not fully mitigate the shift issue. Florida is problematic because it is a popular destination for retirees who will often not be moving with their entire families. Future work might improve the BISG modeling by incorporating data on migration flows between states (assuming these transport from the general population to the Jewish population). This network information would, for example, mean that the surname distribution of New Jersey would borrow more from Florida since older New Jersey residents often move there.

    \item \textbf{Funeral home location}: It is possible that a funeral home in one county sometimes serves neighboring counties. It might be possible to formulate a model with more hierarchical levels and geographic modeling so that counties without a Jewish funeral home borrow their Jewish surname distribution more from neighboring counties but this again adds model complexity.
\end{enumerate}
\noindent These considerations lead us to believe there would be some cost and difficulty in moving to the county level, especially for learning BISG probabilities. 

On the other hand, if we take the AJPP county estimates as roughly correct, these do indicate that there could be some gain from a geographic targeting at the county level. figure~\ref{fig-countyvsstate} shows the per-county estimates are more precise than states, though the difference is mainly in the tails. The maximum county prevalence is 0.179 -- higher than the maximum state prevalence of 0.07 in New York. Looking at the per-county estimated $\Pr(G=g\mid R=1)$ indicates that about 39.4\% of Jewish people live in 27 counties with prevalence in $[0.05,0.065]$ and an additional 34.6\% live in 64 counties with prevalence in $[0.02,0.05)$. That still leaves 25.9\% living in counties with under $0.02$ prevalence. figure~\ref{fig-countyprevbystate} also indicates that within some states, there is substantial variation in prevalence -- particularly New York, Jew Jersey, Florida, Pennsylvania, and Illinois. These are the states where county-level geographic targeting would potentially be most useful. Overall, combining state-level surname distributions with county-level targeting (e.g., by assuming $S\indep \text{County} \mid R, \text{State}$) is a clear next step to potentially achieving even greater yield. At the same time, the county level does not eliminate the minority population sampling problem as the population remains quite dispersed.

\begin{figure}[h]
    \centering
    \includegraphics[width=0.80\linewidth]{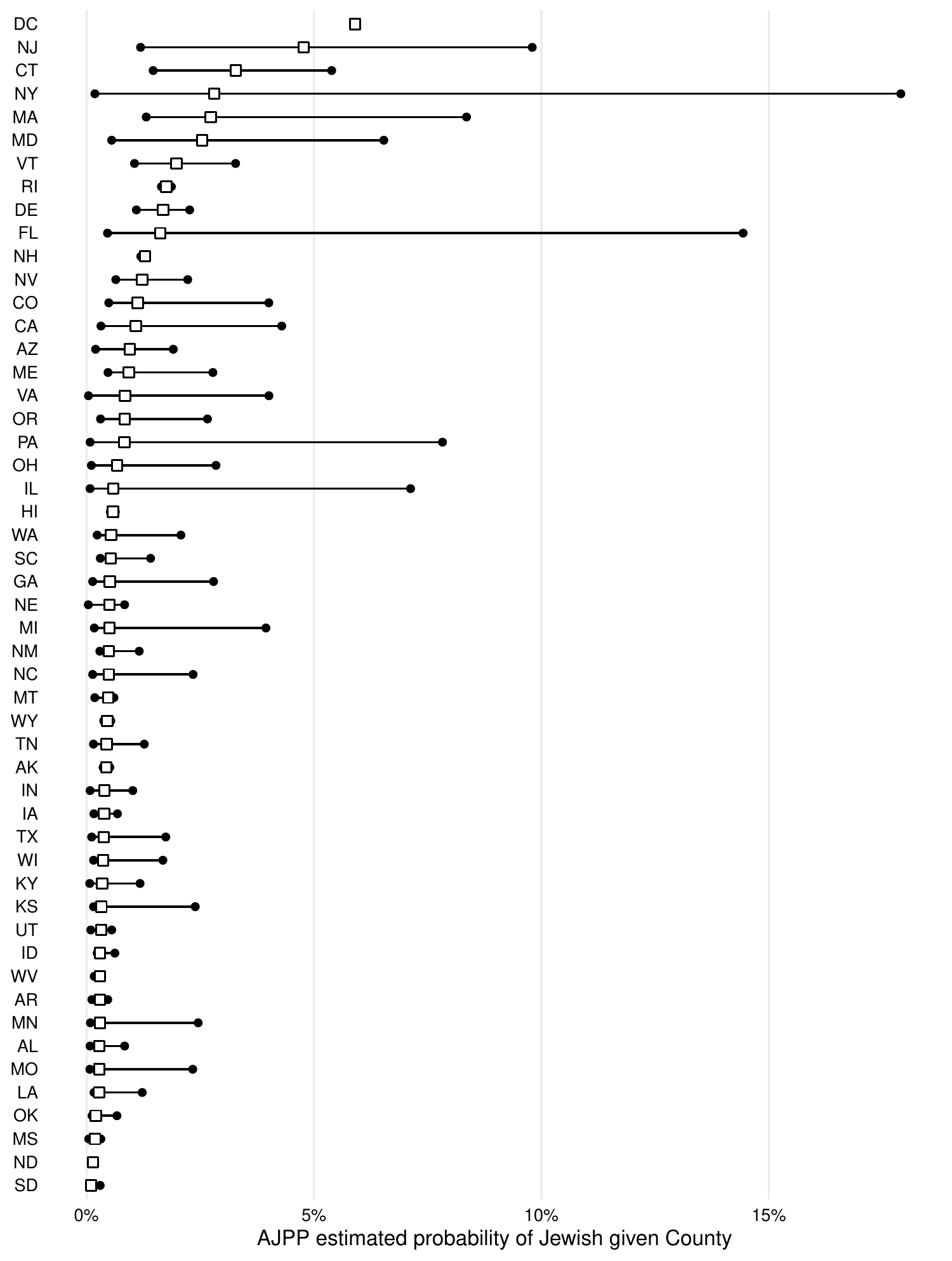}
    \caption{\textbf{Variation in per-county Jewish prevalence within each state, as estimated by AJPP.} Squares mark the median and black circles are the minimum and maximum prevalence for each state.}
    \label{fig-countyprevbystate}
\end{figure}

\clearpage 
\subsection{Survey Design, Implementation, and Analysis}\label{sec-designimplementation}

Section~\ref{app-postcard} describes the invitation and disclosures, as well as how we measured Jewish identity. Section~\ref{app-response-rate} gives details on the survey response. Lastly, Section~\ref{app-weights} checks the robustness of our results to different weighting schemes.

\subsubsection{Invitation and disclosure}\label{app-postcard}
Respondents were solicited based on a postcard mailed to their home address. One side of the postcard contained a Tufts University logo. The other side contained the following message: 

\begin{quote}
    Dear [firstname] [lastname]
I am a professor at Tufts University working in collaboration with researchers at Harvard University. As part of our research on public opinion, we are conducting a brief academic survey of U.S. adults. The survey only takes about 5 minutes to complete. The questions gauge attitudes on civic, religious, and political topics. You have been selected to participate because you are a registered voter. Please help us by taking this short survey. If you would like to learn more about the study and take the survey online now, you can do so by navigating to \url{https://as.tufts.edu/politicalscience/survey2025} and entering [survey id] as your ID number. 
With gratitude, 
Eitan Hersh, PhD
Tufts University
\end{quote}

At the end of the survey, respondents saw a disclosure that informed them of the full purpose of the study. This study design was approved by the Tufts University Institutional Review Board (STUDY00005675). 

\paragraph{Initial disclosure} When individuals began the survey, they read an informed consent page stating that the survey would gauge their civic, religious, and political attitudes, and that it had a specific further purpose that would not be disclosed until the end of the survey. Nothing yet indicated that the survey concerned Jewish identity. This was done to reduce response bias that could arise if Jews with specific demographics or views are more or less likely to contribute to a survey they know to be focused on Jews.

When the survey was first approved by the Tufts University Institutional Review Board in April of 2024, it included the following disclosures. Before agreeing or declining to take the survey, respondents saw this informed consent script:
\begin{quote}
You are being asked to volunteer for an academic research study conducted by Dr. Eitan Hersh of Tufts University in collaboration with researchers from Harvard University. The purpose of this study is to understand the civic, religious, and political attitudes. You are being asked to participate because you are a registered voter in the United States. Participation in the study entails filling out an online questionnaire. The questionnaire will take approximately five minutes to fill out.   The study has a more specific scholarly purpose that will be disclosed upon completion of the survey. That purpose is not disclosed in advance out of concern that it would bias results of the study.   

You may derive an intellectual benefit from thinking about your civic attitudes and behaviors, but the researcher offers no tangible benefits to you for participating. Participation is completely voluntary. The questions are not sensitive, are not expected to cause offense or embarrassment, and no foreseeable risks are anticipated. You are free to decline to participate, to end participation at any time for any reason, or to refuse to answer any individual question. Refusal to participate or discontinuing participation after consent will involve no penalty or loss of benefits to which you are otherwise entitled.   

I will publish articles and datasets from this questionnaire. I will take measures to protect your privacy and confidentiality. All of your responses will be held in confidence. As this study is a collaboration with researchers at Harvard University, your identifiable data may be shared with Harvard researchers for analysis in this research. Once we have gathered your responses, we will store and analyze your responses in a de-identified database. The information you provide will be kept for at least 3 years after the study is closed. You will not be named, and you will not be personally identifiable in any publication of results or data once the data are de-identified.  People responsible for monitoring this research may be able to access the data. This includes the Tufts University Institutional Review Board. Any identifying information linking your identity to your responses will be removed prior to analysis. Please note, however, that unlike information you provide to your doctor or lawyer, the investigator can be compelled by a court to disclose this information. There is no compensation for participation.   

We will merge information from your voter registration (age, history of voter turnout, party affiliation, if applicable) with your survey data, and all identifiers, including your survey participant ID code, will be removed before the data are analyzed. Once de-identified, the data could be used for future research studies without additional informed consent from you or your legally authorized representatives. Should you complete the survey and then seek to be withdrawn from the study, please contact the researcher immediately (contact information below). Soon after the data collection is complete, the researcher will produce a de-identified dataset to be used for analysis. At that point, it will not be possible to withdraw consent.   

For questions or concerns about the research study or procedures, or if you need to notify someone of a complaint, please contact the Tufts research team. The Principal Investigator is Prof. Eitan Hersh. Email: eitan.hersh@tufts.edu. Telephone: 617-627-2043. If you have questions or concerns about your rights as a research participant, or if you would like to discuss the study with someone outside of the research team, contact the Tufts University Social Behavioral and Educational Research Institutional Review Board (SBER IRB). 75 Kneeland Street, Boston MA, 02111. Telephone: 617-627-8804. Email: sber@tufts.edu. Website: https://viceprovost.tufts.edu/about-sber-irb  Do you consent to participate in this research? 
o	Yes  (1) 
o	No  (2) 

\end{quote}

For those who agreed to take the survey, at the end of the survey questions they saw the following disclosure: 
\begin{quote}
DISCLOSURE: Participants were told that the purpose of the research was to study the civic, religious, and political attitudes of U.S. adults, and that participants were recruited because they are registered voters. That information is incomplete. Participants were also not told the true purpose of the research. The true purpose of this survey is to validate a predictive model estimating which U.S. registered voters might identify as Jewish. Respondents were recruited based on estimating their probability of being Jewish based on their name and geography, given that they are registered voters. The true purpose of the study was withheld from participants to reduce bias. If the purpose was disclosed, then one's religious or ethnic identity might predict the decision to participate, which would bias the analysis. If you have questions, you may reach out to Dr. Eitan Hersh at eitan.hersh@tufts.edu.
\end{quote}

\paragraph{Updated disclosure}
In July 2025, once the survey was in the field, we received several emails and phone calls expressing concern about the survey. Several respondents either wanted to be removed from the study after they read the disclosure or expressed that our incomplete disclosure caused concern. Several Jewish respondents felt uncomfortable contributing to research that targets Jewish Americans, even if for research purposes.

We quickly worked with the IRB to adjust the language of the second disclosure. Thus, part way through the fielding of the survey, the post-survey disclosure was changed to the following
\begin{quote}

DISCLOSURE: You were told that the purpose of the research was to study the civic, religious, and political attitudes of U.S. adults, and that you were recruited because you are a registered voter. That information is incomplete. You were also not told the true purpose of the research. The purpose of the research is to validate a predictive model that has been created to better identify Jewish Americans.  The purpose of model is to allow researchers to improve the ability to conduct surveys of Jewish Americans. Jewish Americans are a distinct subpopulation in the United States, but they are hard to study due to their relatively small numbers. Scholars, as well as foundations and non-profit organizations, seek to better understand the attitudes of this subpopulation, and this model aims to make it easier for them to do so. The present research is motivated by the desire to be able to improve research that could benefit this population.

You were recruited based on estimating your probability of being Jewish based on you name and geography, given that you are registered voter. The true purpose of the study was withheld from you to reduce bias. If the purpose was disclosed, then your religious or ethnic identity might predict the decision to participate, which would bias the analysis. 

Please confirm your consent to participate in this study, or whether you would like to withdraw, by checking a box below. 
-I consent to the use of my data in this research.
-I do not consent to the use of my data in this research. Please remove and delete my data.

You may also withdraw from the study by contacting the Principal Investigator at eitan.hersh@tufts.edu Your data can be removed from the dataset up until it is de-identified, at which point it will not be possible to know which data are yours.

If you have questions, you may reach out to Dr. Eitan Hersh at eitan.hersh@tufts.edu.
\end{quote}

Hence the updated disclosure provided more details about the motivations of the study as well as an additional opportunity for respondents to opt out after having read the complete disclosure. 
Twenty-two individuals not included in our reported response count withdrew their records after reading the post-survey disclosure.

\paragraph{Measuring Jewish Identity}\label{app-identity}

 Jewish identity was measured based on a two-part question. The first part was a question that asked: ``What is your present religion, if any?'' Following Pew, we then asked a separate question, ``ASIDE from religion, do you consider yourself to be any of the following in any way (for example, ethnically, culturally, or because of your family's background)?'' On this item, we asked if the respondent identified in this way as Jewish, Catholic, Mormon, or Muslim. We count Jewish identifiers as those who are either Jewish by religion or who are atheist, agnostic or of no religion but who identify as ethnically or culturally Jewish.

\subsubsection{Other survey questions}

The survey questionnaire contained a number of questions, some of which are analyzed in a separate paper. The full questionnaire will be available in the replication archive. We note that for comparison, Pew’s 4-page screening survey and 20-page extended survey are available at \url{https://www.pewresearch.org/religion/wp-content/uploads/sites/7/2021/04/PF_05.11.21.Jewish_Survey_Mail_Screening-Questionnaire.pdf} and
 \url{https://www.pewresearch.org/religion/wp-content/uploads/sites/7/2021/04/PF_05.11.21.Jewish_Survey_Mail_Extended_Questionnaire.pdf} respectively.

\subsubsection{Response}\label{app-response-rate}

After selecting records from the voter file, we used the names and addresses listed in the voter file to send each sampled individual a postcard from Tufts University, inviting them to take 5-minute Qualtrics survey on ``civic, religious, and political topics'' on a website. Of 49,546 people sent the postcard, 1,765 or 3.6\% completed the survey (not including 22 who completed but later opted out of the survey). Of these, 1,759 were linkable back to their voter file record while 6 respondents did not fill in a necessary unique identifier to allow this. Over 90\% of responses were submitted between June 27 and July 15, 2025. Additional respondents trickled in through August and September. 

The true response rate among those who actually received the postcard is likely slightly higher. In similar recent studies that mailed survey solicitations to U.S. voters, approximately 2\% of solicitations come back as undeliverable as addressed, due to factors such as residential moves and deaths \cite{hershshah25}. We did not track undeliverable mail, so we calculate the response rate based on mail sent. 

Table \ref{tab-sampleresponsecounts} and figure~\ref{fig-hist-response-rates} (top) summarize the response rates by state. Excluding North and South Dakota, which have response rates of 0\% and 11\% respectively because we only sampled a very small number of people there, most response rates are between 1\% and 8\% with an average of 4.2\%. Table \ref{tab-sampleresponsecounts} and figure~\ref{fig-hist-response-rates} (bottom) summarize the Jewish response rates by state. Specifically, this is the fraction of respondents for each state who were Jewish. There were no Jewish respondents for Arkansas, Louisiana, West Virginia, and North Dakota.

\begin{figure}[!t]
    \centering
    \includegraphics[width=.8\linewidth]{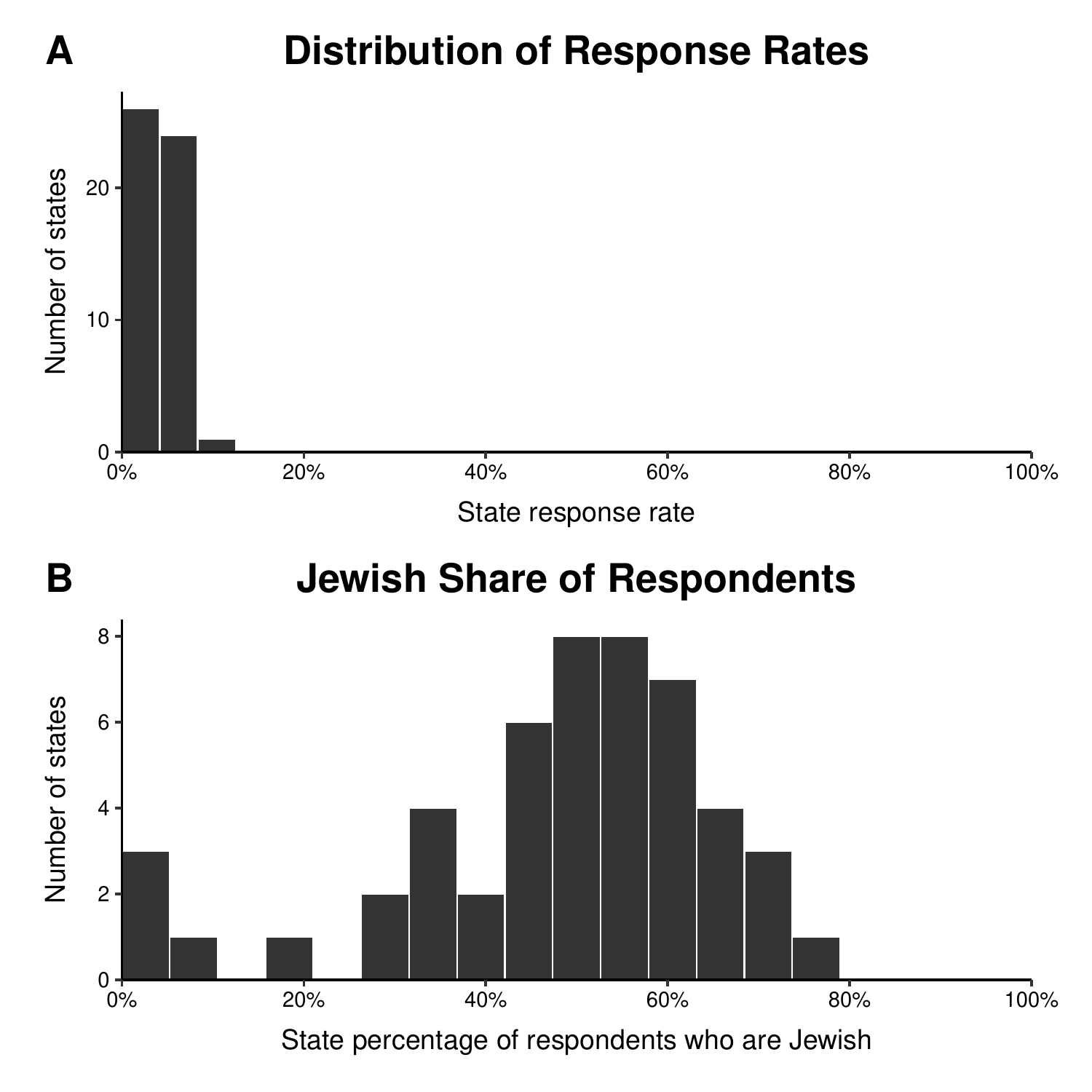}
    \caption{\textbf{Distribution of survey response rates across states}. \textbf{(A)} overall response rates (number responded / number sampled). \textbf{(B)} Jewish response rates (number Jewish/number responded) }
    \label{fig-hist-response-rates}
\end{figure}

\begin{table}[h]
\centering
\renewcommand{\arraystretch}{.8} 
\footnotesize
\begin{tabular}{lcccc} 		
State	&	Brandeis AJPP Est. (2020)	&	Sample	&	All respondents	&	Jewish respondents	\\
\hline	
NY	&	20.9	&	16.5	&	12.2	&	13.5	\\
FL	&	10.2	&	8.9	&	5.8	&	5.8	\\
NJ	&	7.5	&	8.9	&	7.3	&	8.9	\\
CA	&	15.4	&	7.1	&	6.3	&	6.3	\\
MA	&	4.4	&	5.3	&	8.2	&	10.2	\\
MD	&	3.0	&	4.9	&	5.9	&	7.3	\\
PA	&	4.6	&	4.6	&	4.8	&	5.1	\\
IL	&	4.4	&	4.1	&	3.5	&	3.2	\\
CT	&	2.0	&	3.4	&	4.3	&	4.1	\\
DC	&	0.7	&	2.2	&	3.4	&	3.8	\\
AZ	&	1.9	&	2.2	&	2.1	&	1.8	\\
VA	&	1.9	&	2.1	&	2.7	&	2.5	\\
MI	&	1.7	&	2.0	&	1.4	&	1.4	\\
OH	&	1.9	&	2.0	&	1.4	&	1.0	\\
CO	&	1.4	&	1.8	&	1.9	&	1.4	\\
TX	&	3.2	&	1.7	&	1.8	&	1.0	\\
GA	&	1.7	&	1.7	&	1.3	&	1.3	\\
WA	&	1.4	&	1.6	&	2.0	&	1.5	\\
NV	&	1.0	&	1.6	&	0.9	&	1.0	\\
NC	&	1.4	&	1.5	&	2.1	&	2.2	\\
OR	&	0.9	&	1.4	&	1.7	&	1.5	\\
MN	&	0.8	&	1.0	&	1.8	&	1.2	\\
WI	&	0.7	&	1.0	&	1.0	&	0.5	\\
MO	&	0.6	&	1.0	&	1.0	&	1.0	\\
RI	&	0.3	&	0.9	&	1.5	&	1.5	\\
All Others	&	6.0	&	10.9	&	14.0	&	11.0	\\
\hline 
\end{tabular}
\caption{\textbf{Geographic distribution in survey.} Distribution of sampled, respondents, and Jewish respondents over states. Note: Cell percentages represent the percentage of that column's group who reside in a given state.}
\label{tab-state-response}
\end{table}

\begin{table}[h]
\footnotesize 
\centering
\renewcommand{\arraystretch}{.60}
\begin{tabular}[t]{lrrrrr}
\hline 
State & Sampled & Responded & Jewish & Response rate & Jewish rate\\
\hline 
AK & 87 & 5 & 1 & 0.057 & 0.200\\
AL & 184 & 6 & 2 & 0.033 & 0.333\\
AR & 123 & 5 & 0 & 0.041 & 0\\
AZ & 1067 & 36 & 18 & 0.034 & 0.500\\
CA & 3500 & 110 & 63 & 0.031 & 0.573\\
CO & 892 & 33 & 14 & 0.037 & 0.424\\
CT & 1675 & 75 & 41 & 0.045 & 0.547\\
DC & 1079 & 59 & 38 & 0.055 & 0.644\\
DE & 450 & 21 & 13 & 0.047 & 0.619\\
FL & 4427 & 102 & 58 & 0.023 & 0.569\\
GA & 822 & 23 & 13 & 0.028 & 0.565\\
HI & 145 & 8 & 5 & 0.055 & 0.625\\
IA & 168 & 10 & 1 & 0.060 & 0.100\\
ID & 118 & 8 & 4 & 0.068 & 0.500\\
IL & 2013 & 62 & 32 & 0.031 & 0.516\\
IN & 302 & 8 & 6 & 0.026 & 0.750\\
KS & 275 & 10 & 3 & 0.036 & 0.300\\
KY & 256 & 11 & 4 & 0.043 & 0.364\\
LA & 213 & 6 & 0 & 0.028 & 0\\
MA & 2635 & 144 & 102 & 0.055 & 0.708\\
MD & 2419 & 104 & 73 & 0.043 & 0.702\\
ME & 363 & 22 & 12 & 0.061 & 0.545\\
MI & 991 & 24 & 14 & 0.024 & 0.583\\
MN & 514 & 31 & 12 & 0.060 & 0.387\\
MO & 472 & 18 & 10 & 0.038 & 0.556\\
MS & 79 & 2 & 1 & 0.025 & 0.500\\
MT & 108 & 5 & 2 & 0.046 & 0.400\\
NC & 716 & 37 & 22 & 0.052 & 0.595\\
ND & 15 & 0 & 0 & 0 & NA\\
NE & 166 & 6 & 3 & 0.036 & 0.500\\
NH & 383 & 28 & 13 & 0.073 & 0.464\\
NJ & 4413 & 129 & 89 & 0.029 & 0.690\\
NM & 258 & 17 & 8 & 0.066 & 0.471\\
NV & 798 & 16 & 10 & 0.020 & 0.625\\
NY & 8194 & 215 & 136 & 0.026 & 0.633\\
OH & 976 & 24 & 11 & 0.025 & 0.458\\
OK & 123 & 3 & 2 & 0.024 & 0.667\\
OR & 680 & 29 & 15 & 0.043 & 0.517\\
PA & 2274 & 85 & 51 & 0.037 & 0.600\\
RI & 457 & 26 & 15 & 0.057 & 0.577\\
SC & 396 & 14 & 9 & 0.035 & 0.643\\
SD & 34 & 4 & 2 & 0.118 & 0.500\\
TN & 394 & 6 & 3 & 0.015 & 0.500\\
TX & 844 & 31 & 10 & 0.037 & 0.323\\
UT & 155 & 8 & 5 & 0.052 & 0.625\\
VA & 1025 & 47 & 25 & 0.046 & 0.532\\
VT & 440 & 26 & 12 & 0.059 & 0.462\\
WA & 804 & 35 & 15 & 0.044 & 0.429\\
WI & 480 & 18 & 5 & 0.038 & 0.278\\
WV & 80 & 4 & 0 & 0.050 & 0\\
WY & 64 & 3 & 1 & 0.047 & 0.333\\
\hline 
\textbf{Total} & \textbf{49546} & \textbf{1759} & \textbf{1004} &  &  \\
\textbf{Mean} & &  & & \textbf{0.042} & \textbf{0.509} \\
\hline 
\end{tabular}
\caption{\textbf{Response counts from survey.} Counts of All Sampled, All Respondents, Jewish Respondents by State. For a few respondents, we were not able to determine their state of residence, which is why the second-to-last column shows 1,759 respondents rather than 1,765.}
\label{tab-sampleresponsecounts}
\end{table}

\clearpage 
\subsubsection{Estimating raking targets from the voter file}\label{app-bisgestimatedweights}

As described in Section~\ref{sec-estimation}, we can use the voter file and BISG probabilities to estimate $\Pr(X=x\mid G=g,R=1)$ and $\Pr(X=x\mid G=g, R=1)$ for variables $X$ present in the voter file -- here age, gender, and partisanship. For age and gender, we judge that it is more plausible to assume $X\indep R\mid S,G$ than to assume $X\indep S\mid G,R$ because there are known strong relationships between age and name and between gender and name, whereas the variation between the Jewish population and the non-Jewish population in terms of age and gender distribution (given geography and name) is likely to be smaller. This leads to the weighted estimator in equation~\ref{eq-weightedest}. 

For partisanship, we judge that it is more plausible to assume $X\indep S\mid G,R$ because there is a known direct relationship between being Jewish and partisanship. Given a Jewish or non-Jewish person from a certain geography, the remaining power of name to predict partisanship seems likely to be weaker, though there could be some violation because of how age and race correlate with both name and partisanship. To estimate the distribution of partisanship given Jewish under this assumption, we apply the EM algorithm from \cite{mccartan2025estimating} as implemented in the \texttt{birdie} package. We obtain estimates for each state and then combine them using 
\begin{equation}
    \Pr(X=x\mid R=1)=\sum_{g=1}^{L}\Pr(X=x\mid G=g,R=1)\Pr(G=g\mid R=1)
\end{equation}
where $\Pr(G=g\mid R=1)$ are the AJPP estimates.

In the voter file, age and gender each had a small fraction of records with missing values. For our initial estimation, we treat ``missing'' as a level. In 80\% of states, the fraction missing age is under 0.7\% but in the remaining 20\% of states (AK, DC, HI, MD, MS, ND, SD, WI, and WY), it ranges from about from 19\% to 39\%. Overall, the estimated fraction of the Jewish population missing age is 1\%. For gender, the fraction missing is under 3.7\% in all states and the final estimated fraction of the Jewish population with missing gender is 0.4\%. In our final raking targets, we re-normalize the probabilities over the non-missing values of age and gender to exclude the missing category

For partisanship, there is an ``unknown'' category for which the estimated fraction in the Jewish population is 0.05\%. This is smaller than many of the other rare parties present in the dataset. To align with Pew's classification, we classify the partisanship into Democrat, Republican, Independent, and Something else. Specifically, we map those listed in the voter file as Independent or Non-Partisan as Independent and all other smaller parties as something else. We again drop the missing category and then re-normalize.

Tables~\ref{tab-gender},~\ref{tab-age},~\ref{tab-party},~and~\ref{tab-region} give our estimates of $\Pr(X=x\mid R=1)$ for each $X$ variable after the re-normalization as well as the AJPP estimates in case of region. These are our raking targets when not relying on Pew targets. For comparison, the tables also include the Pew targets. These comparisons also appear in main paper figure~\ref{fig-bisgrakingweights}.

\begin{table}[h]
   \centering
     \begin{tabular}{c|cc}
    \textbf{Gender} & \textbf{Voter File Estimate} & \textbf{Pew Estimate}\\
     \hline   
       M  & 0.507 & 0.489 \\
       F  & 0.493 & 0.511\\ 
     \hline  
    \end{tabular}
     \caption{\textbf{Gender estimates for Jewish population}. Voter file estimates come from the weighted estimator in equation \ref{eq-weightedest} under the assumption that given name and geography, whether or not someone is Jewish has no relationship to gender.}    \label{tab-gender}
\end{table}

\begin{table}[h]
    \centering
    \begin{tabular}{c|cc}
    \textbf{Age} & \textbf{Voter File Estimate} & \textbf{Pew Estimate}\\
     \hline   
      18-29 & 0.138 & 0.183\\ 
      30-49 & 0.272 & 0.323\\ 
      50-64 & 0.212 & 0.202\\ 
      65+ & 0.378 & 0.292\\ 
     \hline  
    \end{tabular}
    \caption{\textbf{Age estimates for Jewish population}.  Voter file estimates come from the weighted estimator in equation \ref{eq-weightedest} under the assumption that given name and geography, whether or not someone is Jewish has no relationship to age.}
    \label{tab-age}
\end{table}

\begin{table}[h]
    \centering
    \begin{tabular}{c|cc}
    \textbf{Partisanship} & \textbf{Voter File Estimate} & \textbf{Pew Estimate}\\
     \hline   
     Democratic & 0.503 & 0.576 \\ 
     Republican & 0.249 & 0.159\\       Independent & 0.229 & 0.199\\ 
       Something else & 0.019 & 0.066\\ 
     \hline  
    \end{tabular}
    \caption{\textbf{Partisanship estimates for Jewish population.} Voter file estimates come from using the Birdie EM algorithm from \cite{mccartan2025estimating} under the assumption that given Jewish status and geography, name has no relationship to party.}
    \label{tab-party}
\end{table}

\begin{table}[h]
    \centering
    \begin{tabular}{c|cc}
    \textbf{Region} & \textbf{AJPP Estimate} & \textbf{Pew Estimate}\\
     \hline   
     Midwest & 0.109 & 0.103\\ 
     Northeast &0.411 & 0.384\\
     South & 0.266 & 0.265\\
     West & 0.213 & 0.248\\
     \hline  
    \end{tabular}
    \caption{\textbf{Region estimates for Jewish population.}}
    \label{tab-region}
\end{table}

\subsubsection{Robustness check: comparing survey weighting options}\label{app-weights}

In this section, we augment the plots in \ref{sec-pew-comparison} 
by considering the impact of different survey weighting options on the estimated mean response for a number of survey questions. As in the main paper, we compare these to the results of the 2020 Pew survey.

Figures \ref{fig-denomination} to \ref{fig-services} show the results for the following weighting options:
\begin{enumerate}
    \item \textbf{Unweighted} - simply the proportion of respondents that gave each answer to the survey question.
    \item \textbf{Raked to Pew} - raking to the Pew sample marginal distributions of age, gender, race, party affiliation, and state of residence.
    \item \textbf{IPW} - weight using untrimmed inverse propensity scores.
    \item \textbf{IPW trimmed} - weight using inverse propensity scores that have been trimmed by the given numbers. For example 10/90 means that we calculate the 10th and 90th percentile of the weights and then set any weights below and above these numbers to the 10th and 90th percentile respectively. 
    \item \textbf{Rake to Pew after trimmed IPW} - raking as above but with the algorithm initialized to the 10/90 trimmed IPW estimates as described in \cite{mercer2018weighting}.
    \item \textbf{Rake to estimated targets after trimmed IPW} - raking as in previous but now using targets estimated from voter file. 
\end{enumerate}

Overall, we see similar results to the main paper, with estimates mostly very similar to those in Pew across weighting methods. For IPW weighting, the trimming is important as the IPW without trimming has very large confidence intervals, an indication of instability from extreme weights.

\begin{figure}[h]
\centering   
\includegraphics[width=\linewidth]{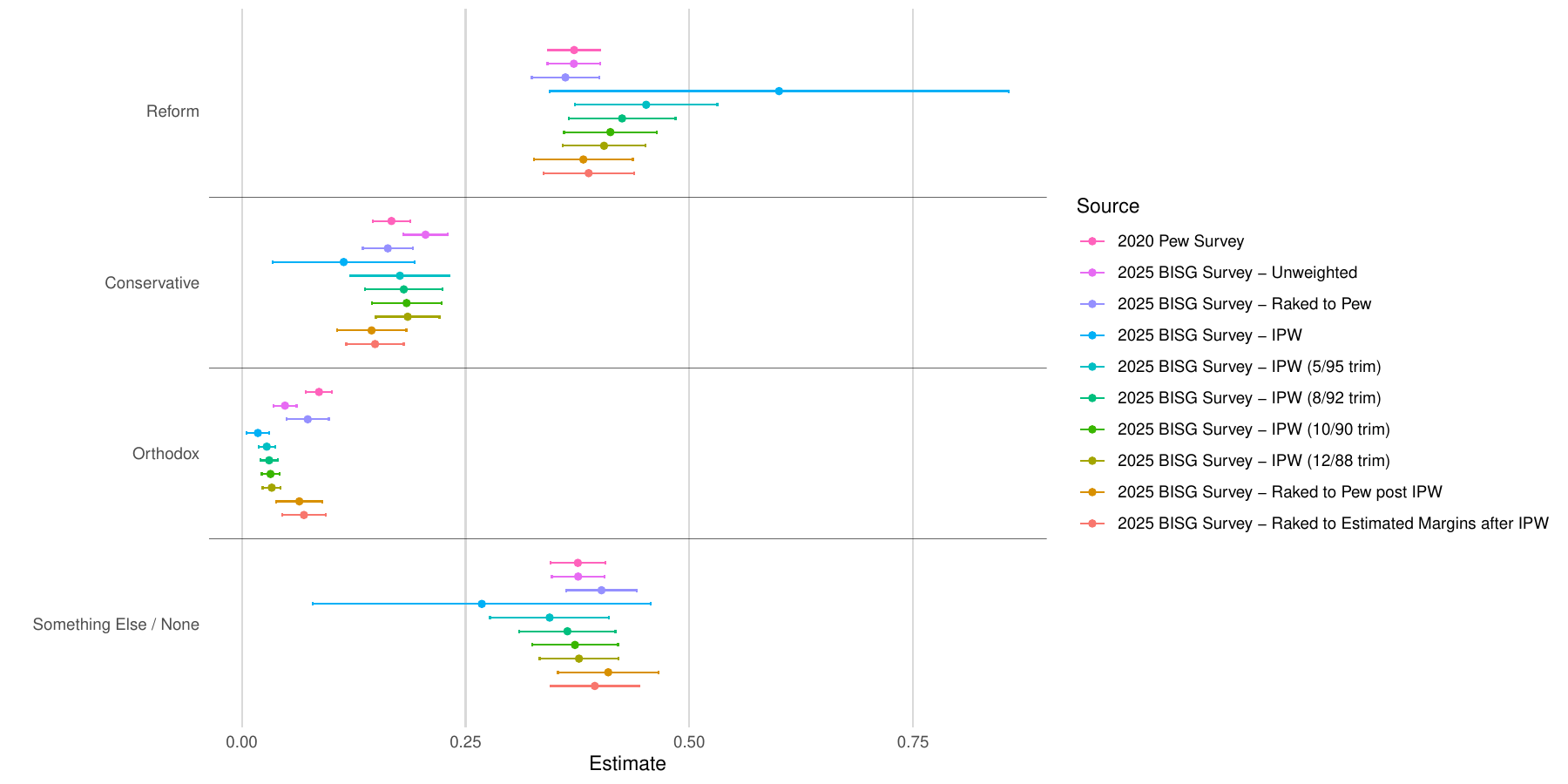}
\caption{\textbf{Survey results for denomination.} Estimates of the distribution over responses to the question: ``Thinking about Jewish religious denominations, do you consider yourself to be \ldots''}
\label{fig-denomination}
\end{figure}

\begin{figure}[h]
    \centering
    \includegraphics[width=\linewidth]{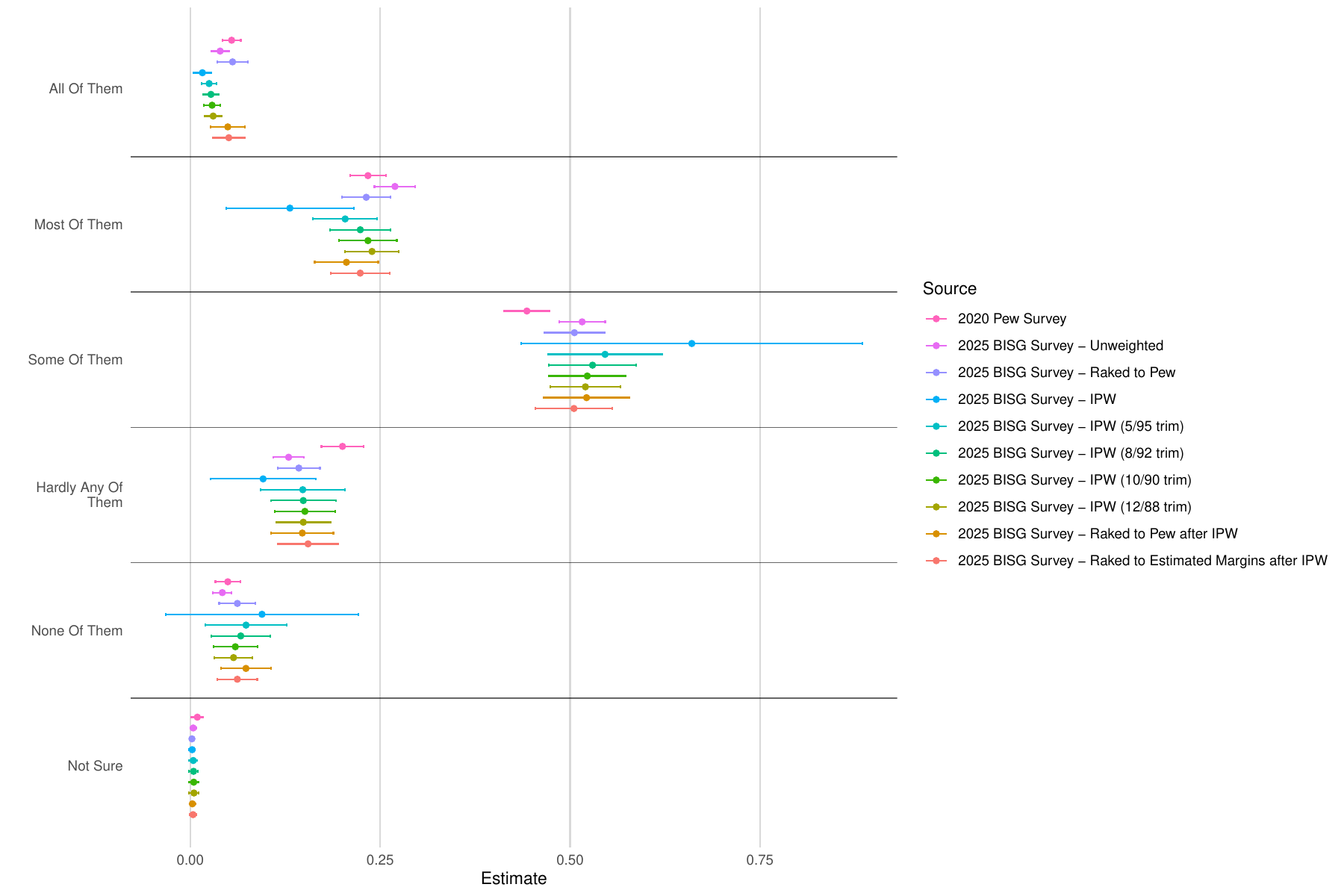}
    \caption{\textbf{Survey results for friendship network.} Estimates of the distribution over responses to the question: ``How many of your close friends are Jewish?''}
    \label{fig-friends}
\end{figure}

\begin{figure}[h]
    \centering
    \includegraphics[width=\linewidth]{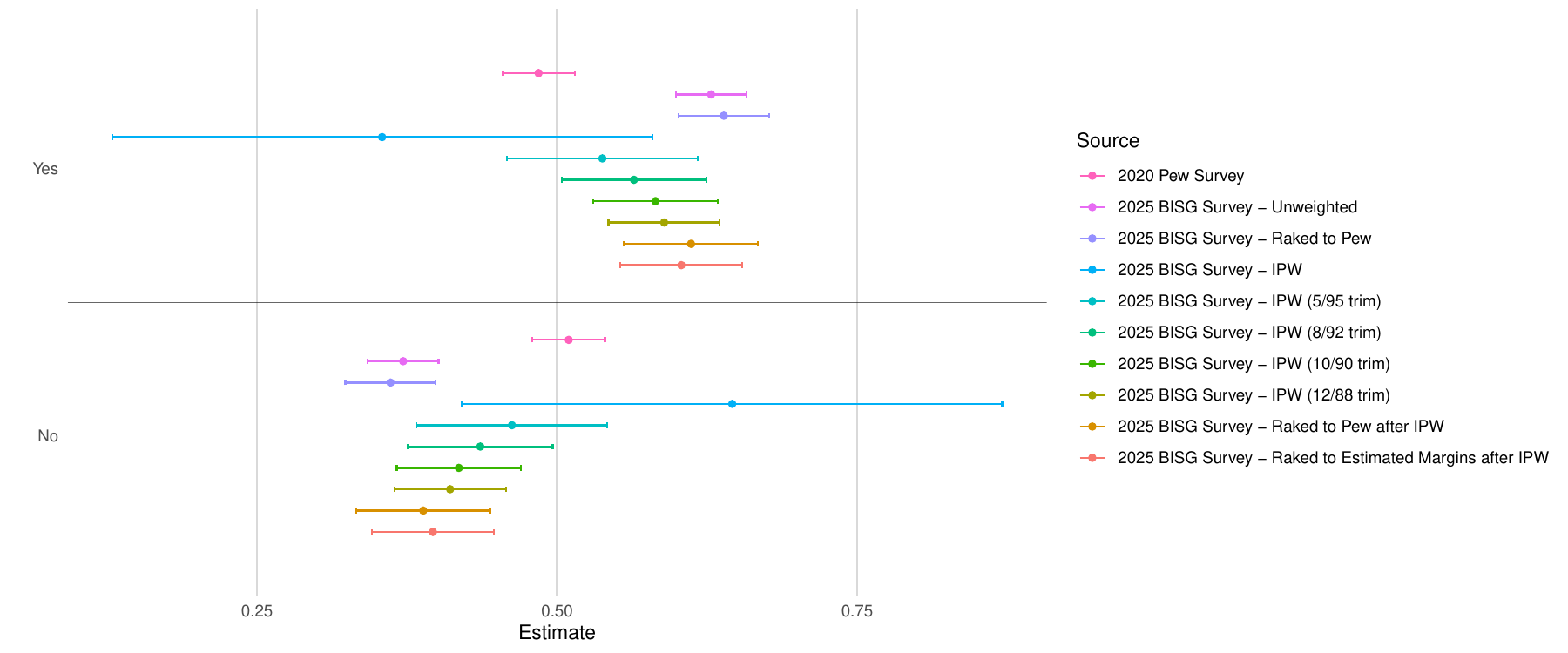}
     \caption{\textbf{Survey results for siddur ownership.} Estimates of the distribution over responses to the question: ``Did you have a bar/bat mitzvah?''}
    \label{fig-mitzvah}
\end{figure}

\begin{figure}[h]
    \centering
    \includegraphics[width=\linewidth]{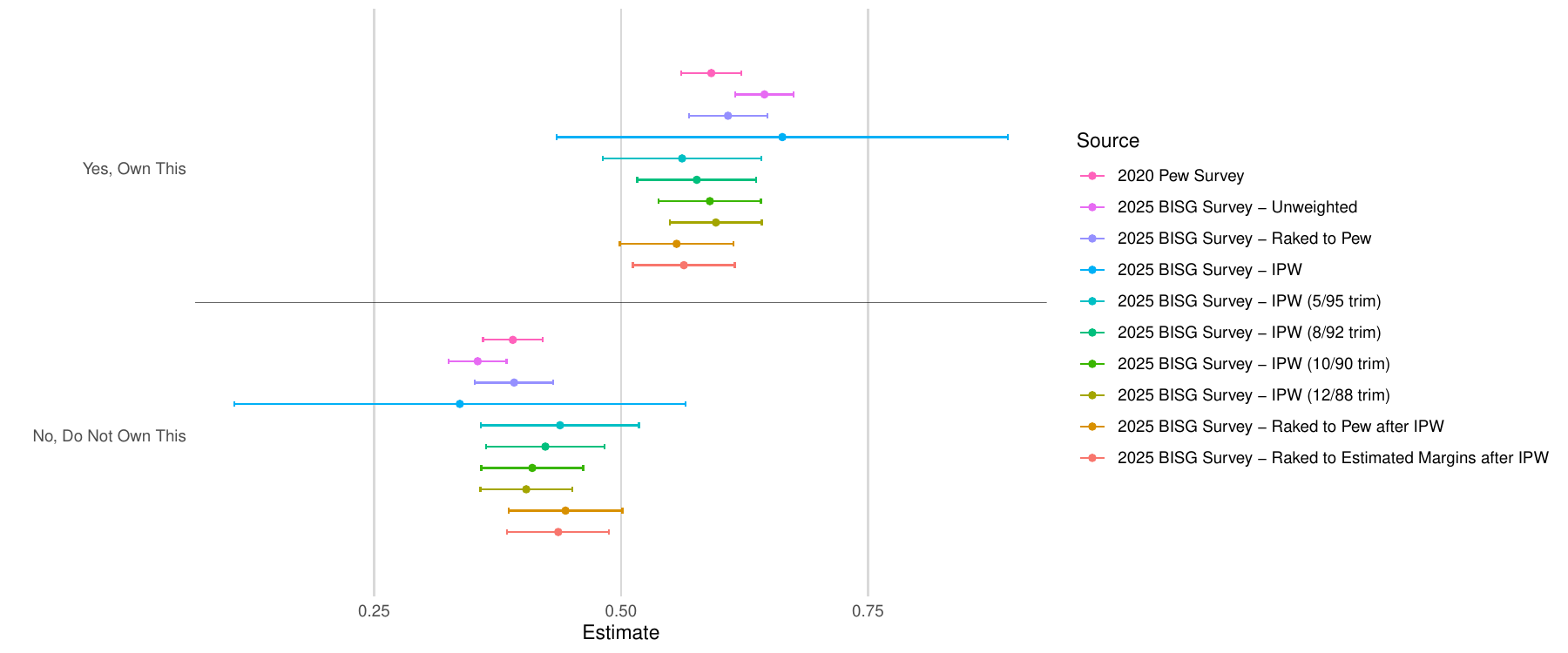}
     \caption{\textbf{Survey results for siddur ownership.} Estimates of the distribution over responses to the question: ``Do you own a siddur or prayer book?''}
    \label{fig-siddur}
\end{figure}

\begin{figure}[h]
    \centering
    \includegraphics[width=\linewidth]{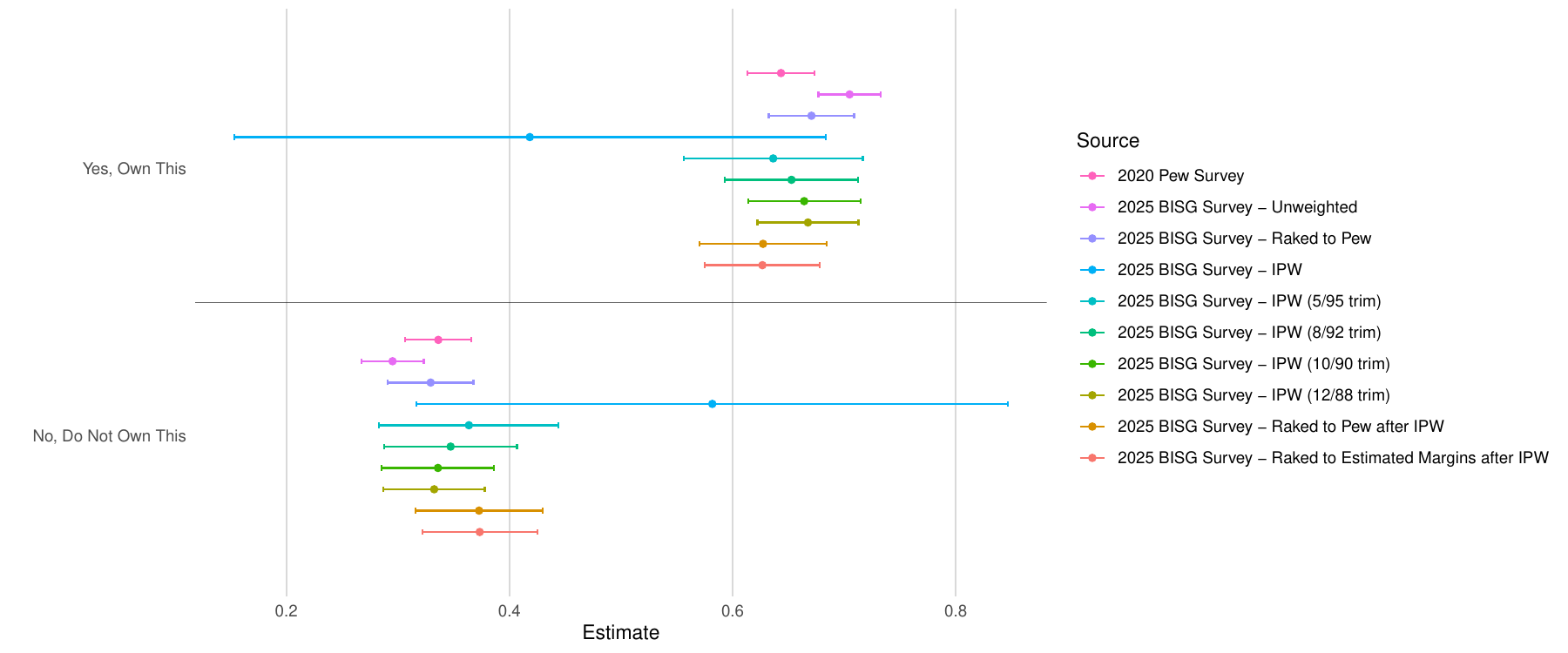}
    \caption{\textbf{Survey results for mezuzah ownership.} Estimates of the distribution over responses to the question: ``Do you own a mezuzah?''}
\label{fig-mezuzah}
\end{figure}

\begin{figure}[h]
    \centering
    \includegraphics[width=\linewidth]{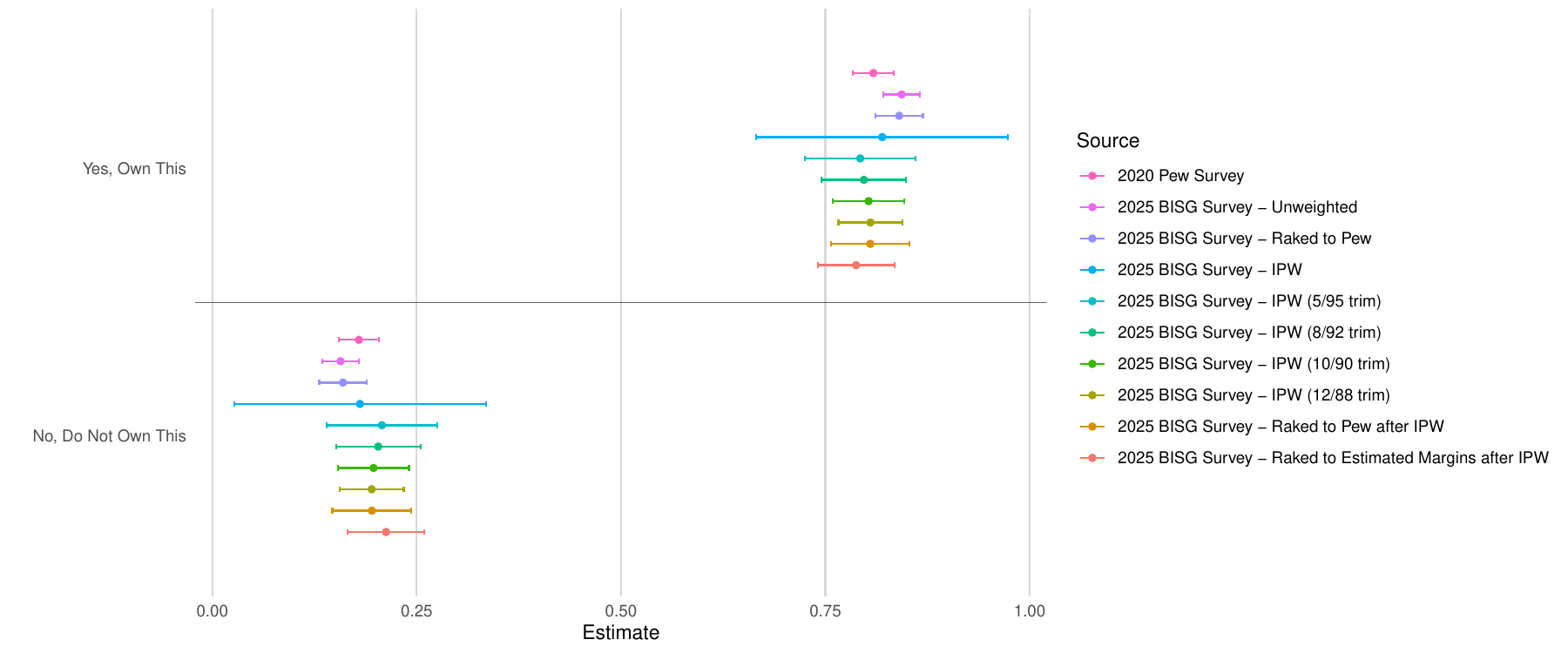}
     \caption{\textbf{Survey results for menorah ownership.} Estimates of the distribution over responses to the question: ``Do you own a menorah?''}
    \label{fig-menorah}
\end{figure}

\begin{figure}[h]
    \centering
    \includegraphics[width=\linewidth]{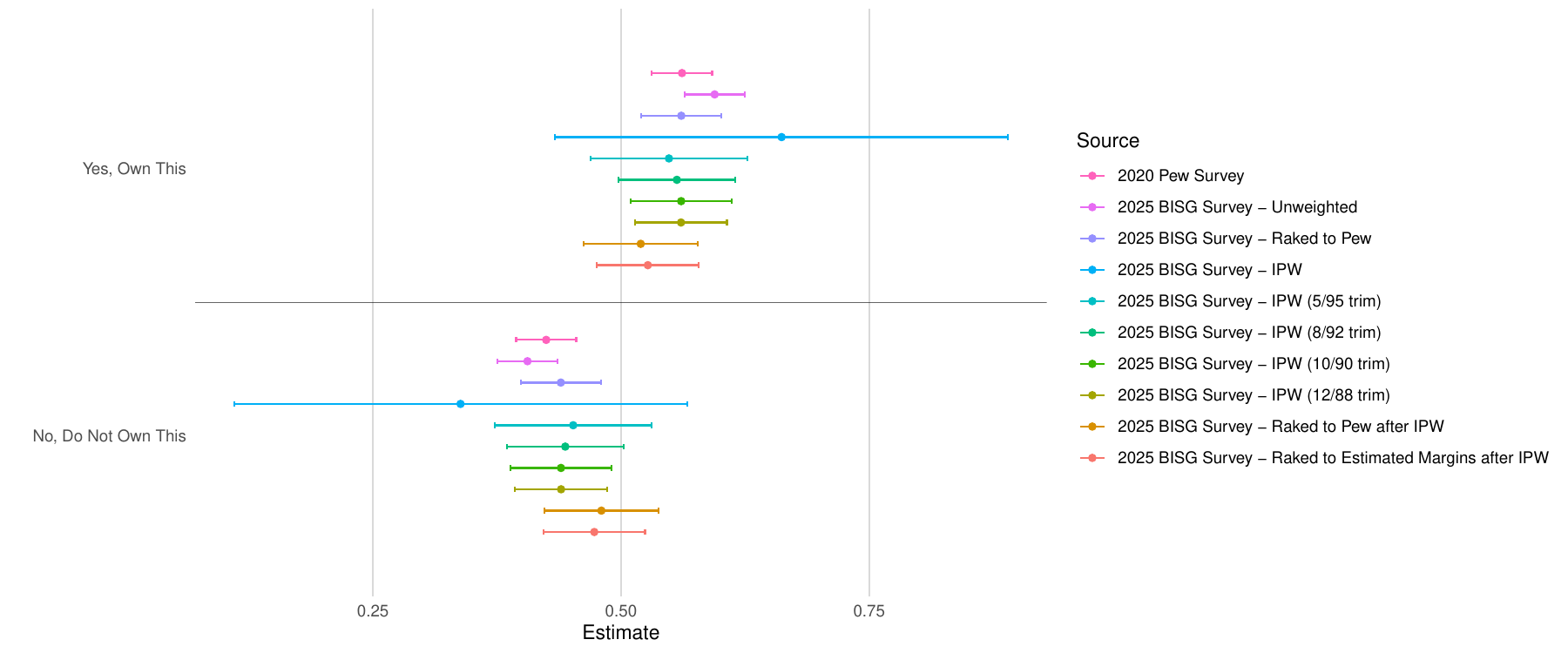}
    \caption{\textbf{Survey results for seder plate ownership.} Estimates of the distribution over responses to the question: ``Do you own a seder plate?''}
    \label{fig-seder}
\end{figure}

\begin{figure}[h]
    \centering
    \includegraphics[width=\linewidth]{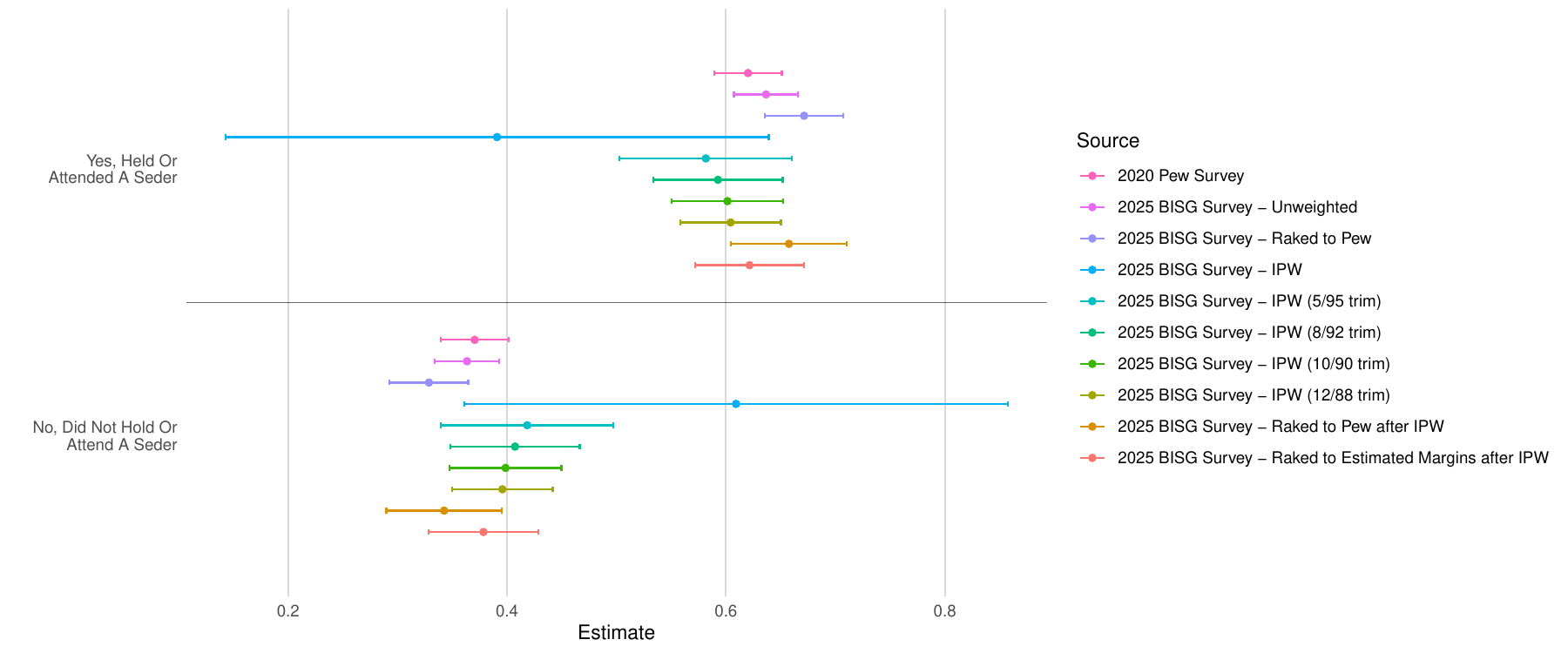}
    \label{fig-seder-attend}
    \caption{\textbf{Survey results for seder attendance.} Estimates of the distribution over responses to the question: ``Last Passover, did you hold or attend a seder?''}
\end{figure}

\begin{figure}[h]
    \centering
    \includegraphics[width=1\linewidth]{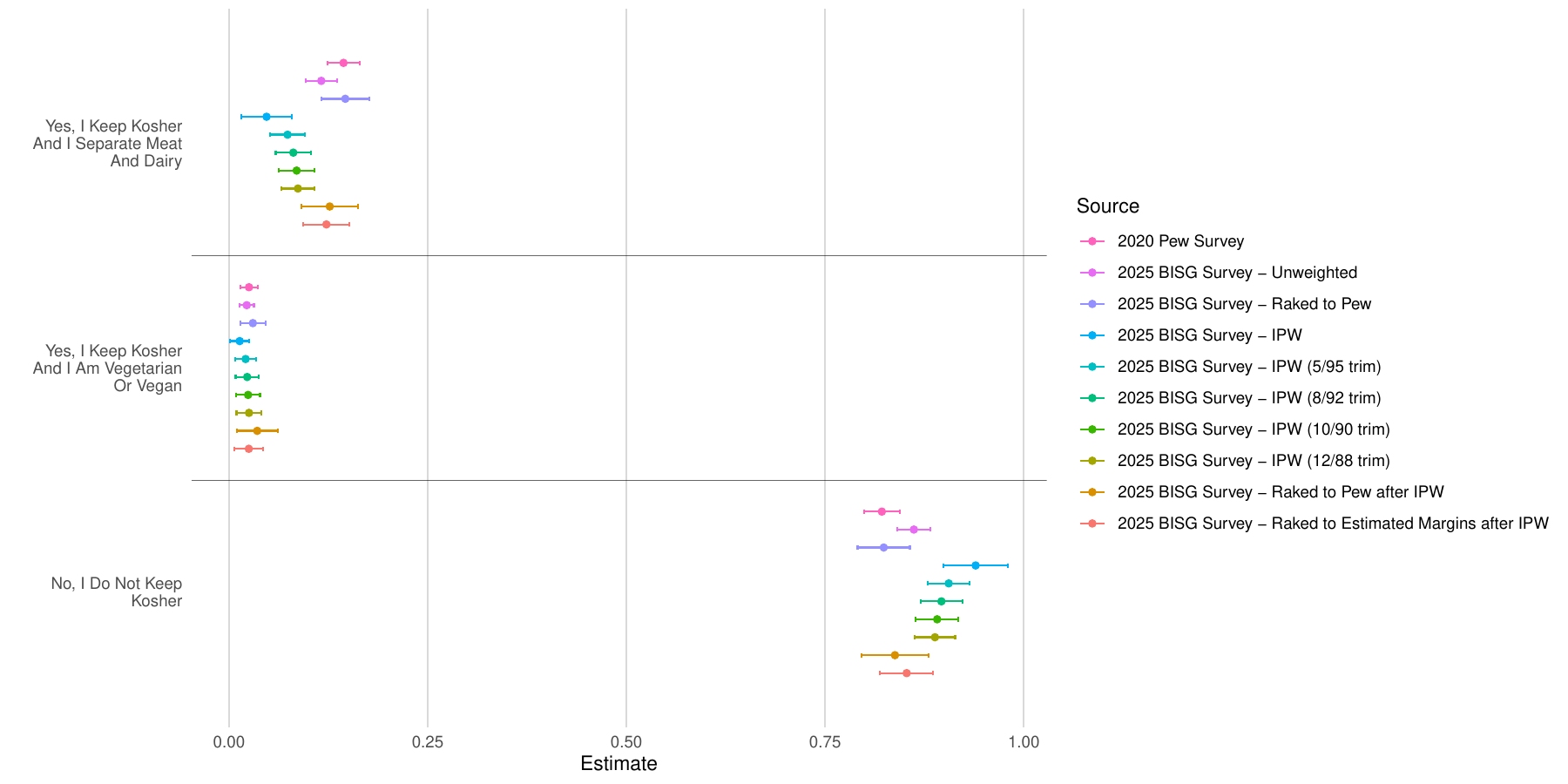}
    \caption{\textbf{Survey results for keeping kosher.} Estimates of the distribution over responses to the question: ``Do you keep kosher?''}
    \label{fig-kosher}
\end{figure}

\begin{figure}[h]
    \centering
    \includegraphics[width=\linewidth]{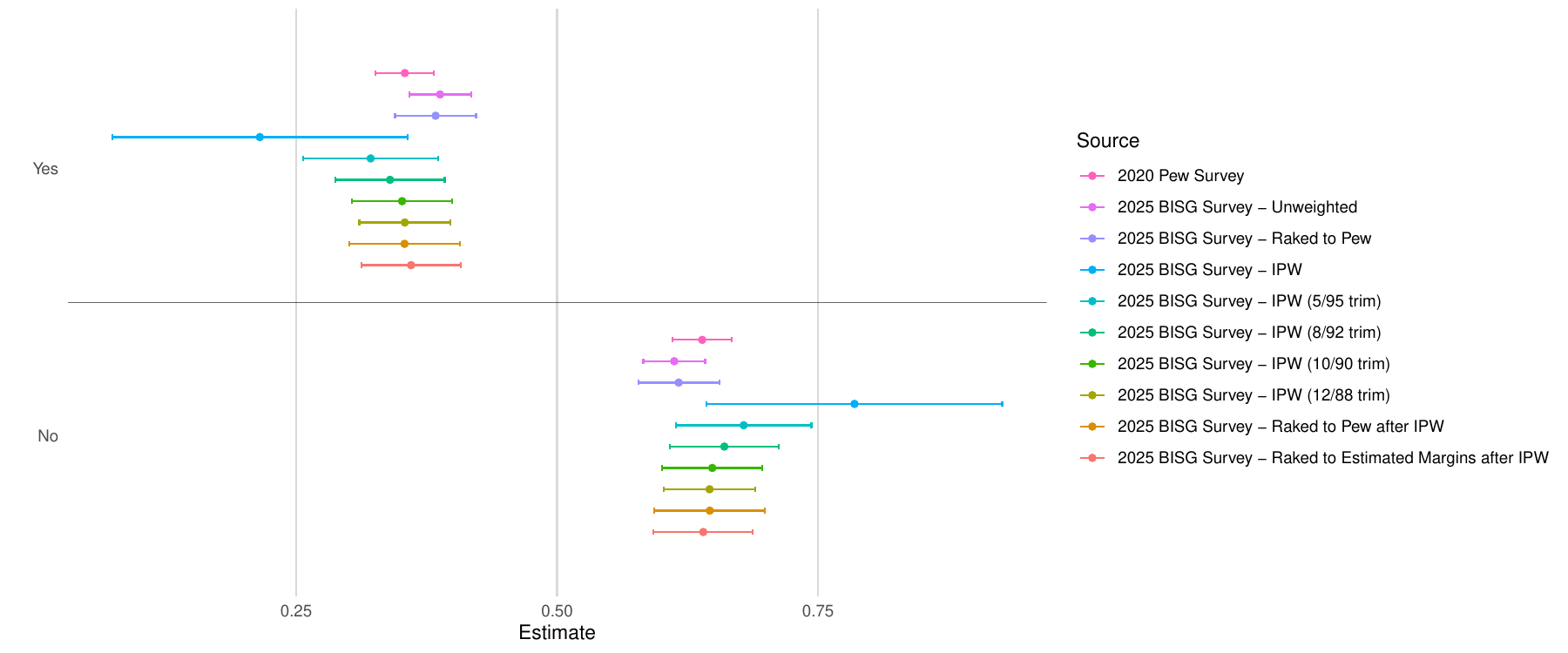}
    \caption{\textbf{Survey results for synagogue membership.} Estimates of the distribution over responses to the question: ``Is anyone in your household a member of a synagogue?''}
    \label{fig-memsyn}
\end{figure}

\begin{figure}[h]
    \centering
    \includegraphics[width=\linewidth]{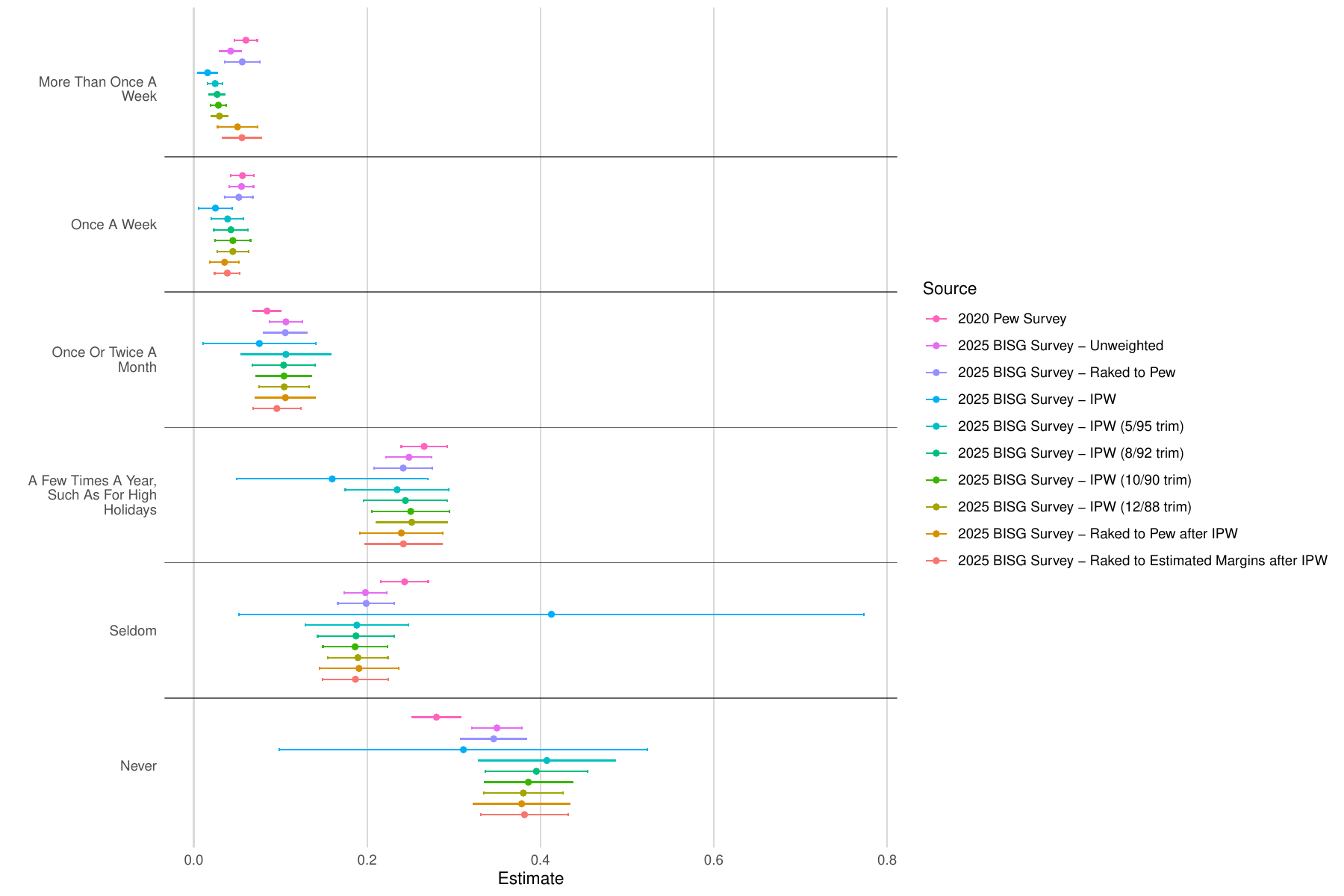}
    \caption{\textbf{Survey results for Jewish services.} Estimates of the distribution over responses to the question: ``How often do you attend Jewish services?''}
    \label{fig-services}
\end{figure}

\clearpage
\subsubsection{Comparison to CES Study}\label{app-ces}

The Cooperative Election Study (CES) is a large collaborative academic survey administered by YouGov and is designed to produce a nationally representative sample of U.S. adults for public opinion research. We compare our results to the 2025 CES Common Content survey. The survey had 17,000 respondents in total, including 404 who identified as Jewish. Additional information on the survey design, sampling methodology, weighting procedures, and other details can be found on the CES website \url{https://tischcollege.tufts.edu/research-faculty/research-centers/cooperative-election-study}.

\begin{figure}[h]
    \centering
\includegraphics[width=\linewidth]{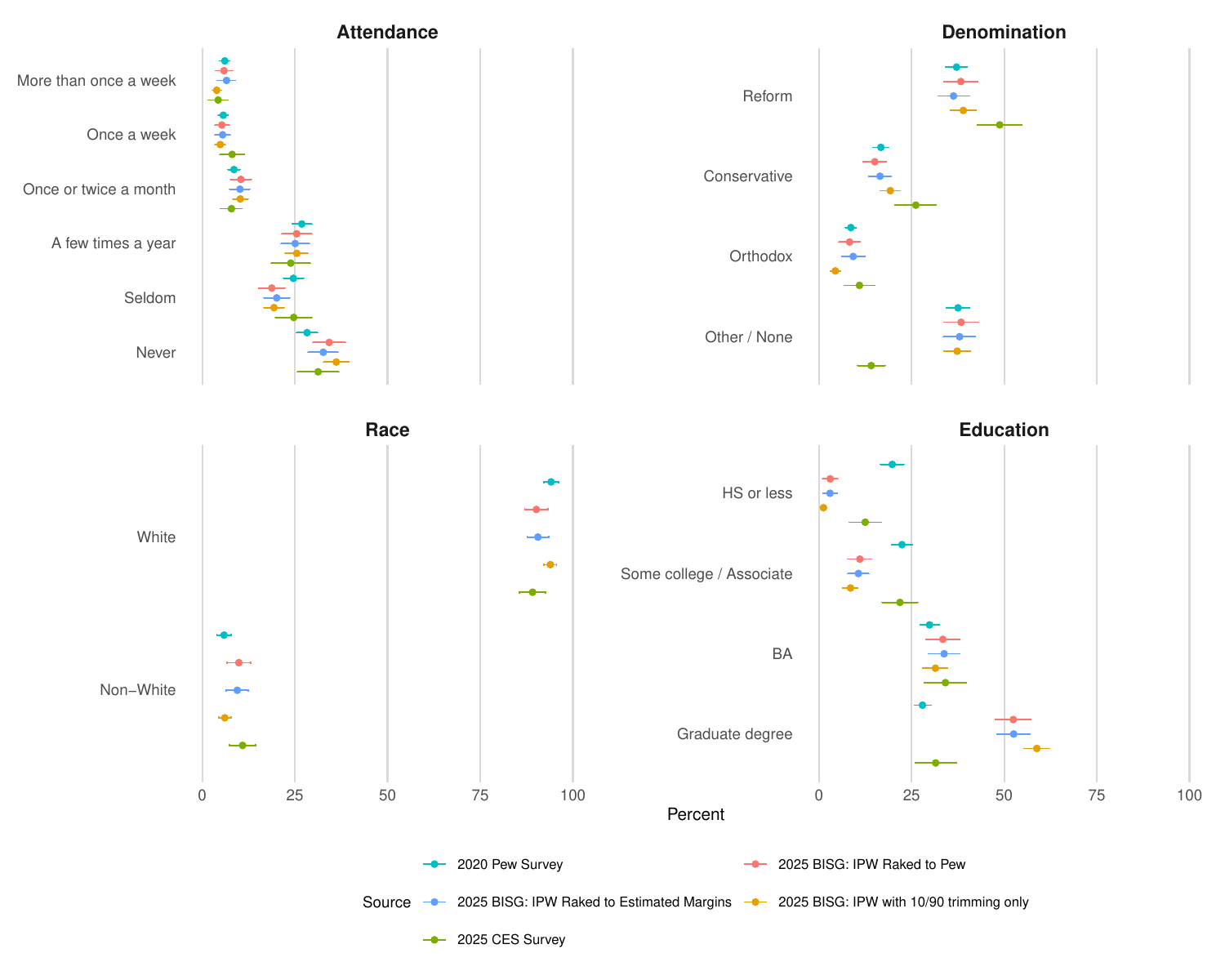}
    \label{fig-CES}
    \caption{\textbf{Comparison to CES.} Comparison of the three main paper estimators from our survey as well as Pew to estimates based on the Jewish Respondents of the Cooperative Election Study}
\end{figure}

\clearpage 


\end{document}